\documentclass[11pt]{article}
\usepackage{jheppub}
\pdfoutput=1
\usepackage{amsmath,bbm,array,amsfonts,graphicx,wrapfig,lscape,float,mathtools,multirow,longtable,amsthm}
\usepackage[dvipsnames]{xcolor}
\usepackage{stackrel}
\usepackage[all]{xy}
\usepackage{caption}
\usepackage{subcaption}
\usepackage[bottom]{footmisc}
\usepackage{mathrsfs}
\usepackage{tikz}
\usepackage{bm}
\usepackage{color}
\usepackage{enumerate}
\usetikzlibrary{decorations.pathreplacing}
\usepackage{diagbox}
\numberwithin{equation}{section}
\numberwithin{figure}{section}
\numberwithin{table}{section}
\usepackage{pgfplots}
\pgfplotsset{compat=1.14}
\usepackage{makecell}
\usepackage{lscape}
\usepackage[utf8]{inputenc}
\usepackage{mathtools}
\allowdisplaybreaks
\usepackage{upgreek}
\usepackage[dvipsnames]{xcolor}

\usepackage[autostyle=true]{csquotes}
\newtheorem{definition}{Definition}[section]
\newtheorem{theorem}{Theorem}[section]
\newtheorem{corollary}{Corollary}[theorem]
\newtheorem{lemma}[theorem]{Lemma}
\newtheorem{proposition}[theorem]{Proposition}
 
\newtheorem{remark}{Remark}

\usepackage[toc,page]{appendix}

\usetikzlibrary{external}
\tikzset{external/system call={pdflatex \tikzexternalcheckshellescape 
		-halt-on-error
		-interaction=batchmode 
		-jobname "\image" "\texsource"
		&& pdftops -eps "\image.pdf"}}
\usepackage{tocloft}

	\title{Quiver Yangians and $\mathcal{W}$-Algebras for Generalized Conifolds}
	
	\author[a,b]{Jiakang Bao}
	
	\affiliation[a]{
		Department of Mathematics, City, University of London, EC1V 0HB, UK}
	\affiliation[b]{
		London Institute for Mathematical Sciences, Royal Institution, London W1S 4BS, UK}
	
	\emailAdd{jiakang.bao@city.ac.uk}
	
	\preprint{
		\begin{flushright}
			
		\end{flushright}
	}

	\abstract{We focus on quiver Yangians for most generalized conifolds. We construct a coproduct of the quiver Yangian following the similar approach by Guay-Nakajima-Wendlandt. We also prove that the quiver Yangians related by Seiberg duality are indeed isomorphic. Then we discuss their connections to $\mathcal{W}$-algebras analogous to the study by Ueda. In particular, the universal enveloping algebras of the $\mathcal{W}$-algebras are truncations of the quiver Yangians, and therefore they naturally have truncated crystals as their representations.
	}

\begin{document}
	\maketitle

\section{Introduction and Summary}\label{intro}
The quiver Yangians were first introduced in \cite{Li:2020rij} as BPS algebras \cite{Harvey:1996gc} for Type IIA string theories on toric Calabi-Yau (CY) threefolds. Since then, the quiver Yangians and their cousins have been extensively studied in \cite{Galakhov:2020vyb,Galakhov:2021xum,Noshita:2021ldl,Galakhov:2021vbo,Noshita:2021dgj,Bao:2022fpk,Galakhov:2022uyu}. As a nice combinatoric way to encode the BPS spectrum, the crystal melting models \cite{Okounkov:2003sp,Iqbal:2003ds,Ooguri:2009ijd} naturally provide representations for the algebras. It is also expected that the quiver Yangians have intimate relations with cohomological Hall algebras \cite{Joyce:2008pc,Kontsevich:2008fj,Kontsevich:2010px,Rapcak:2018nsl,Rapcak:2020ueh} and many other quantum algebras.

An important algebra structure of the quiver Yangian is its coproduct. As studied in \cite{Litvinov:2020zeq,Litvinov:2021phc,Chistyakova:2021yyd,Kolyaskin:2022tqi,Bao:2022fpk,Galakhov:2022uyu}, one can construct the $\mathcal{R}$-matrices from the BPS algebras. This realizes the Bethe/gauge correspondence \cite{Nekrasov:2009uh,Nekrasov:2009ui} for supersymmetric gauge theories with (non-chiral) toric quiver descriptions. For instance, the rapidities in the Bethe ansatz equation are known to be identified with the supersymmetric vacua of the 2d $\mathcal{N}=(2,2)$ gauge theory in this dictionary. Moreover, each 2-magnon $S$-matrix corresponds to the so-called bond factor (see \eqref{bondfactor}) stemmed from the action of the quiver Yangian generators on the 2d molten crystal \cite{Bao:2022fpk}. By analyzing the coproduct structures of the algebra generators and the Lax operators, a remarkable no-go theorem for chiral quivers was discovered for the Bethe/gauge correspondence in \cite{Galakhov:2022uyu}. In other words, it would require a more delicate and involved study if we want to extend such correspondence for shifted 2d crystals or the quivers associated to toric CY threefolds with compact divisors.

Here, we shall mainly focus on generalized conifolds $xy=z^Mw^N$ with $M+N>2$, $MN\neq2$ and $M\neq N$. All but one toric CY threefolds without compact 4-cycles are generalized conifolds. Their quiver Yangians have salient features and have been systematically studied in \cite{Li:2020rij}. As the restrictions on $M,N$ suggest, the construction of the coproduct would greatly benefit from the underlying untwisted affine Lie superalgebra $\widehat{\mathfrak{sl}}_{M|N}$. To obtain the coproduct, we need to consider different presentations of the quiver Yangians. In fact, our approach is analogous to the non-super affine case in \cite{guay2018coproduct}. It is also worth noting that a super version of the affine Yangian was introduced in \cite{ueda2019affine} with a coproduct given using the similar method. However, as pointed out in \cite{Bao:2022fpk}, this super Yangian is different from the quiver Yangian (unless these two-parameter algebras degenerate to one-parameter families). Physically, this is because the Yangian in \cite{ueda2019affine} does not respect the gauge symmetries associated to the vertices in the quiver. Therefore, we would still provide detailed steps when deriving the coproduct of the quiver Yangian. Nevertheless, it is not surprising to find the expressions similar to the forms in \cite{guay2018coproduct,ueda2019affine}.

For affine Lie algebras and quiver Yangians with only bosonic generators, we only have even reflections, and there is one single toric phase for each quiver Yangian. For affine Lie superalgebras and quiver Yangians with fermionic generators, we would further have odd reflections \cite{serganova1985automorphisms,hoyt2007classification,serganova2011kac} that do not respect the $\mathbb{Z}_2$-grading of the algebras. However, it is exactly this property of the odd reflections that allows us to transform the quiver from one toric phase to another. Often in literature, only the distinguished case (which has two fermionic nodes) would be discussed. It is always natural to expect the results obtained would still hold for those with more fermionic nodes due to the invariance of the underlying Kac-Moody algebra. Here, along with the help of a presentation obtained in studying the coproduct, we shall extend the odd reflections to the quiver Yangians and prove that the Seiberg dual quiver Yangians are indeed isomorphic algebras.

On the other hand, the $\mathcal{W}$-algebras \cite{Zamolodchikov:1985wn,Prochazka:2017qum,Prochazka:2018tlo,Creutzig:2018pts,Creutzig:2019qos,Rapcak:2019wzw,Eberhardt:2019xmf} should play a crucial role in the tensionless limit of string theory in AdS$_3$ \cite{Gross:1988ue,Henneaux:2010xg,Campoleoni:2010zq,Gaberdiel:2014cha}. In particular, the rectangular $\mathcal{W}$-algebra can be realized as the symmetry algebra of the coset CFT whose holographic dual gives higher spin gravity \cite{Creutzig:2013tja}. Such vertex algebras have been well-studied in mathematics literature such as \cite{kac2003quantum,Kac:2003jh,Kac2005CorrigendumT}.

In this paper, we shall discuss the BPS/CFT (aka AGT, 2d/4d) correspondence \cite{Alday:2009aq,Wyllard:2009hg} at the level of algebras. The BPS algebras and the vertex operator algebras (VOAs) are expected to be contained in a broader picture under the BPS/CFT correspondence. In the finite cases, the relations between Yangians and $\mathcal{W}$-algebras have been explored in \cite{Briot:2000fn,Brundan:2004ca,briot2001yangians}. For $\widehat{\mathfrak{gl}}_1$ whose associated CY is the simplest $\mathbb{C}^3$, it was shown in \cite{Gaberdiel:2017dbk,Prochazka:2019dvu} that the affine/quiver Yangian is isomorphic to the universal enveloping algebra of the $\mathcal{W}_{1+\infty}$-algebra. Moreover, in such case, the AGT conjecture was proven in \cite{schiffmann2013cherednik} with a surjective homomorphism from the quiver Yangian for $\mathbb{C}^3$ to the universal enveloping algebra of the principal $\mathcal{W}$-algebra. Physically, the Nekrasov partition function of the 4d supersymmetric gauge theories can be identified with the conformal blocks of the corresponding VOAs. From a geometric perspective, the Verma module of the VOA results from its action on the equivariant intersection cohomology of the instanton moduli space \cite{Braverman:2014xca}. Later in \cite{Prochazka:2017qum,Gaberdiel:2017hcn,Gaberdiel:2018nbs,Li:2019nna,Li:2019lgd}, a similar study was extended to the supersymmetric case with the $\mathcal{N}=2$ $\mathcal{W}$-algebra under the construction of gluing two trivalent vertices/plane partitions.

Similar to the $\mathcal{W}_{1+\infty}$-algebra case whose truncation gives the algebra at the corner \cite{Gaiotto:2017euk}, the matrix-extened $\mathcal{W}$-algebra (called $\mathcal{W}_{M|N\times\infty}$ in \cite{Rapcak:2019wzw}) for any generalized conifold has generators of spins $s=1,2,\dots$ and is expected to truncate to VOAs describing certain interface of a 4d supersymmetric gauge theory. In \cite{ueda2022affine}, a surjective homomorphism from Ueda's affine super Yangian to the rectangular $\mathcal{W}$-algebras was constructed in a way similar to the evaluation maps in Yangian algebras \cite{kodera2021guay,ueda2019affine}. However, as mentioned above, such super Yangian is not the BPS algebra physically. Therefore, due to this difference as well as the different conventions, we decide to give a careful construction of the surjective homomorphism from the quiver Yangian to the $\mathcal{W}$-algebras (under certain condition on the parameters of the quiver Yangian) although the approach is still similar to those in \cite{ueda2022affine,kodera2021guay}. This shows that the universal enveloping algebras of the VOAs are indeed the truncations of the BPS algebras, and the quiver Yangians can in this sense be viewed as some realization of $U(\mathcal{W}_{M|N\times\infty})$. As the (truncated) crystals are then naturally the representations of the universal enveloping algebras of the $\mathcal{W}$-algebras, we show that they are highest weight representations. We will also see that the parameters on the $\mathcal{W}$-algebra side can be related to the vacuum charge on the quiver Yangian side. Let us also mention here that similar to the quiver Yangian discussions, we shall also consider the $\mathcal{W}$-algebras not only restricted to the distinguished case. However, unlike the proof for the quiver Yangians, the isomorphisms for these $\mathcal{W}$-algebras are much more straightforward. This allows us to connect both the BPS algebras and the VOAs in different phases freely.

The paper is organized as follows. In \S\ref{QY}, we give a brief review on the concept of quiver Yangians for generalized conifolds. In \S\ref{coprod}, we discuss different presentations of the quiver Yangians which enable us to obtain the coproducts of them. We shall then prove that the quiver Yangians in different toric phases are isomorphic in \S\ref{isoQY}. In \S\ref{YandW}, we shall study the relations between quiver Yangians and rectangular $\mathcal{W}$-algebras. Some future directions will be mentioned in \S\ref{outlook}. In Appendix \ref{generators2}, we shall also discuss the generators of the quiver Yangians for $\mathbb{C}\times\mathbb{C}^2/\mathbb{Z}_2$ and the conifold, as well as $\mathbb{C}^3/(\mathbb{Z}_2\times\mathbb{Z}_2)$ which is the only toric CY$_3$ without compact 4-cycles that is not a generalized conifold. We discuss the odd reflections of the Kac-Moody algebras in Appendix \ref{oddrefchevgen}. In Appendix \ref{recW}, we give an introduction to the rectangular $\mathcal{W}$-algebras.

\section{Quiver Yangians for Generalized Conifolds}\label{QY}
Let us start with a recap on some fundamentals of the quiver Yangians. Here, we shall only give the defining relations for generalized conifolds $xy=z^Mw^N$ ($M+N>2$), but the quiver Yangians can be defined for any quivers as in \cite{Li:2020rij}. We shall fix the convention that the set of natural numbers $\mathbb{N}$ is always counted from 0.

Given any generalized conifold with $M+N\geq1$, its toric diagram is
\begin{equation}
	\includegraphics[width=5cm]{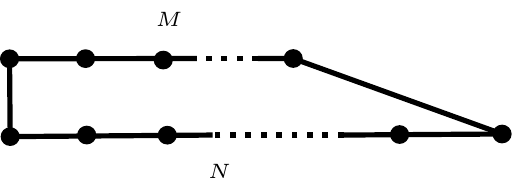}.
\end{equation}
Its quivers in different toric phases can be obtained from the triangulations of the lattice polygon \cite{nagao2008derived,Nagao:2009rq}. These triangulations can in turn be encoded by a sequence of signs $\varsigma=\{\varsigma_a|a\in\mathbb{Z}/(M+N)\mathbb{Z}\}$, one for each simplex in the toric diagram. There are $M$ minus ones and $N$ plus ones. When two simplices are glued side by side, they have the same sign. When they are glued in the alternative way, they have opposite signs. An illustration can be found in Figure \ref{sigmaex}.
\begin{figure}[h]
	\centering
	\includegraphics[width=6cm]{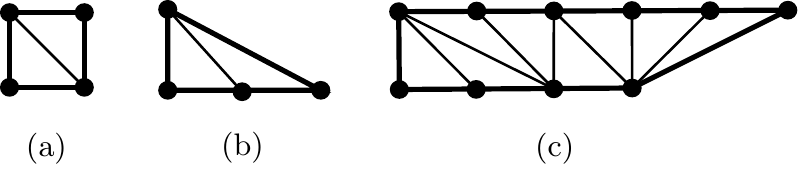}
	\caption{In these examples, we have (a) $\varsigma=\{-1,+1\}$, (b) $\varsigma=\{-1,-1\}$ and (c) $\varsigma=\{-1,-1,+1,+1,-1,+1,+1,+1\}$.}\label{sigmaex}
\end{figure}

The quiver is constructed as follows. First, there is always a pair of opposite arrows connecting node $a$ and node $a+1$ ($a\in\mathbb{Z}/(M+N)\mathbb{Z}$). Then the node $a$ is bosonic/even and has a self-loop if $\varsigma_a=\varsigma_{a+1}$. If $\varsigma_a=-\varsigma_{a+1}$, then it is fermionic/odd and has no self-loops. Hence, the resulting quiver is essentially the double of the untwisted affine $\mathfrak{sl}_{M|N}$ Dynkin quiver with extra loops on the bosonic nodes. The superpotential $W$ can be fully determined in the toric quiver gauge theory and is composed of the terms
\begin{equation}
	\begin{cases}
		\varsigma_a\text{tr}(I_{a,a}I_{a,a-1}I_{a-1,a}-I_{a,a}I_{a,a+1}I_{a+1,a}),&\varsigma_a=\varsigma_{a+1},\\
		\varsigma_a\text{tr}(I_{a,a+1}I_{a+1,a}I_{a,a-1}I_{a-1,a}),&\varsigma_a=-\varsigma_{a+1},
	\end{cases}
\end{equation}
where $I_{a,b}$ denotes the arrow/field from node $a$ to $b$.

To construct the quiver Yangian, we need to assign parameters $\widetilde{\epsilon}_I$ to the arrows $I$ in the quiver. As $\widetilde{\epsilon}_I$ can be viewed as charges under a global symmetry of the quiver quantum mechanics, they should be compatible with the superpotential. This yields the loop constraint
\begin{equation}
	\sum_{I\in L}\widetilde{\epsilon}_I=0,
\end{equation}
for any closed loop $L$ in the periodic quiver. It turns out that the number of independent parameters is given by
\begin{equation}
	|Q_1|-|Q_2|-1=|Q_0|+1,
\end{equation}
where $Q_0$, $Q_1$ are the sets of nodes and arrows respectively. Moreover, $Q_2$ denotes the faces of the periodic quiver, or equivalently, the monomial terms in the superpotential.

Furthermore, as pointed out in \cite{Li:2020rij}, the mixing of global and gauge symmetries associated to each node would cause shifts of $\widetilde{\epsilon}_I$. One can then introduce a gauge fixing condition to get rid of this shift. This is known as the vertex constraint:
\begin{equation}
	\sum_{I\in a}\text{sgn}_a(I)\widetilde{\epsilon}_I=0,
\end{equation}
where the sign function $\text{sgn}_a(I)$ is equal to $+1$ (resp. $-1$) when the arrow $I$ starts from (resp. ends at) the node $a$, and $0$ otherwise. As an overall U(1) symmetry decouples, the total number of the vertex constraints is $|Q_0|-1$. Together with the $|Q_0|+1$ loop constraints, we are then left with two independent parameters denoted as $\epsilon_{1,2}$. Along with the R-symmetry, they parametrize the $\text{U}(1)^3$ isometry of the toric CY$_3$. Sometimes, it would also be convenient to introduce a third parameter $\epsilon_3$ such that $\epsilon_1+\epsilon_2+\epsilon_3=0$. This condition then specializes to the sub-torus $T^2$ that preserves the volume form.

Therefore, the quiver Yangian constructed therefrom is a two-parameter algebra. The general rule of the parameter assignment to the arrows is summarized in Figure \ref{epsilongencon}.
\begin{figure}[h]
	\centering
	\includegraphics[width=12cm]{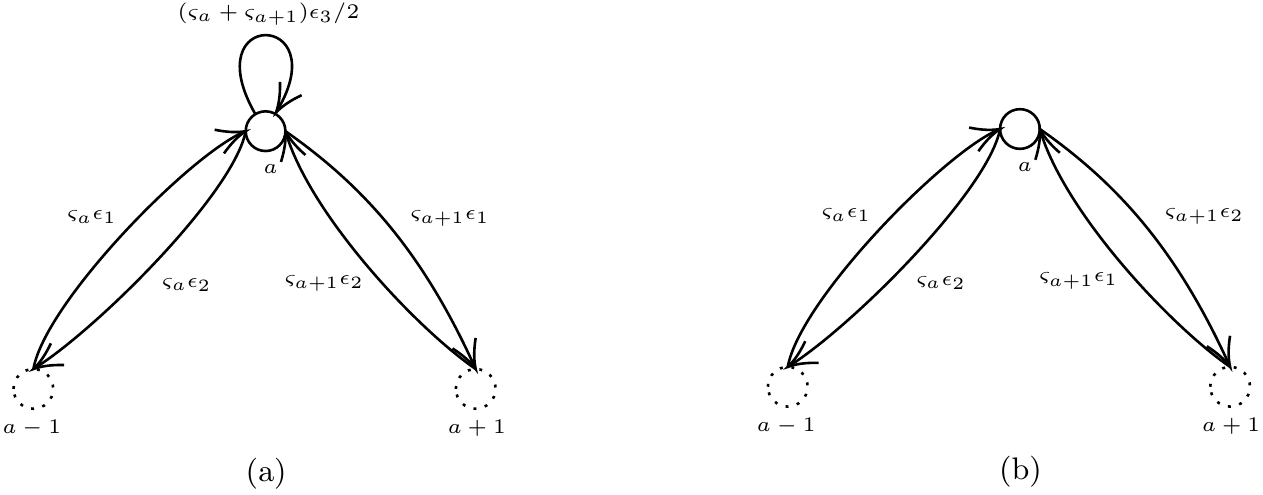}
	\caption{The values of $\widetilde{\epsilon}_I$ associated to bifundamentals and adjoints subject to the loop and vertex constraints for generalized conifolds. We have (a) $\varsigma_a=\varsigma_{a+1}$ and (b) $\varsigma_a=-\varsigma_{a+1}$, where $\epsilon_{1,2,3}$ are parameters of the quiver Yangian.}\label{epsilongencon}
\end{figure}

To write down the definition of the quiver Yangian, we need to set up some notations. In physics literature, the super bracket is often denoted as $[\text{-},\text{-}\}$. However, for convenience, we shall simply use $[x,y]=xy-(-1)^{|x||y|}yx$ here, where $|x|\in\{0,1\}$ indicates the Bose-Fermi statistics/$\mathbb{Z}_2$-grading of $x$. For any two nodes $a,b\in Q_0$, write $\{a\rightarrow b\}$ as the set of arrows from $a$ to $b$. The $k^\text{th}$ symmetric sum of $\widetilde{\epsilon}_I$ for all $I\in\{a\rightarrow b\}$ is denoted as $\sigma^{ab}_k$. As a result, $\sigma^{ab}_1$ and $\sigma^{ba}_1$ are the only non-trivial parameters when $M+N>2$ if there is an arrow from $a$ to $b$. Otherwise, $\sigma^{ab}_1=\sigma^{ba}_1=0$.

\begin{definition}\label{Ydef}
	Given a quiver $Q=(Q_0,Q_1)$ and its superpotential $W$ (with $M+N>2$), the \emph{non-reduced} quiver Yangian is generated by the modes $e^{(a)}_n$, $f^{(a)}_n$ and $\psi^{(a)}_n$ ($a\in Q_0$, $n\in\mathbb{N}$) satisfying the relations
	\begin{align}
		&\left[\psi^{(a)}_n,\psi^{(b)}_m\right]=0,\label{psipsi}\\
		&\left[e^{(a)}_n,f^{(b)}_m\right]=\delta_{ab}\psi^{(a)}_{m+n},\label{ef}\\
		&\left[\psi^{(a)}_{n+1},e^{(b)}_m\right]-\left[\psi^{(a)}_n,e^{(b)}_{m+1}\right]=\sigma^{ba}_1\psi^{(a)}_ne^{(b)}_m+\sigma^{ab}_1e^{(b)}_m\psi^{(a)}_n,\label{psie}\\
		&\left[\psi^{(a)}_{n+1},f^{(b)}_m\right]-\left[\psi^{(a)}_n,f^{(b)}_{m+1}\right]=-\sigma^{ab}_1\psi^{(a)}_nf^{(b)}_m-\sigma^{ba}_1f^{(b)}_m\psi^{(a)}_n,\label{psif}\\
		&\left[e^{(a)}_{n},e^{(b)}_m\right]=\left[f^{(a)}_{n},f^{(b)}_m\right]=0\qquad(\sigma_1^{ab}=0),\label{eeff}\\
		&\left[e^{(a)}_{n+1},e^{(b)}_m\right]-\left[e^{(a)}_n,e^{(b)}_{m+1}\right]=\sigma^{ba}_1e^{(a)}_ne^{(b)}_m+(-1)^{|a||b|}\sigma^{ab}_1e^{(b)}_me^{(a)}_n\qquad(\sigma_1^{ab}\neq0),\label{ee}\\
		&\left[f^{(a)}_{n+1},f^{(b)}_m\right]-\left[f^{(a)}_n,f^{(b)}_{m+1}\right]=-\sigma^{ab}_1f^{(a)}_nf^{(b)}_m-(-1)^{|a||b|}\sigma^{ba}_1f^{(b)}_mf^{(a)}_n\qquad(\sigma_1^{ab}\neq0).\label{ff}
	\end{align}
	The generators $e^{(a)}_n$ and $f^{(a)}_n$ have the $\mathbb{Z}_2$-grading same as the corresponding node $a$ while $\psi^{(a)}_n$ is always bosonic. In the above relations, we also allow $\psi^{(a)}_{-1}:=1/(\epsilon_1+\epsilon_2)$ so that
	\begin{equation}
		\left[\psi^{(a)}_0,e^{(b)}_m\right]=\frac{1}{\epsilon_1+\epsilon_2}\left(\sigma^{ab}_1+\sigma^{ba}_1\right)e^{(b)}_m,\quad\left[\psi^{(a)}_0,f^{(b)}_m\right]=-\frac{1}{\epsilon_1+\epsilon_2}\left(\sigma^{ab}_1+\sigma^{ba}_1\right)f^{(b)}_m
	\end{equation}
    can be deduced from the $\psi e$ and $\psi f$ relations.
\end{definition}

Notice that for future convenience, we have rescaled the generators compared to the original convention in \cite{Li:2020rij} with $\mathtt{e}^{(a)}_n$, $\mathtt{f}^{(a)}_n$ and $\uppsi^{(a)}_n$. The two sets of modes are related by $\mathtt{e}^{(a)}_n=(\epsilon_1+\epsilon_2)^{1/2}e^{(a)}_n$, $\mathtt{f}^{(a)}_n=(\epsilon_1+\epsilon_2)^{1/2}f^{(a)}_n$ and $\uppsi^{(a)}_n=(\epsilon_1+\epsilon_2)\psi^{(a)}_n$ (including $\uppsi^{(a)}_{-1}$).

To correctly recover the BPS degeneracies, we also need the Serre relations.
\begin{definition}
	Given the above quiver data, the (reduced) quiver Yangian $\mathtt{Y}_{Q,W}$ is the non-reduced quiver Yangian with the Serre relations given as follows. When $MN\neq2$, we have
	\begin{align}
		&\textup{Sym}_{n_1,n_2}\left[e^{(a)}_{n_1},\left[e^{(a)}_{n_2},e^{(a\pm1)}_{m}\right]\right]=0&\qquad(|a|=0),\label{Serre1}\\
		&\textup{Sym}_{n_1,n_2}\left[e^{(a)}_{n_1},\left[e^{(a+1)}_{m_1},\left[e^{(a)}_{n_2},e^{(a-1)}_{m_2}\right]\right]\right]=0&\qquad(|a|=1)\label{Serre2},
	\end{align}
    and the same relations with all $e$ replaced by $f$. When $(M,N)=(2,1)$ (or equivalently, $(M,N)=(1,2)$), namely for the suspended pinch point (SPP), we have
    \begin{equation}
    	\begin{split}
    		&\textup{Sym}_{n_1,n_2}\textup{Sym}_{m_1,m_2}\left[e^{(0)}_{n_1},\left[e^{(2)}_{m_1},\left[e^{(0)}_{n_2},\left[e^{(2)}_{m_2},e^{(1)}_{k}\right]\right]\right]\right]\\
    		=&\textup{Sym}_{n_1,n_2}\textup{Sym}_{m_1,m_2}\left[e^{(2)}_{m_1},\left[e^{(0)}_{n_1},\left[e^{(2)}_{m_2},\left[e^{(0)}_{n_2},e^{(1)}_{k}\right]\right]\right]\right],
    	\end{split}
    \end{equation}
    and the same relation with all $e$ replaced by $f$, where the node $(1)$ is taken to be the single bosonic node.
\end{definition}
For toric CYs, as the superpotential can be unambiguously determined for a given quiver, we shall abbreviate $\mathtt{Y}_{Q,W}$ as $\mathtt{Y}_Q$ or even $\mathtt{Y}$ if it would not cause confusions. Moreover, since the quiver Yangian\footnote{In this paper, the name quiver Yangian is always referred to as the reduced one.} is always a two-parameter Yangian algebra, we will omit $\epsilon_i$ as well.

\section{Coproduct of Quiver Yangians}\label{coprod}
Following the strategy in \cite{guay2018coproduct,ueda2019affine}, a coproduct of the quiver Yangian can be obtained based on the underlying Kac-Moody superalgebra $\mathfrak{g}=A^{(1)}_{M-1,N-1}$. In the Chevalley basis, we have the generators with $\left[x^{(a)}_+,x^{(a)}_-\right]=h^{(a)}$ and $\left(x^{(a)}_+,x^{(a)}_-\right)=1$, where $(\text{-},\text{-})$ is an invariant inner product on the Kac-Moody superalgebra. Let $\varDelta=\varDelta_+\cup\varDelta_-$ be the set of roots composed of positive and negative roots. Denote the sets of real and imaginary roots as $\varDelta^\text{re}$ and $\varDelta^\text{im}$ respectively. Write $\mathfrak{g}_\alpha$ as the root space attached to the root $\alpha$, and the simple roots will be labelled as $\alpha^{(a)}$. In particular, $\varDelta^\text{re}_+=\mathring{\varDelta}_+\cup\left\{n\delta+\alpha|n\in\mathbb{Z}_+,\alpha\in\mathring{\varDelta}\right\}$ and $\varDelta^\text{im}_+=\{n\delta|n\in\mathbb{Z}_+\}$, where $\mathring{\varDelta}$ is the set of roots of the underlying Lie superalgebra with the zeroth vertex removed in the Dynkin diagram of $\mathfrak{g}$ and $\delta=\sum\limits_{a}\alpha^{(a)}$ is the minimal positive imaginary root of $\mathfrak{g}=A^{(1)}_{M-1,N-1}$. Notice that all the odd roots are isotropic (i.e., with vanishing inner product) in such case.

Following the Cartan matrix (with the first non-zero diagonal element being 2), our convention would be taken as $\sigma^{ab}_1+\sigma^{ba}_1=(\epsilon_1+\epsilon_2)\left(\alpha^{(a)},\alpha^{(b)}\right)$ for future convenience. In other words, we shall always choose $\varsigma_a=-1$ for the corresponding simplex in the toric diagram. Therefore, $\sigma^{aa}_1=-\frac{1}{2}\epsilon_3\left(\alpha^{(a)},\alpha^{(a)}\right)$. Then it is straightforward to see that there is an algebra homomorphism $\iota$ from $U(\mathfrak{g})$ to $\mathtt{Y}$ with $h^{(a)}\mapsto\psi^{(a)}_0$, $x^{(a)}_+\mapsto e^{(a)}_0$ and $x^{(a)}_-\mapsto f^{(a)}_0$. For each positive root $\alpha$, choose a basis $\left\{x^{(\alpha,k)}_+\right\}$ of $\mathfrak{g}_\alpha$ with a dual basis $\left\{x^{(\alpha,k)}_-\right\}$ of $\mathfrak{g}_{-\alpha}$ such that $\left(x^{(\alpha,k)}_+,x^{(\alpha,l)}_-\right)=\delta_{kl}$. We will also denote $e^{(\alpha,k)}=\iota\left(x^{(\alpha,k)}_+\right)$ and $f^{(\alpha,k)}=\iota\left(x^{(a,k)}_-\right)$, where $k=1,\dots,\dim\mathfrak{g}_\alpha$. When $\alpha$ is a real root, $\dim\mathfrak{g}_\alpha=1$ and we shall simply write $e^{(\alpha)}=e^{(\alpha,1)}$, $f^{(\alpha)}=f^{(\alpha,1)}$. In particular, given a simple root $\alpha^{(a)}$, we have $e^{(a)}_0=e^{\left(\alpha^{(a)}\right)}$ and $f^{(a)}_0=f^{\left(\alpha^{(a)}\right)}$.

\subsection{A Minimalistic Presentation}\label{minimalistic}
The definition of the quiver Yangian in \S\ref{QY} involves infinitely many generators. To write the coproduct, we first need to give a presentation with generators of a finite number.

As shown in \cite{Bao:2022fpk}, all the generators can in fact be inductively obtained from $e^{(a)}_0$, $f^{(a)}_0$ and $\psi^{(a)}_{0,1}$ by
\begin{align}
	&e^{(a)}_{m+1}=\frac{1}{\left(\alpha^{(a)},\alpha^{(b)}\right)}\left[\widetilde{\psi}^{(b)}_1,e^{(a)}_m\right]-\frac{\sigma^{ab}_1-\sigma^{ba}_1}{2\left(\alpha^{(a)},\alpha^{(b)}\right)}\left[\psi^{(b)}_0,e^{(a)}_m\right]=\frac{1}{\left(\alpha^{(a)},\alpha^{(b)}\right)}\left[\widetilde{\psi}^{(b)}_1,e^{(a)}_m\right]-\frac{\sigma^{ab}_1-\sigma^{ba}_1}{2}e^{(a)}_m,\nonumber\\
	&f^{(a)}_{m+1}=-\frac{1}{\left(\alpha^{(a)},\alpha^{(b)}\right)}\left[\widetilde{\psi}^{(b)}_1,f^{(a)}_m\right]+\frac{\sigma^{ab}_1-\sigma^{ba}_1}{2\left(\alpha^{(a)},\alpha^{(b)}\right)}\left[\psi^{(b)}_0,f^{(a)}_m\right]=-\frac{1}{\left(\alpha^{(a)},\alpha^{(b)}\right)}\left[\widetilde{\psi}^{(b)}_1,f^{(a)}_m\right]-\frac{\sigma^{ab}_1-\sigma^{ba}_1}{2}f^{(a)}_m,\nonumber\\
	&\psi^{(a)}_{m+1}=\left[e^{(a)}_{m+1},f^{(a)}_0\right],\nonumber\\
	\label{finitegenerators}
\end{align}
where $\widetilde{\psi}^{(b)}_1:=\psi^{(b)}_1-\frac{\epsilon_1+\epsilon_2}{2}\left(\psi^{(b)}_0\right)^2$, and the node $b$ can be taken as $a$ (resp. $a+1$) when $a$ is bosonic (resp. fermionic). Therefore, it is natural to expect that the quiver Yangian can be generated only by finitely many relations of the three sets of zero modes together with $\psi^{(a)}_1$ (or equivalently, $\widetilde{\psi}^{(a)}_1$). To confirm that this is the case, we need to show that they can recover all the defining relations of the quiver Yangian in \S\ref{QY}.

For this minimalistic presentation of the (reduced) quiver Yangians\footnote{We discuss how we can get all the generators from finitely many of them for the quiver Yangians associated to $\mathbb{C}\times\mathbb{C}^2/\mathbb{Z}_2$ and the conifold in Appendix \ref{generators2} (though their possible minimalistic presentations with finitely many relations are postponed to future study). We also mention the minimalistic presentation for $\mathbb{C}^3/(\mathbb{Z}_2\times\mathbb{Z}_2)$ therein although it is not a generalized conifold.}, we also need to exclude two special cases: (1) $(M,N)=(2,1),(1,2)$; (2) $(M,N)=(2,2)$ with only fermionic nodes\footnote{Notice that the other toric phase for $(M,N)=(2,2)$ is not excluded, where the quiver has four nodes being bosonic and fermionic alternatively, though all cases with $M=N$ will not be considered when we discuss coproducts later.}. Their quivers are depicted in Figure \ref{exclude}. We will further comment on this at the end of this subsection.
\begin{figure}[h]
	\centering
	\includegraphics[width=6cm]{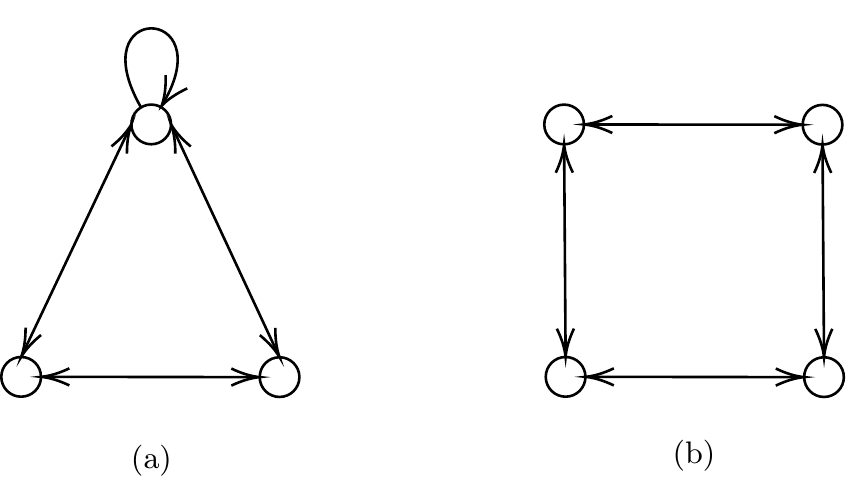}
	\caption{(a) The SPP quiver with $(M,N)=(2,1),(1,2)$. (b) The quiver for the $\mathbb{Z}_2$ orbifold of the conifold with $(M,N)=(2,2)$ in one of the toric phases.}\label{exclude}
\end{figure}

\begin{theorem}\label{minthm}
	For generalized conifolds with $M+N>2$, the non-reduced quiver Yangian is generated by the modes $e^{(a)}_r$, $f^{(a)}_r$ and $\psi^{(a)}_r$ ($a\in Q_0$, $r=0,1$) satisfying the relations
	\begin{align}
		&\left[\psi^{(a)}_r,\psi^{(b)}_s\right]=0,\label{psipsimin}\\
		&\left[e^{(a)}_0,f^{(b)}_0\right]=\delta_{ab}\psi^{(a)}_0,\quad\left[e^{(a)}_1,f^{(b)}_0\right]=\left[e^{(a)}_0,f^{(b)}_1\right]=\delta_{ab}\psi^{(a)}_1,\label{efmin}\\
		&\left[\psi^{(a)}_0,e^{(b)}_r\right]=\left(\alpha^{(a)},\alpha^{(b)}\right)e^{(b)}_r,\label{psi0emin}\\
		&\left[\psi^{(a)}_1,e^{(b)}_0\right]=\left(\alpha^{(a)},\alpha^{(b)}\right)e^{(b)}_1+\sigma^{ba}_1\psi^{(a)}_0e^{(b)}_0+\sigma^{ab}_1e^{(b)}_0\psi^{(a)}_0,\label{psi1emin}\\
		&\left[\psi^{(a)}_0,f^{(b)}_r\right]=-\left(\alpha^{(a)},\alpha^{(b)}\right)f^{(b)}_r,\label{psi0fmin}\\
		&\left[\psi^{(a)}_1,f^{(b)}_0\right]=-\left(\alpha^{(a)},\alpha^{(b)}\right)f^{(b)}_1+\sigma^{ab}_1\psi^{(a)}_0f^{(b)}_0+\sigma^{ba}_1f^{(b)}_0\psi^{(a)}_0,\label{psi1fmin}\\
		&\left[e^{(a)}_0,e^{(b)}_0\right]=\left[f^{(a)}_0,f^{(b)}_0\right]=0\qquad(\sigma_1^{ab}=0),\label{eeffmin}\\
		&\left[e^{(a)}_1,e^{(b)}_0\right]-\left[e^{(a)}_0,e^{(b)}_1\right]=\sigma^{ba}_1e^{(a)}_0e^{(b)}_0+(-1)^{|a||b|}\sigma^{ab}_1e^{(b)}_0e^{(a)}_0,\label{eemin}\\
		&\left[f^{(a)}_1,f^{(b)}_0\right]-\left[f^{(a)}_0,f^{(b)}_1\right]=-\sigma^{ab}_1f^{(a)}_0f^{(b)}_0-(-1)^{|a||b|}\sigma^{ba}_1f^{(b)}_0f^{(a)}_0.\label{ffmin}
	\end{align}
    Then the higher modes $\psi^{(a)}_n$, $e^{(a)}_m$ and $f^{(a)}_m$ ($n>1,m>0$) are defined via \eqref{finitegenerators}. Moreover, when the quiver does \emph{not} belong to those in Figure \ref{exclude}, the quiver Yangian is generated by these relations together with
    \begin{align}
    	&\left[e^{(a)}_0,\left[e^{(a)}_0,e^{(a\pm1)}_0\right]\right]=\left[f^{(a)}_0,\left[f^{(a)}_0,f^{(a\pm1)}_0\right]\right]=0&\qquad(|a|=0),\label{Serre1min}\\
    	&\left[e^{(a)}_0,\left[e^{(a+1)}_0,\left[e^{(a)}_0,e^{(a-1)}_0\right]\right]\right]=\left[f^{(a)}_0,\left[f^{(a+1)}_0,\left[f^{(a)}_0,f^{(a-1)}_0\right]\right]\right]=0&\qquad(|a|=1).\label{Serre2min}
    \end{align}
\end{theorem}
In terms of $\widetilde{\psi}^{(a)}_1$, the $\psi_1e_0$ and $\psi_1f_0$ relations can be written as
\begin{equation}
	\begin{split}
		&\left[\widetilde{\psi}^{(a)}_1,e^{(b)}_0\right]=\left(\alpha^{(a)},\alpha^{(b)}\right)e^{(b)}_1+\frac{\sigma^{ba}_1-\sigma^{ab}_1}{2}\left(\alpha^{(a)},\alpha^{(b)}\right)e^{(b)}_0,\\
		&\left[\widetilde{\psi}^{(a)}_1,f^{(b)}_0\right]=-\left(\alpha^{(a)},\alpha^{(b)}\right)f^{(b)}_1-\frac{\sigma^{ba}_1-\sigma^{ab}_1}{2}\left(\alpha^{(a)},\alpha^{(b)}\right)f^{(b)}_0.
	\end{split}\label{psitildeefmin}
\end{equation}
For simplicity, we shall denote all the relations \eqref{psipsimin}$\sim$\eqref{ffmin} as (R) and the relations \eqref{Serre1min}, \eqref{Serre2min} as (S).
\begin{proof}
	The proof is similar to the ones in \cite{guay2018coproduct,ueda2019affine}. Essentially, we need to show that the defining relations in \S\ref{QY} can be deduced from (R) for the non-reduced quiver Yangians and also (S) for the reduced ones. This follows from Lemma \ref{psi01eflemma}$\sim$Lemma \ref{Serrelemma} below.
\end{proof}

\begin{lemma}
	For $m\in\mathbb{N}$, we have
	\begin{align}
		&\left[\psi^{(a)}_0,e^{(b)}_m\right]=\left(\alpha^{(a)},\alpha^{(b)}\right)e^{(b)}_m,\\
		&\left[\psi^{(a)}_1,e^{(b)}_m\right]=\left(\alpha^{(a)},\alpha^{(b)}\right)e^{(b)}_{m+1}+\sigma^{ba}_1\psi^{(a)}_0e^{(b)}_m+\sigma^{ab}_1e^{(b)}_m\psi^{(a)}_0,
	\end{align}
    and similar relations for $f^{(b)}_m$ from \textup{(R)}.\label{psi01eflemma}
\end{lemma}
It would be helpful to also spell out these relations using $\widetilde{\psi}^{(a)}_1$:
\begin{equation}
	\begin{split}
		&\left[\widetilde{\psi}^{(a)}_1,e^{(b)}_m\right]=\left(\alpha^{(a)},\alpha^{(b)}\right)e^{(b)}_{m+1}+\frac{\sigma^{ba}_1-\sigma^{ab}_1}{2}\left(\alpha^{(a)},\alpha^{(b)}\right)e^{(b)}_m,\\
		&\left[\widetilde{\psi}^{(a)}_1,f^{(b)}_m\right]=-\left(\alpha^{(a)},\alpha^{(b)}\right)f^{(b)}_{m+1}-\frac{\sigma^{ba}_1-\sigma^{ab}_1}{2}\left(\alpha^{(a)},\alpha^{(b)}\right)f^{(b)}_m.
	\end{split}
\end{equation}
\begin{proof}
	This can be proven by induction. When $m=0$, this is the given \eqref{psi0emin}. Now suppose both the $\psi_0e_m$ and $\psi_1e_m$ relations hold for $m=k$. Then
	\begin{equation}
		\begin{split}
			\left[\psi^{(a)}_0,e^{(b)}_{k+1}\right]=&\left[\psi^{(a)}_0,\frac{1}{\left(\alpha^{(c)},\alpha^{(b)}\right)}\left[\widetilde{\psi}^{(c)}_1,e^{(b)}_k\right]-\frac{\sigma^{bc}_1-\sigma^{cb}_1}{2}e^{(b)}_k\right]\\
			=&\frac{1}{\left(\alpha^{(c)},\alpha^{(b)}\right)}\left[\widetilde{\psi}^{(c)}_1,\left[\psi^{(a)}_0,e^{(b)}_k\right]\right]-\frac{\sigma^{bc}_1-\sigma^{cb}_1}{2}\left[\psi^{(a)}_0,e^{(b)}_k\right]\\
			=&\left(\alpha^{(a)},\alpha^{(b)}\right)e^{(b)}_{k+1},
		\end{split}
	\end{equation}
    where the second equality follows from the fact that $\psi^{(a)}_0$ commutes with $\widetilde{\psi}^{(c)}_1$, and we have used the $\psi_0e_m$ relation at $m=k$ followed by expanding the only remaining bracket to obtain the third equality. The $\psi_1e_m$ relation, as well as the $\psi_{0,1}f_m$ relations, can be shown in the same manner using $\widetilde{\psi}^{(a)}_1$.
\end{proof}

With this useful lemma prepared, let us first show the $\psi\psi$ and $ef$ relations.
\begin{lemma}
	For $m,n\in\mathbb{N}$, we have \eqref{psipsi} and \eqref{ef} from \textup{(R)}.
\end{lemma}
\begin{proof}
	Again, we prove this by induction. When $n=m=0$, they are the given \eqref{psipsimin} and \eqref{efmin}. Now suppose the relations hold when $n+m=p=l+k$ for any such $m,n$. Then we can consider the commutation relation of $\widetilde{\psi}^{(c)}_1$ with the (super) commutators in question. For $e^{(a)}_l$ and $f^{(a)}_k$, we have
	\begin{equation}
		0=\left[\widetilde{\psi}^{(c)}_1,\psi^{(a)}_{l+k}\right]=\left[\widetilde{\psi}^{(c)}_1,\left[e^{(a)}_l,f^{(a)}_k\right]\right],
	\end{equation}
    where $c$ is taken to be $a$ if $|a|=0$ and $a+1$ otherwise. Expanding the brackets and using Lemma \ref{psi01eflemma}, this becomes
    \begin{equation}
    	(\alpha^{(a)},\alpha^{(b)})\left(\left[e^{(a)}_{l+1}-f^{(a)}_k\right]-\left[e^{(a)}_l-f^{(a)}_{k+1}\right]\right)=0.
    \end{equation}
    Therefore,
    \begin{equation}
    	\left[e^{(a)}_{l+1},f^{(a)}_k\right]=\left[e^{(a)}_l,f^{(a)}_{k+1}\right].
    \end{equation}
    We first take $l=p$ and $k=0$ and get
    \begin{equation}
    	\left[e^{(a)}_p,f^{(a)}_1\right]=\left[e^{(a)}_{p+1},f^{(a)}_0\right]=\psi^{(a)}_{p+1},
    \end{equation}
    where the second equality follows from our definition of the higher modes as in \eqref{finitegenerators}. Next, we can take $l=p-1$ and $k=1$ to get the one for $e^{(a)}_{p-1}$ and $f^{(a)}_2$. Then all relations for the possible $(l,k)$ with $l+k+1=p+1$ can be obtained in such way.
    
    Likewise, for the $ef$ relation with $a\neq b$, we consider
    \begin{equation}
    	0=\left[\widetilde{\psi}^{(c)}_1,\left[e^{(a)}_l,f^{(b)}_k\right]\right].
    \end{equation}
    Here, $c$ can be choosen such that $\sigma^{ac}_1\neq0$ and/or $\sigma^{bc}_1\neq0$. Again, after using Lemma \ref{psi01eflemma}, we can reach the relation for any $l+k+1=p+1$. This proves the $ef$ relation. The $\psi\psi$ relation can be shown using $\widetilde{\psi}^{(c)}_1$ in the similar way.
\end{proof}

Now, let us show the $\psi e$, $\psi f$, $ee$ and $ff$ relations.
\begin{lemma}
	For $m,n\in\mathbb{N}$, we have \eqref{psie}$\sim$\eqref{ff} from \textup{(R)}.
\end{lemma}
\begin{proof}
	We again prove this by induction. When $n=m=0$, they are given by \eqref{psi0emin}$\sim$\eqref{ffmin}. Suppose they hold for $n=l,m=k$. We need to show that they are still true for $(l+1,k)$ and $(l,k+1)$. Consider
	\begin{equation}
		\left[\widetilde{\psi}^{(a)}_1,\left[e^{(a)}_{l+1},e^{(b)}_k\right]-\left[e^{(a)}_l,e^{(b)}_{k+1}\right]\right]=\left[\widetilde{\psi}^{(a)}_1,\sigma^{ba}_1e^{(a)}_le^{(b)}_k+(-1)^{|a||b|}e^{(b)}_ke^{(a)}_l\right],
	\end{equation}
    which follows from our assumption when $n=l,m=k$. Using Lemma \ref{psi01eflemma}, a straightforward calculation yields
    \begin{equation}
    	\begin{split}
    		&2\sigma^{aa}_1\left(\left[e^{(a)}_{l+2},e^{(b)}_k\right]-\left[e^{(a)}_{l+1},e^{(b)}_{k+1}\right]\right)+\left(\sigma^{ab}_1+\sigma^{ba}_1\right)\left(\left[e^{(a)}_{l+1},e^{(b)}_{k+1}\right]-\left[e^{(a)}_l,e^{(b)}_{k+2}\right]\right)\\
    		=&2\sigma^{aa}_1\left(\sigma^{ba}_1e^{(a)}_{l+1}e^{(b)}_k+(-1)^{|a||b|}e^{(b)}_ke^{(a)}_{l+1}\right)+\left(\sigma^{ab}_1+\sigma^{ba}_1\right)\left(\sigma^{ba}_1e^{(a)}_le^{(b)}_{k+1}+(-1)^{|a||b|}e^{(b)}_{k+1}e^{(a)}_l\right).
    	\end{split}
    \end{equation}
    Replacing $\widetilde{\psi}^{(a)}_1$ with $\widetilde{\psi}^{(b)}_1$, we obtain
    \begin{equation}
    	\begin{split}
    		&\left(\sigma^{ab}_1+\sigma^{ba}_1\right)\left(\left[e^{(a)}_{l+2},e^{(b)}_k\right]-\left[e^{(a)}_{l+1},e^{(b)}_{k+1}\right]\right)+2\sigma^{bb}_1\left(\left[e^{(a)}_{l+1},e^{(b)}_{k+1}\right]-\left[e^{(a)}_l,e^{(b)}_{k+2}\right]\right)\\
    		=&\left(\sigma^{ab}_1+\sigma^{ba}_1\right)\left(\sigma^{ba}_1e^{(a)}_{l+1}e^{(b)}_k+(-1)^{|a||b|}e^{(b)}_ke^{(a)}_{l+1}\right)+2\sigma^{bb}_1\left(\sigma^{ba}_1e^{(a)}_le^{(b)}_{k+1}+(-1)^{|a||b|}e^{(b)}_{k+1}e^{(a)}_l\right).
    	\end{split}
    \end{equation}
    When $|a||b|=0$ or when $|a||b|=1$ with $\sigma^{ab}_1\neq0$, we essentially have two linearly independent equations with solution
    \begin{equation}
    	\begin{split}
    		&\left(\left[e^{(a)}_{l+2},e^{(b)}_k\right]-\left[e^{(a)}_{l+1},e^{(b)}_{k+1}\right]\right)-\left(\sigma^{ba}_1e^{(a)}_{l+1}e^{(b)}_k+(-1)^{|a||b|}e^{(b)}_ke^{(a)}_{l+1}\right)=0,\\
    		&\left(\left[e^{(a)}_{l+1},e^{(b)}_{k+1}\right]-\left[e^{(a)}_l,e^{(b)}_{k+2}\right]\right)-\left(\sigma^{ba}_1e^{(a)}_le^{(b)}_{k+1}+(-1)^{|a||b|}e^{(b)}_{k+1}e^{(a)}_l\right)=0.
    	\end{split}
    \end{equation}
    When $|a||b|=1$ with $\sigma^{ab}_1=0$, the two linear equations are trivial, but we can instead use $\widetilde{\psi}^{(a\pm1)}_1$ and $\widetilde{\psi}^{(b\pm1)}_1$. This leads to the relation \eqref{ee}. The other relations can be shown in the similar way using $\widetilde{\psi}^{(a)}_1$.
\end{proof}

Now, we have proven the non-reduced quiver Yangian part in Theorem \ref{minthm}. To complete the proof for (reduced) quiver Yangians, we need to check the Serre relations.
\begin{lemma}
	For $m,m_{1,2},n_{1,2}\in\mathbb{N}$, we have \eqref{Serre1} and \eqref{Serre2} from \textup{(S)}.\label{Serrelemma}
\end{lemma}
\begin{proof}
	Again, we prove this by induction. When $m,m_{1,2},n_{1,2}$ are zero, they are given in (S). First, let us consider the relation \eqref{Serre1} with $|a|=0$. For brevity, we shall write it as $\text{Sym}_{n_1,n_2}(m)=0$. Suppose this holds when $n_1+n_2+m=p=l_1+l_2+k$ for any such $n_{1,2},m$. We still consider the commutation relations with $\widetilde{\psi}^{(a)}_1$ and $\widetilde{\psi}^{(a\pm1)}_1$, and we get the equations
	\begin{equation}
		\begin{cases}
			-2\text{Sym}_{l_1+1,l_2}(k)-2\text{Sym}_{l_1,l_2+1}(k)+\text{Sym}_{l_1,l_2}(k)=0,&\\
			\text{Sym}_{l_1+1,l_2}(k)+\text{Sym}_{l_1,l_2+1}(k)-2\text{Sym}_{l_1,l_2}(k)=0.&
		\end{cases}
	\end{equation}
    Taking $l_1=l_2=l$, we have two variables for the two linearly independent equations with solution
    \begin{equation}
    	\text{Sym}_{l+1,l}(k)=\text{Sym}_{l,l}(k+1)=0.
    \end{equation}
    Next, taking $l_1=l-1$ and $l_2=l+1$, the two linear independent equations again have two variables solved by
    \begin{equation}
    	\text{Sym}_{l-1,l+2}(k)=\text{Sym}_{l-1,l+1}(k+1)=0.
    \end{equation}
    Keep this procedure, and we can prove this relation for any $l_1+l_2+k+1=p+1$.
    
    Now, let us consider the relation \eqref{Serre2} with $|a|=1$. For brevity, we shall write it as $\text{Sym}_{n_1,n_2}(m_1,m_2)=0$. Suppose this holds when $n_1+n_2+m_1+m_2=p=l_1+l_2+k_1+k_2$ for any such $n_{1,2},m_{1,2}$. We still consider the commutation relations with $\widetilde{\psi}^{(a)}_1$ and $\widetilde{\psi}^{(a\pm1)}_1$, and we get the equations
    \begin{align}
    	&\left(\sigma^{a,a+1}_1+\sigma^{a+1,a}_1\right)\text{Sym}_{l_1,l_2}(k_1+1,k_2)+\left(\sigma^{a,a-1}_1+\sigma^{a-1,a}_1\right)\text{Sym}_{l_1,l_2}(k_1,k_2+1)=0,\label{S2prf1}\\
    	&\left(\sigma^{a,a+1}_1+\sigma^{a+1,a}_1\right)\text{Sym}_{l_1+1,l_2}(k_1,k_2)+2\sigma^{a+1,a+1}_1\text{Sym}_{l_1,l_2}(k_1+1,k_2)\nonumber\\
    	&+\left(\sigma^{a,a+1}_1+\sigma^{a+1,a}_1\right)\text{Sym}_{l_1,l_2+1}(k_1,k_2)=0,\label{S2prf2}\\
    	&\left(\sigma^{a,a-1}_1+\sigma^{a-1,a}_1\right)\text{Sym}_{l_1+1,l_2}(k_1,k_2)+2\sigma^{a-1,a-1}_1\text{Sym}_{l_1,l_2}(k_1,k_2+1)\nonumber\\
    	&+\left(\sigma^{a,a-1}_1+\sigma^{a-1,a}_1\right)\text{Sym}_{l_1,l_2+1}(k_1,k_2)=0,\label{S2prf3}
    \end{align}
    where we have used $\sigma^{a-1,a+1}_1=\sigma^{a+1,a-1}_1=0$ since the quiver (with at least one fermionic node) would always have more than three nodes here. In particular, \eqref{S2prf1} can be simplified to
    \begin{equation}
    	\text{Sym}_{l_1,l_2}(k_1+1,k_2)-\text{Sym}_{l_1,l_2}(k_1,k_2+1)=0\label{S2prf4}
    \end{equation}
    since $(a)$ is fermionic. Adding \eqref{S2prf2} and \eqref{S2prf3} together, we have
    \begin{equation}
    	2\sigma^{a+1,a+1}_1\text{Sym}_{l_1,l_2}(k_1+1,k_2)+2\sigma^{a-1,a-1}_1\text{Sym}_{l_1,l_2}(k_1,k_2+1)=0.\label{S2prf5}
    \end{equation}
    To get more linearly independent equations, we can also consider the one using $\text{ad}_{\widetilde{\psi}^{(a)}_1}$, which yields
    \begin{equation}
    	\left(\sigma^{a+1,a+2}_1+\sigma^{a+2,a+1}_1\right)\text{Sym}_{l_1,l_2}(k_1+1,k_2)+\left(\sigma^{a+2,a-1}_1+\sigma^{a-1,a+2}_1\right)\text{Sym}_{l_1,l_2}(k_1,k_2+1)=0.\label{S2prf6}
    \end{equation}
    Since $\sigma^{a+2,a-1}_1=\sigma^{a-1,a+2}_1=0$ here, this can be simplified to $\text{Sym}_{l_1,l_2}(k_1+1,k_2)=0$. Together with \eqref{S2prf5}, we find $\text{Sym}_{l_1,l_2}(k_1,k_2+1)=0$. Now, \eqref{S2prf2} becomes
    \begin{equation}
    	\text{Sym}_{l_1+1,l_2}(k_1,k_2)+\text{Sym}_{l_1,l_2+1}(k_1,k_2)=0.\label{S2prf7}
    \end{equation}
    Apply the same trick as before and take $l_1=l_2=l$, we have $\text{Sym}_{l,l+1}(k_1,k_2)=0$. Next, taking $l_1=l-1,l_2=l+1$, we get $\text{Sym}_{l-1,l+2}(k_1,k_2)=0$. Keep this procedure, and we can find this holds for any $l_{1,2},k_{1,2}$ satisfying $l_1+l_2+k_1+k_2+1=p+1$. This completes the proof for the Serre relations.
\end{proof}

Let us now make a comment on the cases in Figure \ref{exclude}. We have to exclude them in Theorem \ref{minthm} since we do not have enough nodes to get sufficiently many linearly independent equations. For Figure \ref{exclude}(a), we have five different modes in one relation, but there are only three nodes in the quiver. For Figure \ref{exclude}(b), we have four modes in one relation while considering $\text{ad}_{\widetilde{\psi}^{(a+2)}_1}$ is also not very useful since \eqref{S2prf6} would coincide with \eqref{S2prf5}. We leave these cases to future work.

\subsection{Another Presentation and Coproduct}\label{Jpresentation}
From now on, besides the restrictions $M+N>2$ and $MN\neq2$, we will mainly focus on the cases with $M\neq N$ due to the subtleties from the underlying simple Lie superalgebra $\mathfrak{psl}(M|M)$ (when $M=N$). Analogous to \cite{guay2018coproduct}, we can write an algebra homomorphism $\Delta_{V_1,V_2}:\mathtt{Y}\rightarrow\text{End}_{\mathbb{C}}(V_1\otimes V_2)$ for any modules $V_{1,2}$ in the category $\mathcal{O}$. In particular, this can be promoted to a coproduct of the Yangian algebra by considering its completion $\widehat{\mathtt{Y}}$ following the argument in \cite[\S5]{guay2018coproduct}. Then any $\Delta_{V_1,V_2}$ can be recovered from $\Delta:\mathtt{Y}\rightarrow\mathtt{Y}\widehat{\otimes}\mathtt{Y}$, where $\mathtt{Y}\widehat{\otimes}\mathtt{Y}$ is the completion of $\mathtt{Y}\otimes\mathtt{Y}$ we are now going to discuss.

The quiver Yangian has the triangular decomposition $\mathtt{Y}\cong\mathtt{Y}^+\otimes\mathtt{Y}^0\otimes\mathtt{Y}^-$, where $\mathtt{Y}^+$ ($\mathtt{Y}^-$, resp. $\mathtt{Y}^0$) is generated by $e^{(a)}_n$ ($f^{(a)}_n$, resp. $\psi^{(a)}_n$) for all $a\in Q_0$ and $n\in\mathbb{N}$ \cite{Li:2020rij}. We shall also assume that the positive (resp. negative) part $\mathtt{Y}^+$ (resp. $\mathtt{Y}^-$) is isomorphic to the free algebra on $e^{(a)}_n$ (resp. $f^{(a)}_n$) quotiented out by the $ee$ (resp. $ff$) relations. We will denote the subalgebra generated by $e^{(a)}_n$ (resp. $f^{(a)}_n$) and $\psi^{(a)}_n$ as $\mathtt{Y}^{\geq0}$ (resp. $\mathtt{Y}^{\leq0}$).

We can set a degree as $\deg e^{(a)}_n=1$ whose grading is compatible with the algebra structure. With respect to this grading, $\mathtt{Y}^+=\bigoplus\limits_{k=0}^\infty\mathtt{Y}^+_k$ with $\mathtt{Y}^+_k$ spanned by monomials of degree $k$ in $\mathtt{Y}^+$. We also write $\mathtt{Y}^+_{\geq n}:=\bigoplus\limits_{k\geq n}^\infty\mathtt{Y}^+_k$. Therefore, the quiver Yangian is a graded \emph{vector space} as $\mathtt{Y}=\bigoplus\limits_{k=0}^\infty\mathtt{Y}_k$, where $\mathtt{Y}_k=\mathtt{Y}^{\leq0}\otimes\mathtt{Y}^+_k$. Now consider the pair $(A_n,q_n)$ for $n\in\mathbb{N}$ with the left $\mathtt{Y}$-module $A_n:=\mathtt{Y}/\left(\mathtt{Y}\cdot\mathtt{Y}^+_{\geq n}\right)$ and the natural quotient map $q_n$ from $\mathtt{Y}$ to $A_n$. Then $q_{n-1}$ factors through $A_n$, that is, $p_n\circ q_n=q_{n-1}$ with the homomorphism $p_n:A_n\rightarrow A_{n-1}$. The pairs $(A_n,p_n)$ give rise to an inverse system of $\mathtt{Y}$-modules, and we can define the completion of the quiver Yangian as the projective limit \cite[\S10.1]{chari1995guide}:
\begin{equation}
	\widehat{\mathtt{Y}}:=\varprojlim_{n}A_n.
\end{equation}
We will also write $\mathtt{Y}\widehat{\otimes}\mathtt{Y}$ as the completion of $\mathtt{Y}\otimes\mathtt{Y}$.

To write down the coproduct of the quiver Yangian, we need another presentation of the algebra. Drinfeld's $J$ presentation is used for finite dimensional cases in \cite{Drinfeld:1985rx}, but can be appropriately extended to affine cases following the recipe of \cite{guay2018coproduct}. In this presentation, the quiver Yangian is generated by $x$ and $J(x)$ for $x$ elements of the underlying Kac-Moody superalgebra $\mathfrak{g}$. Together with the Chevalley generators of $\mathfrak{g}$ mapped to the zero modes of $\mathtt{Y}$ (recall the beginning of \S\ref{coprod}), the isomorphism is given by
\begin{equation}
	J\left(\psi^{(a)}_0\right)=\psi^{(a)}_1+v^{(a)},\quad J\left(e^{(a)}_0\right)=e^{(a)}_1+w^{(a)}_+,\quad
	J\left(f^{(a)}_0\right)=f^{(a)}_1+w^{(a)}_-,
\end{equation}
where\footnote{Notice that this is a well-defined operator when acting on modules in the category $\mathcal{O}$ as $e^{(\alpha,k)}$ (i.e., $x^{(\alpha,k)}_+$) annihilates a vector for $\alpha$ with sufficiently large height.} \cite{guay2018coproduct}
\begin{equation}
	v^{(a)}=\frac{1}{2}(\epsilon_1+\epsilon_2)\sum_{\alpha\in\varDelta_+}\left(\alpha,\alpha^{(a)}\right)\sum_{k=1}^{\dim\mathfrak{g}_\alpha}f^{(\alpha,k)}e^{(\alpha,k)}-\frac{1}{2}(\epsilon_1+\epsilon_2)\left(\psi^{(a)}_0\right)^2.
\end{equation}
Then $w^{(a)}_{\pm}$ can be obtained by requiring
\begin{equation}
	J\left(\left[\psi^{(a)}_0,e^{(a)}_0\right]\right)=\left[J\left(\psi^{(a)}_0\right),e^{(a)}_0\right],\quad J\left(\left[\psi^{(a)}_0,f^{(a)}_0\right]\right)=\left[J\left(\psi^{(a)}_0\right),f^{(a)}_0\right].\label{wrequirement}
\end{equation}

In general, a direct computation shows that
\begin{equation}
	\begin{split}
		&\left[J\left(\psi^{(a)}_0\right),e^{(b)}_0\right]-J\left(\left[\psi^{(a)}_0,e^{(b)}_0\right]\right)\\
		=&\frac{\epsilon_1+\epsilon_2}{2}\left[\sum_{\alpha\in\varDelta_+}\left(\alpha,\alpha^{(a)}\right)\sum_{k=1}^{\dim\mathfrak{g}_\alpha}f^{(\alpha,k)}e^{(\alpha,k)},e^{(a)}_0\right]+\frac{\sigma^{ba}_1-\sigma^{ab}_1}{2}\left(\alpha^{(a)},\alpha^{(b)}\right)e^{(b)}_0-\left(\alpha^{(a)},\alpha^{(b)}\right)w^{(b)}_+.
	\end{split}
\end{equation}
To simplify this, we need a very useful lemma \cite[Lemma 18.4.1]{musson2012lie}:
\begin{lemma}
	For $z\in\mathfrak{g}_{\beta-\alpha}$, we have
	\begin{equation}
		\sum_k\left[x^{(\beta,k)}_-,z\right]\otimes x^{(\beta,k)}_+=\sum_kx^{(\alpha,k)}_-\otimes\left[z,x^{(\alpha,k)}_+\right]
	\end{equation}
    in $\mathfrak{g}\otimes\mathfrak{g}$ and
    \begin{equation}
    	\sum_k\left[x^{(\beta,k)}_-,z\right]x^{(\beta,k)}_+=\sum_kx^{(\alpha,k)}_-\left[z,x^{(\alpha,k)}_+\right]
    \end{equation}
    in $U(\mathfrak{g})$.\label{betaalphalemma}
\end{lemma}
Then
\begin{equation}
	\begin{split}
		&\left[\sum_{\alpha\in\varDelta_+}\left(\alpha,\alpha^{(a)}\right)\sum_{k=1}^{\dim\mathfrak{g}_\alpha}f^{(\alpha,k)}e^{(\alpha,k)},e^{(a)}_0\right]\\
		=&\sum_{\alpha\in\varDelta_+}\left(\alpha,\alpha^{(a)}\right)\sum_{k=1}^{\dim\mathfrak{g}_\alpha}\left(f^{(\alpha,k)}\left[e^{(\alpha,k)},e^{(b)}_0\right]+(-1)^{|a||b|}\left[f^{(\alpha,k)},e^{(b)}_0\right]e^{(\alpha,k)}\right)
	\end{split}
\end{equation}
is equal to
\begin{equation}
	\begin{split}
		&\sum_{\alpha\in\varDelta_+\backslash\left\{\alpha^{(b)}\right\}}\left(\alpha,\alpha^{(a)}\right)\sum_{k=1}^{\dim\mathfrak{g}_\alpha}f^{(\alpha,k)}\left[e^{(\alpha,k)},e^{(b)}_0\right]\\
		&+(-1)^{|\beta||b|}\sum_{\beta\in\varDelta_+\backslash\left\{\alpha^{(b)}\right\}}\left(\beta-\alpha^{(b)},\alpha^{(a)}\right)\sum_{k=1}^{\dim\mathfrak{g}_\beta}f^{(\beta,k)}\left[e^{(b)}_0,e^{(\beta,k)}\right]+(\alpha^{(b)},\alpha^{(a)})\left[f^{(b)}_0,e^{(b)}_0\right]e^{(b)}_0.
	\end{split}
\end{equation}
Replacing the letter $\beta$ in the second sum with $\alpha$ and using the $ef$ relation in the last term, this is simplified to
\begin{equation}
	\begin{split}
		\left(\alpha^{(a)},\alpha^{(b)}\right)\sum_{\alpha\in\varDelta_+}\sum_{k=1}^{\dim\mathfrak{g}_\alpha}f^{(\alpha,k)}\left[e^{(\alpha,k)},e^{(b)}_0\right]-\left(\alpha^{(a)},\alpha^{(b)}\right)\psi^{(b)}_0e^{(b)}_0.
	\end{split}
\end{equation}
Therefore, we have
\begin{equation}
	\begin{split}
		\left(\alpha^{(a)},\alpha^{(b)}\right)w^{(b)}_+=&\left(\alpha^{(a)},\alpha^{(b)}\right)\frac{\epsilon_1+\epsilon_2}{2}\left(\sum_{\alpha\in\varDelta_+}\sum_{k=1}^{\dim\mathfrak{g}_\alpha}f^{(\alpha,k)}\left[e^{(\alpha,k)},e^{(b)}_0\right]-\psi^{(b)}_0e^{(b)}_0\right)\\
		&-\frac{\sigma^{ba}_1-\sigma^{ab}_1}{2}\left(\alpha^{(a)},\alpha^{(b)}\right)e^{(b)}_0.
	\end{split}\label{findingw}
\end{equation}
Taking $b=a$, the requirement \eqref{wrequirement} on $w^{(a)}_+$ yields
\begin{equation}
	w^{(a)}_+=\frac{\epsilon_1+\epsilon_2}{2}\sum_{\alpha\in\varDelta_+}\sum_{k=1}^{\dim\mathfrak{g}_\alpha}f^{(\alpha,k)}\left[e^{(\alpha,k)},e^{(a)}_0\right]-\frac{\epsilon_1+\epsilon_2}{2}\psi^{(a)}_0e^{(a)}_0.\label{w+}
\end{equation}
Likewise,
\begin{equation}
	w^{(a)}_-=-\frac{\epsilon_1+\epsilon_2}{2}\sum_{\alpha\in\varDelta_+}\sum_{k=1}^{\dim\mathfrak{g}_\alpha}\left[f^{(a)}_0,f^{(\alpha,k)}\right]e^{(\alpha,k)}-\frac{\epsilon_1+\epsilon_2}{2}f^{(a)}_0\psi^{(a)}_0.\label{w-}
\end{equation}
Notice that in \eqref{findingw}, $\left(\alpha^{(a)},\alpha^{(b)}\right)$ is zero when taking $b=a$ for fermionic nodes. Nevertheless, we can check that the expressions for $w^{(a)}_{\pm}$ in \eqref{w+} and \eqref{w-} would still satisfy the requirement \eqref{wrequirement}.

With the expressions for $v^{(a)}$ and $w^{(a)}_{\pm}$, we can write the commutation relations for the generators in the $J$ presentation. Similar relations can also be found in \cite{guay2018coproduct,ueda2019affine} for similar Yangian algebras. For brevity, we shall also define
\begin{equation}
	\widetilde{v}^{(a)}:=v^{(a)}+\frac{\epsilon_1+\epsilon_2}{2}\left(\psi^{(a)}_0\right)^2=\frac{1}{2}(\epsilon_1+\epsilon_2)\sum_{\alpha\in\varDelta_+}\left(\alpha,\alpha^{(a)}\right)\sum_{k=1}^{\dim\mathfrak{g}_\alpha}f^{(\alpha,k)}e^{(\alpha,k)}.
\end{equation}
\begin{lemma}
	We have
	\begin{align}
		&\left[\psi^{(a)}_0,v^{(b)}\right]=0,\\
		&\left[\widetilde{v}^{(a)},e^{(b)}_0\right]=\left(\alpha^{(a)},\alpha^{(b)}\right)w^{(b)}_+,\\
		&\left[\widetilde{v}^{(a)},f^{(b)}_0\right]=-\left(\alpha^{(a)},\alpha^{(b)}\right)w^{(b)}_-,\\
		&\left[w^{(a)}_+,f^{(b)}_0\right]=\left[e^{(a)}_0,w^{(b)}_-\right]=\delta_{ab}v^{(a)},\\
		&\left[w^{(a)}_+,e^{(b)}_0\right]-\left[e^{(a)}_0,w^{(b)}_+\right]=-\frac{\epsilon_1+\epsilon_2}{2}\left(\alpha^{(a)},\alpha^{(b)}\right)\left[e^{(a)}_0,e^{(b)}_0\right],\\
		&\left[w^{(a)}_-,f^{(b)}_0\right]-\left[f^{(a)}_0,w^{(b)}_-\right]=\frac{\epsilon_1+\epsilon_2}{2}\left(\alpha^{(a)},\alpha^{(b)}\right)\left[f^{(a)}_0,f^{(b)}_0\right].
	\end{align}
\end{lemma}
\begin{proof}
	All the proofs follow from direct calculations with the use of Lemma \ref{betaalphalemma}. In particular, the second and third ones are essentially how we found $w^{(a)}_{\pm}$. Here, we would only explicitly write the calculation for the fourth one, which is the most lengthy, while the others should be much quicker.
	
	Let us consider the $e_0w_-$ relation, and the $w_+f_0$ relation would follow from the same argument. For convenience, we write
	\begin{equation}
		\begin{split}
			&-\frac{2}{\epsilon_1+\epsilon_2}\left[e^{(a)}_0,w^{(b)}_-\right]\\
			=&\sum_{\alpha\in\varDelta_+}\sum_{k=1}^{\dim\mathfrak{g}_\alpha}\left(e^{(a)}_0f^{(b)}_0f^{(\alpha,k)}e^{(\alpha,k)}-(-1)^{|\alpha||b|}e^{(a)}_0f^{(\alpha,k)}f^{(b)}_0e^{(\alpha,k)}-(-1)^{|a||b|}f^{(b)}_0f^{(\alpha,k)}e^{(\alpha,k)}e^{(a)}_0\right.\\
			&\left.+(-1)^{|\alpha||b|}(-1)^{|a||b|}f^{(\alpha,k)}f^{(b)}_0e^{(\alpha,k)}e^{(a)}_0\right)+\left[e^{(a)}_0,f^{(b)}_0\psi^{(b)}_0\right],
		\end{split}
	\end{equation}
    where we have expanded all the brackets (except the last term) for clarity. By adding the terms
    \begin{equation}
    	\mp(-1)^{|a||b|}f^{(b)}_0e^{(a)}_0f^{(\alpha,k)}e^{(\alpha,k)},\quad\mp(-1)^{|\alpha|(|a|+|b|)}f^{(\alpha,k)}f^{(b)}_0e^{(a)}_0e^{(\alpha,k)}
    \end{equation}
    inside the summations, we can group the terms into
    \begin{equation}
    	\begin{split}
    		&\sum_{\alpha\in\varDelta_+}\sum_{k=1}^{\dim\mathfrak{g}_\alpha}\left(\left[\delta_{ab}\psi^{(a)}_0,f^{(\alpha,k)}\right]e^{(\alpha,k)}+(-1)^{|a||b|}\left[f^{(b)}_0,\left[e^{(a)}_0,f^{(\alpha,k)}\right]\right]e^{(\alpha,k)}\right.\\
    		&\left.+(-1)^{|a||b|}(-1)^{|a||\alpha|}\left[f^{(b)}_0,f^{(\alpha,k)}\right]\left[e^{(a)}_0,e^{(\alpha,k)}\right]\right)+\left[e^{(a)}_0,f^{(b)}_0\psi^{(b)}_0\right].
    	\end{split}
    \end{equation}
    Picking out those with $\alpha=\alpha^{(a)}$ for the second and third terms in the summations, we have
    \begin{equation}
    	\begin{split}
    		&\sum_{\alpha\in\varDelta_+\backslash\left\{\alpha^{(a)}\right\}}\sum_{k=1}^{\dim\mathfrak{g}_\alpha}\left((-1)^{|a||b|}\left[f^{(b)}_0,\left[e^{(a)}_0,f^{(\alpha,k)}\right]\right]e^{(\alpha,k)}+(-1)^{|a||b|}\left[f^{(b)}_0,f^{(\alpha,k)}\right]\left[e^{(\alpha,k)},e^{(a)}_0\right]\right)\\
    		&-\sum_{\alpha\in\varDelta_+}\sum_{k=1}^{\dim\mathfrak{g}_\alpha}\left(\delta_{ab}\left(\alpha^{(a)},\alpha\right)f^{(\alpha,k)}e^{(\alpha,k)}\right)+(-1)^{|a||b|}\left[f^{(b)}_0,\left[e^{(a)}_0,f^{(a)}_0\right]\right]e^{(a)}_0+\left[e^{(a)}_0,f^{(b)}_0\psi^{(b)}_0\right].
    	\end{split}
    \end{equation}
     The first line vanishes using Lemma \ref{betaalphalemma} while the second line is simply $\frac{2}{\epsilon_1+\epsilon_2}\delta_{ab}v^{(a)}$. This completes the proof.
\end{proof}
From these relations, it is straightforward to get the following corollary by definitions of $J\left(\psi^{(a)}_0\right)$, $J\left(e^{(a)}_0\right)$ and $J\left(f^{(a)}_0\right)$.
\begin{corollary}\label{Jcor}
	We have
	\begin{align}
		&\left[\psi^{(a)}_0,J\left(X^{(b)}_0\right)\right]=J\left(\left[\psi^{(a)}_0,X^{(b)}_0\right]\right)\qquad(X=\psi,e,f),\\
		&\left[J\left(\psi^{(a)}_0\right),e^{(b)}_0\right]=\left(\alpha^{(a)},\alpha^{(b)}\right)J\left(e^{(b)}_0\right)+\frac{\sigma^{ba}_1-\sigma^{ab}_1}{2}\left(\alpha^{(a)},\alpha^{(b)}\right)e^{(b)}_0,\\
		&\left[J\left(\psi^{(a)}_0\right),f^{(b)}_0\right]=-\left(\alpha^{(a)},\alpha^{(b)}\right)J\left(f^{(b)}_0\right)-\frac{\sigma^{ba}_1-\sigma^{ab}_1}{2}\left(\alpha^{(a)},\alpha^{(b)}\right)f^{(b)}_0,\\
		&\left[J\left(e^{(a)}_0\right),f^{(b)}_0\right]=\left[e^{(a)}_0,J\left(f^{(b)}_0\right)\right]=\delta_{ab}J\left(\psi^{(a)}_0\right),\\
		&\left[J\left(e^{(a)}_0\right),e^{(b)}_0\right]-\left[e^{(a)}_0,J\left(e^{(b)}_0\right)\right]=\frac{1}{2}\left(\sigma^{ba}_1-\sigma^{ab}_1\right)\left[e^{(a)}_0,e^{(b)}_0\right],\\
		&\left[J\left(f^{(a)}_0\right),f^{(b)}_0\right]-\left[f^{(a)}_0,J\left(f^{(b)}_0\right)\right]=-\frac{1}{2}\left(\sigma^{ba}_1-\sigma^{ab}_1\right)\left[f^{(a)}_0,f^{(b)}_0\right],\\
		&\left[J\left(e^{(a)}_0\right),e^{(b)}_0\right]=\left[J\left(f^{(a)}_0\right),f^{(b)}_0\right]=0\qquad(\sigma_1^{ab}=0).
	\end{align}
\end{corollary}
Notice that the last line follows from
\begin{equation}
	J\left(e^{(a)}_0\right)=\frac{1}{\left(\alpha^{(c)},\alpha^{(a)}\right)}\left[J\left(\psi^{(c)}_0\right),e^{(a)}_0\right]-\frac{\sigma_1^{ac}-\sigma_1^{ca}}{2}e^{(a)}_0,
\end{equation}
and likewise for $f^{(a)}_0$. When $a$ is bosonic, we can take $c=a$. When $a$ is fermionic, $c$ can be taken as one of $a\pm1$ such that $\sigma_1^{cb}=0$. It is straightforward to see that each relation in this corollary is equivalent to one of the relations (involving non-zero modes) in (R) in Theorem \ref{minthm}.

It is also possible to write $J$ acting on any positive real roots besides the simple ones with the help of the Weyl group of the untwisted affine A-type superalgebra. Since $\dim\mathfrak{g}_\alpha=1$ for $\alpha\in\varDelta^\text{re}_+$, we shall omit the label $k$ in the corresponding elements. Due to the Serre relations, given an even simple root $\alpha^{(b)}$, the operator $\tau^{(b)}:=\exp\left(\text{ad}_{e^{(b)}_0}\right)\exp\left(-\text{ad}_{f^{(b)}_0}\right)\exp\left(\text{ad}_{e^{(b)}_0}\right)$ is well-defined and is an automorphism of the quiver Yangian (see for example \cite{kumar2012kac,serganova2011kac}). Following the same argument as in \cite[Lemma 3.17]{guay2018coproduct}, $\tau^{(b)}$ can be applied to $J\left(\psi^{(a)}_0\right)$, $J\left(e^{(a)}_0\right)$ and $J\left(f^{(a)}_0\right)$ for any simple root $\alpha^{(a)}$. Moreover, we find that
\begin{equation}
	\tau^{(b)}\left(J\left(\psi^{(a)}_0\right)\right)=J\left(\psi^{(a)}_0\right)-\frac{2\left(\alpha^{(b)},\alpha^{(a)}\right)}{\left(\alpha^{(b)},\alpha^{(b)}\right)}J\left(\psi^{(b)}_0\right)-\left(\sigma^{ba}_1-\sigma^{ab}_1\right)\frac{\left(\alpha^{(b)},\alpha^{(a)}\right)}{\left(\alpha^{(b)},\alpha^{(b)}\right)}\psi^{(b)}_0.
\end{equation}

Suppose a root $\alpha$ can be obtained from a simple root $\alpha^{(a)}$ under the even reflections $s^{(b)}$ via $\alpha=s^{(b_1)}\dots s^{(b_p)}\left(\alpha^{(a)}\right)$. Then we may write $e^{(\alpha)}=\tau^{(b_1)}\dots\tau^{(b_p)}\left(e^{(a)}\right)$ and define $J\left(e^{(\alpha)}\right):=\tau^{(b_1)}\dots\tau^{(b_p)}\left(J\left(e^{(a)}\right)\right)$ (and likewise for $f$).
\begin{proposition}
	For any positive real root $\alpha$ and $a\in Q_0$, we have
	\begin{equation}
		\begin{split}
			&\left[J\left(\psi^{(a)}_0\right),e^{(\alpha)}\right]=\left[\psi^{(a)}_0,J\left(e^{(\alpha)}\right)\right]+c^{a\alpha}e^{(\alpha)}=\left(\alpha^{(a)},\alpha\right)J\left(e^{(\alpha)}\right)+c^{a\alpha}e^{(\alpha)},\\
			&\left[J\left(\psi^{(a)}_0\right),f^{(\alpha)}\right]=\left[\psi^{(a)}_0,J\left(f^{(\alpha)}\right)\right]-c^{a\alpha}e^{(\alpha)}=-\left(\alpha^{(a)},\alpha\right)J\left(f^{(\alpha)}\right)-c^{a\alpha}f^{(\alpha)},
		\end{split}
	\end{equation}
    where $c^{a\alpha}\in\frac{\epsilon_1-\epsilon_2}{2}\mathbb{Z}$.\label{Jprop}
\end{proposition}
\begin{proof}
	We will only show the first relation with $e^{(\alpha)}$ as the proof of the second one with $f^{(\alpha)}$ follows in the same way. We prove this by induction. When $p=0$, this is automatically true since $\alpha$ is a simple root. Now suppose this holds for $\beta=s^{(b_2)}\dots s^{(b_p)}\left(\alpha^{(a)}\right)$. Then for $\alpha=s^{(b_1)}\left(\beta\right)$,
	\begin{equation}
		\begin{split}
			&\left[J\left(\psi^{(a)}_0\right),e^{\alpha}\right]\\
			=&\tau^{(b_1)}\left(\left[\left(\tau^{(b_1)}\right)^{-1}J\left(\psi^{(a)}_0\right),e^{(\beta)}\right]\right)\\
			=&\tau^{(b_1)}\left(\left[J\left(\psi^{(a)}_0\right)-\frac{2\left(\alpha^{(b_1)},\alpha^{(a)}\right)}{\left(\alpha^{(b_1)},\alpha^{(b_1)}\right)}J\left(\psi^{(b_1)}_0\right)-\left(\sigma^{b_1a}_1-\sigma^{ab_1}_1\right)\frac{\left(\alpha^{(b_1)},\alpha^{(a)}\right)}{\left(\alpha^{(b_1)},\alpha^{(b_1)}\right)}\psi^{(b_1)}_0,e^{(\beta)}\right]\right).
		\end{split}
	\end{equation}
    With the assumption of the proposition to hold for $\beta$, we have
    \begin{equation}
    	\left[J\left(\psi^{(a)}_0\right),e^{\alpha}\right]=\left(\alpha^{(a)},\alpha\right)J\left(e^{(\alpha)}\right)+c^{a\alpha}e^{(\alpha)},
    \end{equation}
    where
    \begin{equation}
    	c^{a\alpha}=\left(1-\frac{2\left(\alpha^{(b_1)},\alpha^{(a)}\right)}{\left(\alpha^{(b_1)},\alpha^{(b_1)}\right)}\right)c^{a\beta}+\left(\sigma^{b_1a}-\sigma^{ab_1}\right)\frac{\left(\alpha^{(b_1)},\alpha^{(a)}\right)}{\left(\alpha^{(b_1)},\alpha^{b_1}\right)}\left(\alpha^{(b_1)},\beta\right).
    \end{equation}
    Write $c^{a\beta}=\frac{\epsilon_1-\epsilon_2}{2}c'$, and hence $c'\in\mathbb{Z}$. Then
    \begin{equation}
    	c^{a\alpha}=\frac{\epsilon_1-\epsilon_2}{2}\left(\left(1-\frac{2\left(\alpha^{(b_1)},\alpha^{(a)}\right)}{\left(\alpha^{(b_1)},\alpha^{(b_1)}\right)}\right)c'\pm2\frac{\left(\alpha^{(b_1)},\alpha^{(a)}\right)}{\left(\alpha^{(b_1)},\alpha^{b_1}\right)}\left(\alpha^{(b_1)},\beta\right)\right).
    \end{equation}
    In particular, the term in the biggest bracket is an integer. The other equality involving the commutator of $\psi^{(a)}_0$ and $J\left(e^{(\alpha)}\right)$ in this relation can be proven in a similar way.
\end{proof}
As a result, $J\left(e^{(\alpha)}\right)$ and $J\left(f^{(\alpha)}\right)$ are independent of the choice of the sequence of $\tau^{(b)}$ up to a constant multiple. From this proposition, it is also straightforward to obtain the following corollary.
\begin{corollary}
	For any positive real root $\alpha$ and $a\in Q_0$, we have
	\begin{equation}
		\begin{split}
			&\left(\alpha^{(b)},\alpha\right)\left[J\left(\psi^{(a)}_0\right),e^{(\alpha)}\right]-\left(\alpha^{(a)},\alpha\right)\left[J\left(\psi^{(b)}_0\right),e^{(\alpha)}\right]=c^{ab}_\alpha e^{(\alpha)},\\
			&\left(\alpha^{(b)},\alpha\right)\left[J\left(\psi^{(a)}_0\right),e^{(\alpha)}\right]-\left(\alpha^{(a)},\alpha\right)\left[J\left(\psi^{(b)}_0\right),e^{(\alpha)}\right]=c^{ab}_\alpha e^{(\alpha)},
		\end{split}
	\end{equation}
    where $c^{ab}_\alpha=\left(\alpha^{(b)},\alpha\right)c^{a\alpha}-\left(\alpha^{(a)},\alpha\right)c^{b\alpha}$.
\end{corollary}
Incidentally, following the similar proof as in \cite{ueda2019affine}, we also have
\begin{equation}
	\left[J\left(\psi^{(a)}_0\right),\widetilde{v}^{(b)}\right]+\left[J\left(\psi^{(b)}_0\right),\widetilde{v}^{(a)}\right]=\left[J\left(\psi^{(a)}_0\right),J\left(\psi^{(b)}_0\right)\right]+\left[\widetilde{v}^{(a)},\widetilde{v}^{(b)}\right]=0.
\end{equation}

Now, we are prepared to write our coproduct of the quiver Yangians. Recall that in general, $(x\otimes y)(z\otimes w)=(-1)^{|y||z|}(xz)\otimes(yw)$. For brevity, let us write a linear operator $\square(x):=x\otimes1+1\otimes x$ and define a Casimir element
\begin{equation}
	\Omega_-:=\sum_{\alpha\in\varDelta_+}\sum_{k=1}^{\dim\mathfrak{g}_\alpha}f^{(\alpha,k)}\otimes e^{(\alpha,k)}.
\end{equation}
It is straightforward to get the following commutation relations:
\begin{lemma}
	We have
	\begin{equation}
		\left[\square\left(\psi^{(a)}_r\right),\Omega_-\right]=0,\quad\left[\square\left(e^{(a)}_0\right),\Omega_-\right]=\psi^{(a)}_0\otimes e^{(a)}_0,\quad\left[\square\left(f^{(a)}_0\right),\Omega_-\right]=-f^{(a)}_0\otimes\psi^{(a)}_0.
	\end{equation}
\end{lemma}
\begin{proof}
	The relations directly follow the definitions of the operators. Here, we will only explicitly show the second one as an illustration. This can be seen by
	\begin{equation}
		\begin{split}
			\left[e^{(a)}\otimes1,\Omega_-\right]=&\sum_{\alpha\in\varDelta_+}\sum_{k=1}^{\dim\mathfrak{g}_\alpha}\left[e^{(a)}_0,f^{(\alpha,k)}\right]\otimes e^{(\alpha,k)}\\
			=&\psi^{(a)}_0\otimes e^{(a)}_0+\sum_{\alpha\in\varDelta_+\backslash\left\{\alpha^{(a)}\right\}}\sum_{k=1}^{\dim\mathfrak{g}_{\alpha-\alpha^{(a)}}}f^{\left(\alpha-\alpha^{(a)},k\right)}\otimes\left[e^{\left(\alpha-\alpha^{(a)},k\right)},e^{(a)}_0\right]\\
			=&\psi^{(a)}_0\otimes e^{(a)}_0-\left[1\otimes e^{(a)}_0,\Omega_-\right],
		\end{split}
	\end{equation}
    where we have used Lemma \ref{betaalphalemma} in the second equality.
\end{proof}

Let us also introduce another linear map $\Delta$ defined by
\begin{equation}
	\begin{split}
		&\Delta\left(\psi^{(a)}_0\right)=\square\left(\psi^{(a)}_0\right),\quad\Delta\left(e^{(a)}_0\right)=\square\left(e^{(a)}_0\right),\quad\Delta\left(f^{(a)}_0\right)=\square\left(f^{(a)}_0\right),\\
		&\Delta\left(\psi^{(a)}_1\right)=\square\left(\psi^{(a)}_1\right)+(\epsilon_1+\epsilon_2)\psi^{(a)}_0\otimes\psi^{(a)}_0+(\epsilon_1+\epsilon_2)\left[\psi^{(a)}_0\otimes1,\Omega_-\right]\\
		&\qquad\qquad=\square\left(\psi^{(a)}_1\right)+(\epsilon_1+\epsilon_2)\psi^{(a)}_0\otimes\psi^{(a)}_0-(\epsilon_1+\epsilon_2)\sum_{\alpha\in\varDelta_+^\text{re}}\left(\alpha^{(a)},\alpha\right)f^{(\alpha)}\otimes e^{(\alpha)}.
	\end{split}\label{Delta}
\end{equation}
Notice that this uniquely determines $\Delta$ as the actions on all modes can be obtained following the discussions in \S\ref{minimalistic}. For instance,
\begin{proposition}
	We have
	\begin{align}
		&\Delta\left(\widetilde{\psi}^{(a)}_1\right)=\square\left(\widetilde{\psi}^{(a)}_1\right)+(\epsilon_1+\epsilon_2)\left[\psi^{(a)}_0\otimes1,\Omega_-\right],\\
		&\Delta\left(e^{(a)}_1\right)=\square\left(e^{(a)}_1\right)-(\epsilon_1+\epsilon_2)\left[\Omega_-,e^{(a)}_0\otimes1\right],\\
		&\Delta\left(f^{(a)}_1\right)=\square\left(f^{(a)}_1\right)+(\epsilon_1+\epsilon_2)\left[\Omega_-,1\otimes f^{(a)}_0\right].
	\end{align}\label{Delta1prop}
\end{proposition}
\begin{proof}
	The first one follows from a straightforward calculation using the linearity of $\Delta$ and \eqref{Delta}. The second and third ones can be shown in the same way, and we will only explicitly verify the one for $e^{(a)}_1$ here. When $|a|=0$, by \eqref{finitegenerators} (with $b=a$), we have
	\begin{equation}
		\Delta\left(e^{(a)}_1\right)=\frac{1}{\left(\alpha^{(a)},\alpha^{(a)}\right)}\Delta\left(\left[\widetilde{\psi}^{(a)}_1,e^{(a)}_0\right]\right)=\square\left(e^{(a)}_1\right)+\frac{\epsilon_1+\epsilon_2}{\left(\alpha^{(a)},\alpha^{(a)}\right)}\left[\left[\psi^{(a)}_0\otimes1,\Omega_-\right],\square\left(e^{(a)}_0\right)\right],
	\end{equation}
    where we have applied the result of $\Delta\left(\widetilde{\psi}^{(a)}_0\right)$ in the second equality. Using the Jacobi identity $[x,[y,z]]=[[x,y],z]+(-1)^{|x||y|}[y,[x,z]]$, this is equal to
    \begin{equation}
    	\begin{split}
    		&\square\left(e^{(a)}_1\right)-\frac{\epsilon_1+\epsilon_2}{\left(\alpha^{(a)},\alpha^{(a)}\right)}\left(\left[\psi^{(a)}_0\otimes1,\left[\Omega_-,\square\left(e^{(a)}_0\right)\right]\right]-\left[\Omega_-,\left[\psi^{(a)}_0\otimes1,\square\left(e^{(a)}_0\right)\right]\right]\right)\\
    		=&\square\left(e^{(a)}_1\right)-\frac{\epsilon_1+\epsilon_2}{\left(\alpha^{(a)},\alpha^{(a)}\right)}\left(\left[\psi^{(a)}_0\otimes1,-\psi^{(a)}_0\otimes e^{(a)}_0\right]-\left[\Omega_-,\left[\psi^{(a)}_0\otimes1,e^{(a)}_0\otimes1\right]\right]\right)\\
    		=&\square\left(e^{(a)}_1\right)-(\epsilon_1+\epsilon_2)\left[\Omega_-,e^{(a)}\otimes1\right].
    	\end{split}
    \end{equation}
    When $|a|=1$, we take $b=a+1$ in \eqref{finitegenerators}, and a similar calculation leads to the same result.
\end{proof}

This operator in fact gives a coproduct of the quiver Yangian.
\begin{theorem}
	For $M+N>2$, $MN\neq2$ and $M\neq N$, the map $\Delta:\mathtt{Y}\rightarrow\mathtt{Y}\widehat{\otimes}\mathtt{Y}$ specified by \eqref{Delta} is a coassociative algebra homomorphism.\label{coprodthm}
\end{theorem}
\begin{proof}
	Let us first show that $\Delta$ is compatible with the relations (R) and (S) in Theorem \ref{minthm}. This is immediate for those involving only the zero modes $\psi^{(a)}_0$, $e^{(a)}_0$ and $f^{(a)}_0$. For the $\psi_0e_1$ relation \eqref{psi0emin} and the $\psi_0f_1$ relation \eqref{psi0fmin} (with $r=1$), they can be checked following straightforward computations using Proposition \ref{Delta1prop}. For the $e_1f_0$ relation, we have
	\begin{equation}
		\left[\Delta\left(e^{(a)}_1\right),\Delta\left(f^{(a)}_0\right)\right]=\delta_{ab}\square\left(\psi^{(a)}_1\right)-(\epsilon_1+\epsilon_2)\left[\left[e^{(a)}_0\otimes1,\Omega_-\right],\square\left(f^{(a)}_0\right)\right],
	\end{equation}
    where we have used Proposition \ref{Delta1prop} and \eqref{efmin}. As before, using the Jacobi identity, this becomes
    \begin{equation}
    	\begin{split}
    		&\delta_{ab}\square\left(\psi^{(a)}_1\right)+(\epsilon_1+\epsilon_2)\left(\left[e^{(a)}_0\otimes1,f^{(b)}_0\otimes\psi^{(b)}_0\right]+\delta_{ab}\left[\psi^{(a)}\otimes1,\Omega_-\right]\right)\\
    		=&\delta_{ab}\square\left(\psi^{(a)}_1\right)+(\epsilon_1+\epsilon_2)\left(\delta_{ab}\psi^{(a)}_0\otimes\psi^{(a)}_0+\delta_{ab}\left[\psi^{(a)}\otimes1,\Omega_-\right]\right)\\
    		=&\delta_{ab}\Delta\left(\psi^{(a)}_1\right).
    	\end{split}
    \end{equation}
    The $e_0f_1$ relation can be verified in the same way. Next, let us verify the $\psi_1e_0$ relation, which will be similar for $\psi_1f_0$. For convenience, let us consider the equivalent \eqref{psitildeefmin} with $\widetilde{\psi}^{(a)}_1$. Using Proposition \ref{Delta1prop} and the Jacobi identity, we find
    \begin{equation}
    	\begin{split}
    		&\left[\Delta\left(\widetilde{\psi}^{(a)}_1\right),\Delta\left(e^{(b)}_0\right)\right]\\
    		=&\left(\alpha^{(a)},\alpha^{(b)}\right)\square\left(e^{(b)}_1\right)+\frac{\sigma^{ba}_1-\sigma^{ab}_1}{2}\left(\alpha^{(a)},\alpha^{(b)}\right)\square\left(e^{(b)}_1\right)\\
    		&+(\epsilon_1+\epsilon_2)\left(\left[\psi^{(a)}_0\otimes1,\left[\Omega_-,\square\left(e^{(b)}_0\right)\right]\right]-\left[\Omega_-,\left[\psi^{(a)}_0\otimes1,\square\left(e^{(b)}_0\right)\right]\right]\right)\\
    		=&\left(\alpha^{(a)},\alpha^{(b)}\right)\square\left(e^{(b)}_1\right)+\frac{\sigma^{ba}_1-\sigma^{ab}_1}{2}\left(\alpha^{(a)},\alpha^{(b)}\right)\square\left(e^{(b)}_1\right)-(\epsilon_1+\epsilon_2)\left(\alpha^{(a)},\alpha^{(b)}\right)\left[\Omega_-,e^{(b)}\otimes1\right]\\
    		=&\left(\alpha^{(a)},\alpha^{(b)}\right)\Delta\left(e^{(b)}_1\right)+\frac{\sigma^{ba}_1-\sigma^{ab}_1}{2}\left(\alpha^{(a)},\alpha^{(b)}\right)\square\left(e^{(b)}_0\right).
    	\end{split}
    \end{equation}
    The $ee$ and $ff$ relations can likewise be verified using the same method.
    
    Now the only remaining relation yet to be checked is the $\psi\psi$ relation. For the $\psi_1\psi_0$ case, the compatibility follows from the fact that $\square\left(\psi^{(a)}_r\right)$ commutes with $\Omega_-$. For the $\psi_1\psi_1$ case, it would again be more convenient to work with $\widetilde{\psi}^{(a)}_1$. Since
    \begin{equation}
    	\Delta\left(\widetilde{\psi}^{(a)}_1\right)=\square\left(\widetilde{\psi}^{(a)}_1\right)+(\epsilon_1+\epsilon_2)\left[\psi^{(a)}_0\otimes1,\Omega_-\right]\quad\text{and}\quad\square\left(\left[\widetilde{\psi}^{(a)}_1,\widetilde{\psi}^{(b)}_1\right]\right)=0,
    \end{equation}
    it is equivalent to showing that
    \begin{equation}
    	\begin{split}
    		&\left[\square\left(\widetilde{\psi}^{(a)}_1\right),\left[\psi^{(b)}_0\otimes1,\Omega_-\right]\right]+\left[\left[\psi^{(a)}_0\otimes1,\Omega_-\right],\square\left(\widetilde{\psi}^{(b)}_1\right)\right]\\
    		&+(\epsilon_1+\epsilon_2)\left[\left[\psi^{(a)}_0\otimes1,\Omega_-\right],\left[\psi^{(b)}_0\otimes1,\Omega_-\right]\right]
    	\end{split}\label{psi1psi1prf1}
    \end{equation}
    vanishes. In particular, after expanding the brackets and rearranging the terms, the second line in \eqref{psi1psi1prf1} becomes
    \begin{equation}
    	\begin{split}
    		\frac{1}{2}(\epsilon_1+\epsilon_2)\sum_{k>0}\sum_{\substack{\text{ht}(\alpha+\beta)=k\\ \alpha,\beta\in\varDelta_+^\text{re}}}&\left(\alpha^{(a)},\alpha\right)\left(\alpha^{(b)},\beta\right)(-1)^{|\alpha||\beta|}\left(\left\{f^{(\alpha)},f^{(\beta)}\right\}\otimes\left[e^{(\alpha)},e^{(\beta)}\right]\right.\\
    		&\left.+\left(\left[f^{(\alpha)},f^{(\beta)}\right]\otimes\left\{e^{(\alpha)},e^{(\beta)}\right\}\right)\right),
    	\end{split}
    \end{equation}
    where ht denotes the height of a root, and we have defined $\{x,y\}=xy+(-1)^{|x||y|}yx$ for brevity. For the first line in \eqref{psi1psi1prf1}, it equals
    \begin{equation}
    	\sum_{k>0}\left(\sum_{\substack{\text{ht}(\alpha)=k\\ \alpha\in\varDelta_+^\text{re}}}\left(\alpha^{(a)},\alpha\right)\left[\square\left(\widetilde{\psi}^{(b)}_1\right),f^{(\alpha)}\otimes e^{(\alpha)}\right]-\sum_{\substack{\text{ht}(\alpha)=k\\ \alpha\in\varDelta_+^\text{re}}}\left(\alpha^{(b)},\alpha\right)\left[\square\left(\widetilde{\psi}^{(a)}_1\right),f^{(\alpha)}\otimes e^{(\alpha)}\right]\right).
    \end{equation}
    Using Proposition \ref{Jprop}, this becomes
    \begin{equation}
    	\begin{split}
    		\sum_{k>0}\sum_{\substack{\text{ht}(\alpha)=k\\ \alpha\in\varDelta_+^\text{re}}}&\left(\left[\left(\alpha^{(a)},\alpha\right)\widetilde{v}^{(b)}_1-\left(\alpha^{(b)},\alpha\right)\widetilde{v}^{(a)}_1,f^{(\alpha)}\right]\otimes e^{(\alpha)}\right.\\
    		&\left.+\left[\left(\alpha^{(a)},\alpha\right)\widetilde{v}^{(b)}_1-\left(\alpha^{(b)},\alpha\right)\widetilde{v}^{(a)}_1,e^{(\alpha)}\right]\otimes f^{(\alpha)}\right).
    	\end{split}
    \end{equation}
    Now, we claim that for each $k$, we have
    \begin{equation}
    	\begin{split}
    		&\frac{1}{2}(\epsilon_1+\epsilon_2)\sum_{\substack{\text{ht}(\alpha+\beta)=k\\ \alpha,\beta\in\varDelta_+^\text{re}}}\left(\alpha^{(a)},\alpha\right)\left(\alpha^{(b)},\beta\right)(-1)^{|\alpha||\beta|}\left(\left\{f^{(\alpha)},f^{(\beta)}\right\}\otimes\left[e^{(\alpha)},e^{(\beta)}\right]\right)\\
    		=&\sum_{\substack{\text{ht}(\alpha)=k\\ \alpha\in\varDelta_+^\text{re}}}\left(\left[\left(\alpha^{(b)},\alpha\right)\widetilde{v}^{(a)}_1-\left(\alpha^{(a)},\alpha\right)\widetilde{v}^{(b)}_1,f^{(\alpha)}\right]\otimes e^{(\alpha)}\right)
    	\end{split}\label{psi1psi1prf2}
    \end{equation}
    and
    \begin{equation}
    	\begin{split}
    		&\frac{1}{2}(\epsilon_1+\epsilon_2)\sum_{\substack{\text{ht}(\alpha+\beta)=k\\ \alpha,\beta\in\varDelta_+^\text{re}}}\left(\alpha^{(a)},\alpha\right)\left(\alpha^{(b)},\beta\right)(-1)^{|\alpha||\beta|}\left(\left[f^{(\alpha)},f^{(\beta)}\right]\otimes\left\{e^{(\alpha)},e^{(\beta)}\right\}\right)\\
    		=&\sum_{\substack{\text{ht}(\alpha)=k\\ \alpha\in\varDelta_+^\text{re}}}\left(\left[\left(\alpha^{(b)},\alpha\right)\widetilde{v}^{(a)}_1-\left(\alpha^{(a)},\alpha\right)\widetilde{v}^{(b)}_1,e^{(\alpha)}\right]\otimes f^{(\alpha)}\right).
    	\end{split}\label{psi1psi1prf3}
    \end{equation}
    These two equalities would then imply \eqref{psi1psi1prf1} being zero. Here, we will only write the proof for \eqref{psi1psi1prf2} explicitly while \eqref{psi1psi1prf3} can be shown in a similar way.
    
    We first notice that by definition of $\widetilde{v}^{(a)}$,
    \begin{equation}
    	\begin{split}
    		&\left[\left(\alpha^{(b)},\alpha\right)\widetilde{v}^{(a)}_1-\left(\alpha^{(a)},\alpha\right)\widetilde{v}^{(b)}_1,f^{(\alpha)}\right]\\
    		=&\frac{1}{2}(\epsilon_1+\epsilon_2)\sum_{\beta\in\varDelta_+^{\text{re}}}\left(\left(\alpha^{(a)},\beta\right)\left(\alpha^{(b)},\alpha\right)-\left(\alpha^{(b)},\beta\right)\left(\alpha^{(a)},\alpha\right)\right)\left(\left[f^{(\alpha)},f^{(\beta)}\right]e^{(\beta)}+f^{(\beta)}\left[e^{(\alpha)},f^{(\alpha)}\right]\right).
    	\end{split}\label{psi1psi1prf4}
    \end{equation}
    When $\gamma:=\beta-\alpha\in\mathbb{\varDelta}_+^\text{re}$, we have
    \begin{equation}
    	f^{(\beta)}\left[e^{(\beta)},f^{(\alpha)}\right]=\left(f^{(\gamma)},\left[e^{(\beta)},f^{(\alpha)}\right]\right)f^{(\beta)}e^{(\gamma)}.
    \end{equation}
    When $\gamma:=\alpha-\beta\in\mathbb{\varDelta}_+^\text{re}$, we have
    \begin{equation}
    	f^{(\beta)}\left[e^{(\beta)},f^{(\alpha)}\right]=\left(e^{(\gamma)},\left[e^{(\beta)},f^{(\alpha)}\right]\right)f^{(\beta)}f^{(\gamma)}=-(-1)^{|\beta||\gamma|}\left(\left[e^{(\beta)},e^{(\gamma)}\right],f^{(\alpha)}\right)f^{(\beta)}f^{(\gamma)},
    \end{equation}
    where we have used the invariance property $([x,y],z)=(x,[y,z])$ in the second equality. Therefore,
    \begin{equation}
    	\begin{split}
    		&\sum_{\beta\in\varDelta_+^{\text{re}}}\left(\left(\alpha^{(a)},\beta\right)\left(\alpha^{(b)},\alpha\right)-\left(\alpha^{(b)},\beta\right)\left(\alpha^{(a)},\alpha\right)\right)f^{(\beta)}\left[e^{(\alpha)},f^{(\alpha)}\right]\\
    		=&\sum_{\substack{\beta,\gamma\in\varDelta_+^\text{re}\\ \beta-\alpha=\gamma}}\left(\left(\alpha^{(a)},\beta\right)\left(\alpha^{(b)},\beta-\gamma\right)-\left(\alpha^{(b)},\beta\right)\left(\alpha^{(a)},\beta-\gamma\right)\right)\left(f^{(\gamma)},\left[e^{(\beta)},f^{(\alpha)}\right]\right)f^{(\beta)}e^{(\gamma)}\\
    		&-\sum_{\substack{\beta,\gamma\in\varDelta_+^\text{re}\\ \alpha-\beta=\gamma}}\left(\left(\alpha^{(a)},\beta\right)\left(\alpha^{(b)},\beta+\gamma\right)-\left(\alpha^{(b)},\beta\right)\left(\alpha^{(a)},\beta+\gamma\right)\right)(-1)^{|\beta||\gamma|}\left(\left[e^{(\beta)},e^{(\gamma)}\right],f^{(\alpha)}\right)f^{(\beta)}f^{(\gamma)}\\
    		=&\sum_{\substack{\beta,\gamma\in\varDelta_+^\text{re}\\ \beta-\alpha=\gamma}}\left(\left(\alpha^{(a)},\gamma\right)\left(\alpha^{(b)},\beta\right)-\left(\alpha^{(b)},\gamma\right)\left(\alpha^{(a)},\beta\right)\right)\left(f^{(\gamma)},\left[e^{(\beta)},f^{(\alpha)}\right]\right)f^{(\beta)}e^{(\gamma)}\\
    		&-\sum_{\substack{\beta,\gamma\in\varDelta_+^\text{re}\\ \alpha-\beta=\gamma}}\left(\left(\alpha^{(a)},\beta\right)\left(\alpha^{(b)},\gamma\right)-\left(\alpha^{(b)},\beta\right)\left(\alpha^{(a)},\gamma\right)\right)(-1)^{|\beta||\gamma|}\left(\left[e^{(\beta)},e^{(\gamma)}\right],f^{(\alpha)}\right)f^{(\beta)}f^{(\gamma)},
    	\end{split}
    \end{equation}
    where the terms involving imaginary roots have vanishing coefficients. Likewise, the first part in the summation of \eqref{psi1psi1prf4} becomes
    \begin{equation}
    	-\sum_{\substack{\beta,\gamma\in\varDelta_+^\text{re}\\ \beta-\alpha=\gamma}}\left(\left(\alpha^{(a)},\gamma\right)\left(\alpha^{(b)},\beta\right)-\left(\alpha^{(b)},\gamma\right)\left(\alpha^{(a)},\beta\right)\right)\left(f^{(\gamma)},\left[e^{(\beta)},f^{(\alpha)}\right]\right)f^{(\beta)}e^{(\gamma)}.
    \end{equation}
    Thus,
    \begin{equation}
    	\begin{split}
    		&\left[\left(\alpha^{(b)},\alpha\right)\widetilde{v}^{(a)}_1-\left(\alpha^{(a)},\alpha\right)\widetilde{v}^{(b)}_1,f^{(\alpha)}\right]\otimes e^{(\alpha)}\\
    		=&-\frac{1}{2}(\epsilon_1+\epsilon_2)\sum_{\substack{\beta,\gamma\in\varDelta_+^\text{re}\\ \alpha-\beta=\gamma}}\left(\left(\alpha^{(a)},\beta\right)\left(\alpha^{(b)},\gamma\right)-\left(\alpha^{(b)},\beta\right)\left(\alpha^{(a)},\gamma\right)\right)\\
    		&\qquad\qquad\qquad(-1)^{|\beta||\gamma|}\left(\left[e^{(\beta)},e^{(\gamma)}\right],f^{(\alpha)}\right)f^{(\beta)}f^{(\gamma)}\otimes e^{(\alpha)}\\
    		=&-\frac{1}{2}(\epsilon_1+\epsilon_2)\sum_{\substack{\beta,\gamma\in\varDelta_+^\text{re}\\ \alpha-\beta=\gamma}}\left(\left(\alpha^{(a)},\beta\right)\left(\alpha^{(b)},\gamma\right)-\left(\alpha^{(b)},\beta\right)\left(\alpha^{(a)},\gamma\right)\right)(-1)^{|\beta||\gamma|}f^{(\beta)}f^{(\gamma)}\otimes\left[e^{(\beta)},e^{(\gamma)}\right].
    	\end{split}\label{psi1psi1prf5}
    \end{equation}
    This is also equal to
    \begin{equation}
    	-\frac{1}{2}(\epsilon_1+\epsilon_2)\sum_{\substack{\beta,\gamma\in\varDelta_+^\text{re}\\ \alpha-\beta=\gamma}}\left(-\left(\alpha^{(a)},\beta\right)\left(\alpha^{(b)},\gamma\right)+\left(\alpha^{(b)},\beta\right)\left(\alpha^{(a)},\gamma\right)\right)(-1)^{|\beta||\gamma|}(-1)^{|\beta||\gamma|}f^{(\beta)}f^{(\gamma)}\otimes\left[e^{(\gamma)},e^{(\beta)}\right].\label{psi1psi1prf6}
    \end{equation}
    Adding \eqref{psi1psi1prf5} and \eqref{psi1psi1prf6} together and then dividing it by 2, we get
    \begin{equation}
    	\begin{split}
    		&\sum_{\substack{\text{ht}(\alpha)=k\\ \alpha\in\varDelta_+^\text{re}}}\left[\left(\alpha^{(b)},\alpha\right)\widetilde{v}^{(a)}_1-\left(\alpha^{(a)},\alpha\right)\widetilde{v}^{(b)}_1,f^{(\alpha)}\right]\otimes e^{(\alpha)}\\
    		=&\frac{1}{2}(\epsilon_1+\epsilon_2)\sum_{\substack{\beta,\gamma\in\varDelta_+^\text{re}\\ \beta+\gamma=\alpha}}\frac{1}{2}\left(\left(\alpha^{(a)},\beta\right)\left(\alpha^{(b)},\gamma\right)-\left(\alpha^{(b)},\beta\right)\left(\alpha^{(a)},\gamma\right)\right)(-1)^{|\beta||\gamma|}\left\{f^{(\beta)},f^{(\gamma)}\right\}\otimes\left[e^{(\gamma)},e^{(\beta)}\right].
    	\end{split}
    \end{equation}
    To the first line in \eqref{psi1psi1prf2}, we also add itself with $\alpha$ and $\beta$ exchanged and then divide the result by 2. This leads to
    \begin{equation}
    	\begin{split}
    		&\frac{1}{2}(\epsilon_1+\epsilon_2)\sum_{\substack{\text{ht}(\alpha+\beta)=k\\ \alpha,\beta\in\varDelta_+^\text{re}}}\left(\alpha^{(a)},\alpha\right)\left(\alpha^{(b)},\beta\right)(-1)^{|\alpha||\beta|}\left(\left\{f^{(\alpha)},f^{(\beta)}\right\}\otimes\left[e^{(\alpha)},e^{(\beta)}\right]\right)\\
    		=&\frac{1}{2}(\epsilon_1+\epsilon_2)\sum_{\substack{\text{ht}(\alpha+\beta)=k\\ \alpha,\beta\in\varDelta_+^\text{re}}}\frac{1}{2}\left(\left(\alpha^{(a)},\beta\right)\left(\alpha^{(b)},\gamma\right)-\left(\alpha^{(b)},\beta\right)\left(\alpha^{(a)},\gamma\right)\right)(-1)^{|\beta||\gamma|}\left\{f^{(\beta)},f^{(\gamma)}\right\}\otimes\left[e^{(\gamma)},e^{(\beta)}\right],
    	\end{split}
    \end{equation}
    where the terms with $\alpha+\beta\in\varDelta_+^\text{im}$ have vanishing coefficients. This proves the compatibility of $\Delta$ with the $\psi_1\psi_1$ relation, and hence with the all the relations in Theorem \ref{minthm}. As a result, $\Delta$ is an algebra homomorphism.
    
    Now, let us check the coassociativity of $\Delta$. It is immediate for the zero modes, and we only need to verify it for $\psi^{(a)}_1$, or equivalently, $\widetilde{\psi}^{(a)}_1$:
    \begin{equation}
    	\begin{split}
    		(\text{id}\otimes\Delta)\left(\Delta\left(\widetilde{\psi}^{(a)}_1\right)\right)=&\widetilde{\psi}^{(a)}_1\otimes(1\otimes1)+1\otimes\left(1\otimes\widetilde{\psi}^{(a)}_1\right)+1\otimes\left(\widetilde{\psi}^{(a)}_1\otimes1\right)\\
    		&-1\otimes\left((\epsilon_1+\epsilon_2)\sum_{\alpha\in\Delta_+^{\text{re}}}\left(\alpha^{(a)},\alpha\right)\left(f^{(\alpha)}\otimes e^{(\alpha)}\right)\right)\\
    		&-(\epsilon_1+\epsilon_2)\sum_{\alpha\in\Delta_+^{\text{re}}}\left(\alpha^{(a)},\alpha\right)\left(f^{(\alpha)}\otimes\left(1\otimes e^{(\alpha)}\right)+f^{(\alpha)}\otimes\left(e^{(\alpha)}\otimes1\right)\right);
    	\end{split}
    \end{equation}
    \begin{equation}
    	\begin{split}
    		(\Delta\otimes\text{id})\left(\Delta\left(\widetilde{\psi}^{(a)}_1\right)\right)=&\left(1\otimes1\right)\otimes\widetilde{\psi}^{(a)}_1+\left(1\otimes\widetilde{\psi}^{(a)}_1\right)\otimes1+\left(\widetilde{\psi}^{(a)}_1\otimes1\right)\otimes1\\
    		&-\left((\epsilon_1+\epsilon_2)\sum_{\alpha\in\Delta_+^{\text{re}}}\left(\alpha^{(a)},\alpha\right)\left(f^{(\alpha)}\otimes e^{(\alpha)}\right)\right)\otimes1\\
    		&-(\epsilon_1+\epsilon_2)\sum_{\alpha\in\Delta_+^{\text{re}}}\left(\alpha^{(a)},\alpha\right)\left(\left(1\otimes f^{(\alpha)}\right)\otimes e^{(\alpha)}+\left(f^{(\alpha)}\otimes1\right)\otimes e^{(\alpha)}\right).
    	\end{split}
    \end{equation}
    Therefore, $(\text{id}\otimes\Delta)\circ\Delta=(\Delta\otimes\text{id})\circ\Delta$. This completes the proof of Theorem \ref{coprodthm}.
\end{proof}

\begin{remark}
	In terms of the generators $\uppsi^{(a)}_n=(\epsilon_1+\epsilon_2)\psi^{(a)}_n$, $\mathtt{e}^{(a)}_n=(\epsilon_1+\epsilon_2)^{1/2}e^{(a)}_n$, $\mathtt{f}^{(a)}_n=(\epsilon_1+\epsilon_2)^{1/2}f^{(a)}_n$ (and likewise for $\mathtt{e}^{(\alpha,k)}$, $\mathtt{f}^{(\alpha,k)}$), the coproduct in \eqref{Delta} reads
	\begin{equation}
		\begin{split}
			&\Delta\left(\uppsi^{(a)}_0\right)=\square\left(\uppsi^{(a)}_0\right),\quad\Delta\left(\mathtt{e}^{(a)}_0\right)=\square\left(\mathtt{e}^{(a)}_0\right),\quad\Delta\left(\mathtt{f}^{(a)}_0\right)=\square\left(\mathtt{f}^{(a)}_0\right),\\
			&\Delta\left(\uppsi^{(a)}_1\right)=\square\left(\uppsi^{(a)}_1\right)+\uppsi^{(a)}_0\otimes\uppsi^{(a)}_0-(\epsilon_1+\epsilon_2)\sum_{\alpha\in\varDelta_+^\text{re}}\left(\alpha^{(a)},\alpha\right)\mathtt{f}^{(\alpha)}\otimes\mathtt{e}^{(\alpha)}.
		\end{split}
	\end{equation}
\end{remark}

\section{Isomorphism of Quiver Yangians}\label{isoQY}
Given \emph{any} toric CY threefold, its quivers can be in different toric phases. Since these quivers are related by Seiberg/toric duality, it is natural to conjecture that their quiver Yangians are isomorphic. Here, we shall prove this for the generalized conifolds considered in this paper. As a result, different toric phases correspond to different triangulations of the toric diagram, i.e., different sequences $\varsigma$.

A special feature for these generalized conifold is their underlying Kac-Moody superalgebras, which are of untwisted affine A-type. The zero modes of the quiver Yangians are actually different sets of Chevalley generators. For any two inequivalent sets of Chevalley generators, one can reach one from the other by odd reflections and Weyl groupoids \cite{serganova1985automorphisms,hoyt2007classification,serganova2011kac}. We shall now extend this to isomorphisms of quiver Yangians. Recall that each quiver has an underlying Dynkin diagram associated to $\widehat{\mathfrak{sl}}_{M|N}$. In such case, the odd reflection corresponding to an odd simple root $\alpha^{(\digamma)}$ acts on any simple root $\alpha^{(a)}$ as\footnote{More generally (especially when the Kac-Moody superalgebra has non-isotropic odd roots), the odd reflection associated to an isotropic simple odd root has the same action with the second condition as $\alpha^{(a)}+\alpha^{(\digamma)}$ being a root.}
\begin{equation}
	\alpha'^{(a)}=r^{(\digamma)}\left(\alpha^{(a)}\right)=\begin{cases}
		-\alpha^{(a)},&\quad a=\digamma,\\
		\alpha^{(a)}+\alpha^{(\digamma)},&\quad a=\digamma\pm1,\\
		\alpha^{(a)},&\quad\text{otherwise}.
	\end{cases}
\end{equation}
The Cartan matrix $A=(A_{ab})$ where $A_{ab}=\left(\alpha^{(a)},\alpha^{(b)}\right)$ is mapped to $A'=RAR^\text{T}$ with $R=(R_{ab})$ given by
\begin{equation}
	R_{ab}=\begin{cases}
		-1,&\quad a=b=\digamma,\\
		1,&\quad a=b\neq\digamma,\\
		1,&\quad b=\digamma,A_{a\digamma}\neq0,\\
		0,&\quad\text{otherwise}.
	\end{cases}
\end{equation}
In terms of Dynkin diagrams, this manipulation changes the $\mathbb{Z}_2$-grading of the nodes (and hence their $e$, $f$ generators) connected to the node $\digamma$ (i.e., those with $A_{a\digamma}\neq0$) and leaves the remaining ones unchanged. The Chevalley generators are mapped to\footnote{It seems that there could often be typos when writing this transformation in literature. They would lead to inconsistency of signs in some of the relations. Therefore, we give an explicit proof of this in Appendix \ref{oddrefchevgen}.}
\begin{equation}
	\begin{split}
		&\psi'^{(a)}_0=\sum_{b=1}^{M+N}R_{ab}\psi^{(b)}_0=\begin{cases}
			-\psi^{(a)}_0,&\quad a=\digamma,\\
			\psi^{(a)}_0+\psi^{(\digamma)}_0,&\quad a=\digamma\pm1,\\
			\psi^{(a)},&\quad\text{otherwise};
		\end{cases}\\
		&e'^{(a)}_0=\begin{cases}
			f^{(a)}_0,&\quad a=\digamma,\\
			\left[e^{(\digamma)}_0,e^{(a)}_0\right],&\quad a=\digamma\pm1,\\
			e^{(a)}_0,&\quad\text{otherwise};
		\end{cases}
	    \qquad
		f'^{(a)}_0=\begin{cases}
			-e^{(a)}_0,&\quad a=\digamma,\\
			-\frac{1}{A_{a\digamma}}\left[f^{(a)}_0,f^{(\digamma)}_0\right],&\quad a=\digamma\pm1,\\
			f^{(a)}_0,&\quad\text{otherwise}.
		\end{cases}
	\end{split}\label{chevgenprimed}
\end{equation}
Notice that $A_{a\digamma}$ is simply $\pm1$.

It would also be useful to spell out the following lemmas.
\begin{lemma}
	We have
	\begin{equation}
		\left(\alpha'^{(a)},\alpha'^{(b)}\right)=\begin{cases}
			-\left(\alpha^{(a)},\alpha^{(b)}\right),&\quad (a,b)=(\digamma\pm1,\digamma),(\digamma,\digamma\pm1),\\
			\left(\alpha^{(a)},\alpha^{(a)}\right)+2\left(\alpha^{(\digamma)},\alpha^{(a)}\right),&\quad a=b=\digamma\pm1,\\
			\left(\alpha^{(a)},\alpha^{(b)}\right),&\quad\text{otherwise}.
		\end{cases}
	\end{equation}
\end{lemma}
\begin{proof}
	This can be directly shown by considering how the Cartan matrix would change under the odd reflection. More specifically, $A_{\digamma\digamma}$ changes from $\pm2$ to 0 (and vice versa) while $A_{a\digamma}=A_{\digamma a}$ changes from $\mp1$ to $\pm1$ for $a=\digamma\pm1$. The other elements in $A$ are invariant. This can also be checked in the proof in Appendix \ref{oddrefchevgen}.
\end{proof}
\begin{lemma}
	We have $\sigma_1'^{\digamma\pm1,\digamma}=-\sigma_1^{\digamma,\digamma\pm1}$ and $\sigma_1'^{\digamma,\digamma\pm1}=-\sigma_1^{\digamma\pm1,\digamma}$ while the other $\sigma_1^{ab}$ are invariant. Therefore, $\sigma_1'^{ab}-\sigma_1'^{ba}=\sigma_1^{ab}-\sigma_1^{ba}$.
\end{lemma}
\begin{proof}
	This can be directly shown by considering how the charge assignment in Figure \ref{epsilongencon}(b) would change under the odd reflection.
\end{proof}

Now, we are ready to show the algebra isomorphisms of the Seiberg/toric dual quiver Yangians.
\begin{theorem}
	For toric CY threefolds without compact divisors (except $xy=z^2w^2$), the quiver Yangians in different toric phases are isomorphic algebras.\label{QYisothm}
\end{theorem}
\begin{proof}
	When $M+N>2$ and $MN\neq2$, we shall show that the map induced by each odd reflection is an isomorphism\footnote{Here, we have also included the cases with $M=N$ which can be checked to follow the same isomorphic maps using the relations in Theorem \ref{minthm}. Notice that we need to exclude $xy=z^2w^2$ due to the subtlety mentioned at the end of \S\ref{minimalistic} though the non-reduced quiver Yangians in this case are still isomorphic.}. Then a quiver Yangian in any toric phase can be obtained from another under a sequence of such isomorphic maps. It suffices to check the relations involving non-zero modes. It would be more convenient to work with the $J$ presentation. This is then given by \eqref{chevgenprimed} and
	\begin{equation}
		\begin{split}
			&J\left(\psi'^{(a)}_0\right)=\begin{cases}
				-J\left(\psi^{(a)}_0\right),&\quad a=\digamma,\\
				J\left(\psi^{(a)}_0\right)+J\left(\psi^{(\digamma)}_0\right)-\frac{1}{2}\left(\sigma_1^{a\digamma}-\sigma_1^{\digamma a}\right)\psi^{(\digamma)}_0,&\quad a=\digamma\pm1,\\
				J\left(\psi^{(a)}_0\right),\quad\text{otherwise};
			\end{cases}\\
		    &J\left(e'^{(a)}_0\right)=\begin{cases}
		    	J\left(f^{(a)}_0\right),&\quad a=\digamma,\\
		    	\left[e^{(\digamma)}_0,J\left(e^{(a)}_0\right)\right],&\quad a=\digamma\pm1,\\
		    	J\left(e^{(a)}_0\right),&\quad\text{otherwise};
		    \end{cases}\\
		    &J\left(f'^{(a)}_0\right)=\begin{cases}
		    	-J\left(e^{(a)}_0\right),&\quad a=\digamma,\\
		    	-\frac{1}{A_{a\digamma}}\left[J\left(f^{(a)}_0\right),f^{(\digamma)}_0\right],&\quad a=\digamma\pm1,\\
		    	J\left(f^{(a)}_0\right),&\quad\text{otherwise}.
		    \end{cases}
		\end{split}
	\end{equation}
    It is sufficient to show that this is consistent with the relations in Corollary \ref{Jcor}. There is no need to check the last equation therein as it is derived from the previous ones. The $\psi'J(X')$ ($X=\psi,e,f$) relations are straightforward following linearity (and the Jacobi identity when $b=\digamma\pm1$). The remaining relations are verified below in Lemma \ref{primedJpsieflemma}$\sim$Lemma \ref{primedJeeJfflemma}. This then clearly also holds for its inverse.
	
	Notice that we have also included $\mathbb{C}^3$ ($\mathbb{C}\times\mathbb{C}^2/\mathbb{Z}_2$, conifold, resp. SPP) with $(M,N)=(1,0)$ ($(2,0)$, $(1,1)$, resp. $(2,1)$) which do not belong to the family we mainly focus on in this paper, as well as $\mathbb{C}^3/(\mathbb{Z}_2\times\mathbb{Z}_2)$ that is not a genralized conifold. Their isomorphisms are trivial since each of them only has one single toric phase.
\end{proof}

\begin{lemma}\label{primedJpsieflemma}
	We have
	\begin{align}
		&\left[J\left(\psi'^{(a)}_0\right),e'^{(b)}_0\right]=\left(\alpha'^{(a)},\alpha'^{(b)}\right)J\left(e'^{(b)}_0\right)+\frac{\sigma'^{ba}_1-\sigma'^{ab}_1}{2}\left(\alpha'^{(a)},\alpha'^{(b)}\right)e'^{(b)}_0,\\
		&\left[J\left(\psi'^{(a)}_0\right),f'^{(b)}_0\right]=-\left(\alpha'^{(a)},\alpha'^{(b)}\right)J\left(f'^{(b)}_0\right)-\frac{\sigma'^{ba}_1-\sigma'^{ab}_1}{2}\left(\alpha'^{(a)},\alpha'^{(b)}\right)f'^{(b)}_0.
	\end{align}
\end{lemma}
\begin{proof}
	Here, we shall only explicitly write the proof for the $J(\psi')e'$ relation, and the $J(\psi')f'$ relation follows in the same manner. This can be divided into the following cases:
	\begin{itemize}
		\item $a=b=\digamma\pm1$: We have
		\begin{equation}
			\begin{split}
				\left[J\left(\psi'^{(a)}_0\right),e'^{(b)}_0\right]=&\left[J\left(\psi^{(b)}\right),\left[e^{(\digamma)}_0,e^{(b)}_0\right]\right]+\left[J\left(\psi^{(\digamma)}\right),\left[e^{(\digamma)}_0,e^{(b)}_0\right]\right]\\
				&-\frac{1}{2}\left(\sigma_1^{b\digamma}-\sigma_1^{\digamma b}\right)\left[\psi^{(\digamma)}_0,\left[e^{(\digamma)}_0,e^{(b)}_0\right]\right].
			\end{split}
		\end{equation}
	    Using the Jacobi identity, this is equal to
	    \begin{equation}
	    	\begin{split}
	    		&\left[\left[J\left(\psi^{(b)}_0\right),e^{(\digamma)}_0\right],e^{(b)}_0\right]+\left[e^{(\digamma)}_0,\left[J\left(\psi^{(b)}_0\right),e^{(b)}_0\right]\right]+0+\left[e^{(\digamma)}_0,\left[J\left(\psi^{(\digamma)}_0\right),e^{(b)}_0\right]\right]\\
	    		&-0-\frac{1}{2}\left(\sigma_1^{b\digamma}-\sigma_1^{\digamma b}\right)\left[e^{(\digamma)}_0,\left[\psi^{(\digamma)}_0,e^{(b)}_0\right]\right].
	    	\end{split}
	    \end{equation}
        Using the $J(\psi)e$, $J(e)e$ and $\psi e$ relations, the terms with $\left(\sigma_1^{b\digamma}-\sigma_1^{\digamma b}\right)$ get cancelled and this becomes
        \begin{equation}
        	\begin{split}
        		&\left(\alpha^{(b)},\alpha^{(\digamma)}\right)\left[e^{(\digamma)}_0,J\left(e^{(b)}_0\right)\right]+\left(\alpha^{(b)},\alpha^{(b)}\right)\left[e^{(\digamma)}_0,J\left(e^{(b)}_0\right)\right]+\left(\alpha^{(b)},\alpha^{(\digamma)}\right)\left[e^{(\digamma)}_0,J\left(e^{(b)}_0\right)\right]\\
        		=&\left(\alpha'^{(a)},\alpha'^{(b)}\right)e'^{(b)}_0.
        	\end{split}
        \end{equation}
    
    \item $a=\digamma,b=\digamma\pm1$: We have
    \begin{equation}
    	\begin{split}
    		\left[J\left(\psi'^{(a)}_0\right),e'^{(b)}_0\right]=&-\left[J\left(\psi^{(\digamma)}_0\right),\left[e^{(\digamma)}_0,e^{(b)}_0\right]\right]=-\left[e^{(\digamma)}_0,\left[J\left(\psi^{(\digamma)}_0\right),e^{(b)}_0\right]\right]\\
    		=&-\left(\alpha^{(\digamma)},\alpha^{(b)}\right)\left[e^{(\digamma)}_0,J\left(e^{(b)}_0\right)\right]-\frac{1}{2}\left(\sigma_1^{b\digamma}-\sigma_1^{\digamma b}\right)\left(\alpha^{(\digamma)},\alpha^{(b)}\right)\left[e^{(\digamma)}_0,e^{(b)}_0\right]\\
    		=&\left(\alpha'^{(a)},\alpha'^{(b)}\right)J\left(e'^{(b)}_0\right)+\frac{1}{2}\left(\sigma_1'^{ba}-\sigma_1'^{ab}\right)\left(\alpha^{(\digamma)},\alpha^{(b)}\right)e'^{(b)}_0.
    	\end{split}
    \end{equation}

    \item $a=\digamma\mp1,b=\digamma\pm1$: This is similar to the case when $a=b=\digamma\pm1$. By using the Jacobi identity, as well as the $J(\psi)e$, $J(e)e$ and $\psi e$ relations, we have
    \begin{equation}
    	\begin{split}
    		\left[J\left(\psi'^{(a)}_0\right),e'^{(b)}_0\right]=&\left(\alpha^{(a)},\alpha^{(\digamma)}\right)\left[e^{(\digamma)}_0,J\left(e^{(b)}_0\right)\right]+\frac{1}{2}\left(\sigma_1^{b\digamma}-\sigma_1^{\digamma b}\right)\left(\alpha^{(a)},\alpha^{(\digamma)}\right)\left[e^{(\digamma)}_0,e^{(b)}_0\right]\\
    		&+\left(\alpha^{(b)},\alpha^{(\digamma)}\right)\left[e^{(\digamma)}_0,J\left(e^{(b)}_0\right)\right]-\frac{1}{2}\left(\sigma_1^{a\digamma}-\sigma_1^{\digamma a}\right)\left(\alpha^{(b)},\alpha^{(\digamma)}\right)\left[e^{(\digamma)}_0,e^{(b)}_0\right].
    	\end{split}
    \end{equation}
    This vanishes as $\left(\alpha^{(a)},\alpha^{(\digamma)}\right)=-\left(\alpha^{(b)},\alpha^{(\digamma)}\right)$ and $\left(\sigma_1^{a\digamma}-\sigma_1^{\digamma a}\right)=-\left(\sigma_1^{b\digamma}-\sigma_1^{\digamma b}\right)$.
    
    \item $a=b=\digamma$: We have $\left[J\left(\psi'^{(a)}_0\right),e'^{(b)}_0\right]=-\left[J\left(\psi^{(\digamma)}_0\right),f^{(\digamma)}_0\right]=0$.
    
    \item otherwise: The remaining cases are immediate following the expressions of the primed generators and the similar arguments as above.
	\end{itemize}
\end{proof}

\begin{lemma}\label{primedJeflemma}
	We have
	\begin{align}
		&\left[J\left(e'^{(a)}_0\right),f'^{(b)}_0\right]=\left[e'^{(a)}_0,J\left(f'^{(b)}_0\right)\right]=\delta_{ab}J\left(\psi'^{(a)}_0\right).
	\end{align}
\end{lemma}
\begin{proof}
	Here, we shall only explicitly write the proof for the $J(e')f'$ relation, and the $e'J(f')$ relation follows in the same manner. This can be divided into the following cases:
	\begin{itemize}
		\item $a=b=\digamma$: We have $\left[J\left(e'^{(a)}_0\right),f'^{(b)}_0\right]=-\left[J\left(f^{(\digamma)}_0\right),e^{(\digamma)}_0\right]=-J\left(\psi^{(\digamma)}_0\right)=J\left(\psi^{(a)}_0\right)$.
		
		\item $a=b=\digamma\pm1$: We have
		\begin{equation}
			\left[J\left(e'^{(a)}_0\right),f'^{(b)}_0\right]=-\frac{1}{A_{a\digamma}}\left[\left[e^{(\digamma)}_0,J\left(e^{(a)}_0\right)\right],\left[f^{(a)}_0,f^{(\digamma)}_0\right]\right].
		\end{equation}
	    Using the Jacobi identity, this becomes
	    \begin{equation}
	    	-\frac{1}{A_{a\digamma}}\left(\left[\left[\left[e^{(\digamma)}_0,J\left(e^{(a)}_0\right)\right],f^{(a)}_0\right],f^{(\digamma)}_0\right]+(-1)^{|a|(1+|a|)}\left[f^{(a)}_0,\left[\left[e^{(\digamma)}_0,J\left(e^{(a)}_0\right)\right],f^{(\digamma)}_0\right]\right]\right).\label{primedJefprf1}
	    \end{equation}
        In particular,
        \begin{equation}
        	\begin{split}
        		\left[\left[e^{(\digamma)}_0,J\left(e^{(a)}_0\right)\right],f^{(a)}_0\right]=&\left[e^{(\digamma)}_0\left[J\left(e^{(a)}_0\right),f^{(a)}_0\right]\right]-0=\left[e^{(\digamma)}_0,J\left(\psi^{(a)}_0\right)\right]\\
        		=&-\left(\alpha^{(a)},\alpha^{(\digamma)}\right)J\left(e^{(\digamma)}_0\right)-\frac{1}{2}\left(\sigma_1^{\digamma a}-\sigma_1^{a\digamma}\right)\left(\alpha^{(a)},\alpha^{(\digamma)}\right)e^{(\digamma)}_0.
        	\end{split}
        \end{equation}
        Likewise,
        \begin{equation}
        	\left[\left[e^{(\digamma)}_0,J\left(e^{(a)}_0\right)\right],f^{(\digamma)}_0\right]=(-1)^{|a|}\left(\alpha^{(a)},\alpha^{(\digamma)}\right)J\left(e^{(a)}_0\right).
        \end{equation}
        Therefore, \eqref{primedJefprf1} becomes
        \begin{equation}
        	\begin{split}
        		&\left[J\left(e^{(a)}_0\right),f^{(a)}_0\right]+\left[J\left(e^{(\digamma)}_0\right),f^{(\digamma)}_0\right]+\frac{1}{2}\left(\sigma_1^{\digamma a}-\sigma_1^{a\digamma}\right)\left[e^{(\digamma)},f^{(\digamma)}\right]\\
        		=&J\left(\psi^{(a)}_0\right)+J\left(\psi^{(\digamma)}_0\right)-\frac{1}{2}\left(\sigma_1^{a\digamma}-\sigma_1^{\digamma a}\right)\psi^{(\digamma)}_0.
        	\end{split}
        \end{equation}
        
        \item $a=\digamma,b=\digamma\pm1$: We have
        \begin{equation}
        	\begin{split}
        		\left[J\left(e'^{(a)}_0\right),f'^{(b)}_0\right]=&\frac{1}{A_{b\digamma}}\left[J\left(f^{(\digamma)}_0\right),\left[f^{(b)}_0,f^{(\digamma)}_0\right]\right]\\
        		=&\frac{1}{A_{b\digamma}}\left(\left[\left[J\left(f^{(\digamma)}_0\right),f^{(b)}_0\right],f^{(\digamma)}_0\right]+(-1)^{|b|}\left[f^{(b)}_0,\left[J\left(f^{(\digamma)}_0\right),f^{(\digamma)}_0\right]\right]\right).
        	\end{split}
        \end{equation}
        The second term in the parenthesis is zero as $\left[J\left(f^{(\digamma)}_0\right),f^{(\digamma)}_0\right]=0$ while using the $J(f)f$ relation, the first term becomes
        \begin{equation}
        	\left[\left[f^{(\digamma)}_0,J\left(f^{(b)}_0\right)\right],f^{(\digamma)}_0\right]-\frac{1}{2}\left(\sigma_1^{b\digamma}-\sigma_1^{\digamma b}\right)\left[\left[f^{(\digamma)}_0,f^{(b)}_0\right],f^{(\digamma)}_0\right].
        \end{equation}
        Both of the terms vanish as $\text{ad}_{f^{(\digamma)}_0}^2=0$. Therefore, $\left[J\left(e'^{(a)}_0\right),f'^{(b)}_0\right]=0$.
        
        \item $a=\digamma\mp1,b=\digamma\pm1$: This is similar to the case when $a=b=\digamma\pm1$. Using the Jacobi identity, $\left[J\left(e'^{(a)}_0\right),f'^{(b)}_0\right]$ becomes
        \begin{equation}
        	-\frac{1}{A_{a\digamma}}\left(\left[\left[\left[e^{(\digamma)}_0,J\left(e^{(a)}_0\right)\right],f^{(b)}_0\right],f^{(\digamma)}_0\right]+(-1)^{|b|(1+|a|)}\left[f^{(b)}_0,\left[\left[e^{(\digamma)}_0,J\left(e^{(a)}_0\right)\right],f^{(\digamma)}_0\right]\right]\right).
        \end{equation}
        In particular, $\left[\left[e^{(\digamma)}_0,J\left(e^{(a)}_0\right)\right],f^{(b)}_0\right]$ vanishes using the Jacobi identity. Moreover,
        \begin{equation}
        	\left[\left[e^{(\digamma)}_0,J\left(e^{(a)}_0\right)\right],f^{(\digamma)}_0\right]=-(-1)^{|a|}\left[J\left(e^{(a)}_0\right),\psi^{(k)}_0\right],
        \end{equation}
        which consists of two terms proportional to $J\left(e^{(a)}_0\right)$ and $e^{(a)}_0$ respectively. Hence, $\left[f^{(b)}_0,\left[\left[e^{(\digamma)}_0,J\left(e^{(a)}_0\right)\right],f^{(\digamma)}_0\right]\right]$ vanishes. Therefore, $\left[J\left(e'^{(a)}_0\right),f'^{(b)}_0\right]=0$.
        
        \item otherwise: The remaining cases are immediate following the expressions of the primed generators and the similar arguments as above.
	\end{itemize}
\end{proof}

\begin{lemma}\label{primedJeeJfflemma}
	We have
	\begin{align}
		&\left[J\left(e'^{(a)}_0\right),e'^{(b)}_0\right]-\left[e'^{(a)}_0,J\left(e'^{(b)}_0\right)\right]=\frac{1}{2}\left(\sigma'^{ba}_1-\sigma'^{ab}_1\right)\left[e'^{(a)}_0,e'^{(b)}_0\right],\\
		&\left[J\left(f'^{(a)}_0\right),f'^{(b)}_0\right]-\left[f'^{(a)}_0,J\left(f'^{(b)}_0\right)\right]=-\frac{1}{2}\left(\sigma'^{ba}_1-\sigma'^{ab}_1\right)\left[f'^{(a)}_0,f'^{(b)}_0\right].
	\end{align}
\end{lemma}
\begin{proof}
	Here, we shall only explicitly write the proof for the $J(e')e'$ relation, and the $J(f')f'$ relation follows in the same manner. This can be divided into the following cases:
	\begin{itemize}
		\item $a=b=\digamma\pm1$: We have
		\begin{equation}
			\begin{split}
				&\left[J\left(e'^{(a)}_0\right),e'^{(b)}_0\right]-\left[e'^{(a)}_0,J\left(e'^{(b)}_0\right)\right]\\
				=&\left[\left[e^{(\digamma)}_0,J\left(e^{(a)}_0\right)\right],\left[e^{(\digamma)}_0,e^{(a)}_0\right]\right]-\left[\left[e^{(\digamma)}_0,e^{(a)}_0\right],\left[e^{(\digamma)}_0,J\left(e^{(a)}_0\right)\right]\right].
			\end{split}
		\end{equation}
	    For $|e^{(a)}|=0$, the two terms are the same and hence cancel each other. For $|e^{(a)}|=1$, this is twice of the first term. Using the Jacobi identity, we get
	    \begin{equation}
	    	\begin{split}
	    		&\left[\left[e^{(\digamma)}_0,J\left(e^{(a)}_0\right)\right],\left[e^{(\digamma)}_0,e^{(a)}_0\right]\right]\\
	    		=&\left[\left[\left[e^{(\digamma)}_0,J\left(e^{(a)}_0\right)\right],e^{(\digamma)}_0\right],e^{(a)}_0\right]-\left[e^{(\digamma)}_0,\left[\left[e^{(\digamma)}_0,J\left(e^{(a)}_0\right)\right],e^{(a)}_0\right]\right].
	    	\end{split}
	    \end{equation}
        The first term vanishes as $\text{ad}_{e^{(\digamma)}_0}^2=0$. For the second term, using the $eJ(e)$ relation, $\left[e^{(\digamma)}_0,J\left(e^{(a)}_0\right)\right]$ consists of two terms proportional to $\left[J\left(e^{(\digamma)}_0\right),e^{(a)}_0\right]$ and $\left[e^{(\digamma)}_0,e^{(a)}_0\right]$ respectively. Hence, the second term also vanishes as $\text{ad}_{e^{(a)}_0}^2=0$. Therefore, $\left[J\left(e'^{(a)}_0\right),e'^{(b)}_0\right]-\left[e'^{(a)}_0,J\left(e'^{(b)}_0\right)\right]=0$.
        
        \item $a=\digamma,b=\digamma\pm1$: We have
        \begin{equation}
        	\left[J\left(e'^{(a)}_0\right),e'^{(b)}_0\right]-\left[e'^{(a)}_0,J\left(e'^{(b)}_0\right)\right]=\frac{1}{2}\left(\sigma_1^{b\digamma}-\sigma_1^{\digamma b}\right)\left(\alpha^{(b)},\alpha^{(\digamma)}\right)e^{(b)}_0
        \end{equation}
        using the Jacobi identity. On the other hand,
        \begin{equation}
        	\begin{split}
        		&\frac{1}{2}\left(\sigma'^{ba}_1-\sigma'^{ab}_1\right)\left[e'^{(a)}_0,e'^{(b)}_0\right]=\frac{1}{2}\left(\sigma^{b\digamma}_1-\sigma^{\digamma b}_1\right)\left[f^{(\digamma)}_0,\left[e^{(\digamma)}_0,e^{(b)}_0\right]\right]\\
        		=&\frac{1}{2}\left(\sigma^{b\digamma}_1-\sigma^{\digamma b}_1\right)\left[\psi^{(\digamma)}_0,e^{(b)}_0\right]=\frac{1}{2}\left(\sigma_1^{b\digamma}-\sigma_1^{\digamma b}\right)\left(\alpha^{(b)},\alpha^{(\digamma)}\right)e^{(b)}_0.
        	\end{split}
        \end{equation}
        
        \item $a=\digamma\pm1,b=\digamma\pm1$: We have
        \begin{equation}
        	\begin{split}
        		\left[J\left(e'^{(a)}_0\right),e'^{(b)}_0\right]-\left[e'^{(a)}_0,J\left(e'^{(b)}_0\right)\right]=&(-1)^{1+|a|}\left[e^{\digamma}_0,\left[\left[e^{(\digamma)}_0,J\left(e^{(a)}_0\right)\right],e^{(b)}_0\right]\right]\\
        		&-(-1)^{1+|a|}\left[e^{\digamma}_0,\left[\left[e^{(\digamma)}_0,e^{(a)}_0\right],J\left(e^{(b)}_0\right)\right]\right],
        	\end{split}\label{primedJeeprf1}
        \end{equation}
        where we have used the Jacobi identity and $\text{ad}_{e^{(\digamma)}_0}^2=0$. Therefore, we need to show that the expression on the right hand side vanishes. In particular,
        \begin{equation}
        	\left[\left[e^{(\digamma)}_0,J\left(e^{(a)}_0\right)\right],e^{(b)}_0\right]=-(-1)^{|a|}\left[J\left(e^{(a)}_0\right),\left[e^{(\digamma)}_0,e^{(b)}_0\right]\right].
        \end{equation}
        Using
        \begin{equation}
        	J\left(e^{(a)}_0\right)=\frac{1}{\left(\alpha^{(c)},\alpha^{(a)}\right)}\left[J\left(\psi^{(c)}_0\right),e^{(a)}_0\right]-\frac{\sigma_1^{ac}-\sigma_1^{ca}}{2}e^{(a)}_0
        \end{equation}
        with $c$ taken to be $\digamma$, the first term on the right hand side in \eqref{primedJeeprf1} is equal to $(-1)^{1+|a|}\left[e^{(\digamma)}_0,\mathcal{X}_1\right]$ plus a term proportional to $\left[e^{(\digamma)}_0,\left[e^{(a)}_0,\left[e^{(\digamma)}_0,e^{(b)}_0\right]\right]\right]$. In particular, $\left[e^{(\digamma)}_0,\left[e^{(a)}_0,\left[e^{(\digamma)}_0,e^{(b)}_0\right]\right]\right]$ vanishes due to the Serre relation for $\digamma$. Here,
        \begin{equation}
        	\mathcal{X}_1=-(-1)^{|a|}\frac{1}{A_{a\digamma}}\left[\left[J\left(\psi^{(\digamma)}_0\right),e^{(a)}_0\right],\left[e^{(\digamma)}_0,e^{(b)}_0\right]\right].
        \end{equation}
        Likewise, the second line in \eqref{primedJeeprf1} is equal to $-(-1)^{1+|a|}\left[e^{(\digamma)}_0,\mathcal{X}_2\right]$, where
        \begin{equation}
        	\mathcal{X}_2=\frac{1}{A_{b\digamma}}\left[\left[e^{(\digamma)}_0,e^{(a)}_0\right],\left[J\left(\psi^{(\digamma)}_0\right),e^{(b)}_0\right]\right].
        \end{equation}
        Therefore, showing that the right hand side in \eqref{primedJeeprf1} vanishes is equivalent to showing $\left[e^{(\digamma)}_0,\mathcal{X}_1\right]=\left[e^{(\digamma)}_0,\mathcal{X}_2\right]$. Using the Jacobi identity and $\text{ad}_{e^{(\digamma)}_0}^2=0$, we have
        \begin{equation}
        	\mathcal{X}_1=-(-1)^{|a||b|+|b|}\frac{1}{A_{a\digamma}}\left[e^{(b)}_0,\left[e^{(\digamma)}_0,\left[J\left(\psi^{(\digamma)}_0\right),e^{(a)}_0\right]\right]\right].
        \end{equation}
        Keep using the Jacobi identity, and we get
        \begin{equation}
        	\left[e^{(\digamma)}_0,\mathcal{X}_1\right]=\left[e^{(\digamma)}_0,\mathcal{X}_3\right]+\mathcal{X}_4,
        \end{equation}
        where
        \begin{equation}
        	\mathcal{X}_3=-(-1)^{|a||b|+|b|}\frac{1}{A_{a\digamma}}\left[\left[e^{(b)}_0,J\left(\psi^{(\digamma)}_0\right)\right],\left[e^{(\digamma)}_0,,e^{(a)}_0\right]\right]
        \end{equation}
        and $\mathcal{X}_4$ is proportional to
        \begin{equation}
        	\begin{split}
        		&\left[e^{(\digamma)}_0,\left[J\left(\psi^{(\digamma)}_0\right),\left[e^{(b)}_0,\left[e^{(\digamma)}_0,,e^{(a)}_0\right]\right]\right]\right]\\
        		=&\left[\left[e^{(\digamma)}_0,J\left(\psi^{(\digamma)}_0\right)\right],\left[e^{(b)}_0,\left[e^{(\digamma)}_0,,e^{(a)}_0\right]\right]\right]-\left[J\left(\psi^{(\digamma)}_0\right),\left[e^{(\digamma)}_0,\left[e^{(b)}_0,\left[e^{(\digamma)}_0,,e^{(a)}_0\right]\right]\right]\right].
        	\end{split}
        \end{equation}
        Hence, $\mathcal{X}_4$ vanishes due to $\left[e^{(\digamma)}_0,J\left(\psi^{(\digamma)}_0\right)\right]=0$ and the Serre relation for $\digamma$. Therefore, $\left[e^{(\digamma)}_0,\mathcal{X}_1\right]=\left[e^{(\digamma)}_0,\mathcal{X}_3\right]$. It is straightforward to see that $\mathcal{X}_3=\mathcal{X}_2$. This shows that \eqref{primedJeeprf1} vanishes.
        
        \item otherwise: The remaining cases are immediate following the expressions of the primed generators and the similar arguments as above.
	\end{itemize}
\end{proof}

\begin{remark}
	In terms of $\psi^{(a)}_1$, $e^{(a)}_1$ and $f^{(a)}_1$, we have
	\begin{equation}
		\begin{split}
			&\widetilde{\psi}'^{(a)}_1=\begin{cases}
				-\widetilde{\psi}^{(a)}_1,&\quad a=\digamma,\\
				\widetilde{\psi}^{(a)}_1+\widetilde{\psi}^{(\digamma)}_1-\sigma_1^{a\digamma}e^{(\digamma)}_0f^{(\digamma)}_0+\sigma_1^{\digamma a}f^{(\digamma)}_0e^{(\digamma)}_0,&\quad a=\digamma\pm1,\\
				\widetilde{\psi}^{(a)}_1,&\quad\text{otherwise};
			\end{cases}\\
		    &e'^{(a)}_1=\begin{cases}
		    	f^{(\digamma)}_1-\frac{\sigma_1^{b\digamma}}{\left(\alpha^{(\digamma)},\alpha^{(b)}\right)}\psi^{(\digamma)}_0f^{(\digamma)}_0-\frac{\sigma_1^{\digamma b}}{\left(\alpha^{(\digamma)},\alpha^{(b)}\right)}f^{(\digamma)}_0\psi^{(\digamma)}_0,&\quad a=\digamma,\\
		    	\left[e^{(\digamma)}_0,e^{(a)}_1\right],&\quad a=\digamma\pm1,\\
		    	e^{(a)}_1,&\quad\text{otherwise};
		    \end{cases}\\
	        &f'^{(a)}_1=\begin{cases}
	        	-e^{(\digamma)}_1+\frac{\sigma_1^{b\digamma}}{\left(\alpha^{(\digamma)},\alpha^{(b)}\right)}e^{(\digamma)}_0\psi^{(\digamma)}_0+\frac{\sigma_1^{\digamma b}}{\left(\alpha^{(\digamma)},\alpha^{(b)}\right)}\psi^{(\digamma)}_0e^{(\digamma)}_0,&\quad a=\digamma,\\
	        	-\frac{1}{\left(\alpha^{(\digamma)},\alpha^{(a)}\right)}\left[f^{(a)}_1,f^{(\digamma)}_0\right],&\quad a=\digamma\pm1,\\
	        	f^{(a)}_1,&\quad\text{otherwise}.
	        \end{cases}
		\end{split}
	\end{equation}
    Here, $b$ can be taken as either $\digamma+1$ or $\digamma-1$ which would give the same result. One may check that this satisfies the relations for quiver Yangians in Theorem \ref{minthm}. Notice that the coefficients $\sigma_1^{b\digamma}/A_{b\digamma}$ and $\sigma_1^{\digamma b}/A_{b\digamma}$ are equal to $\epsilon_{1,2}$. It is also straightforward to write this isomorphism using $\uppsi'^{(a)}_n=(\epsilon_1+\epsilon_2)\psi'^{(a)}_n$, $\mathtt{e}'^{(a)}_n=(\epsilon_1+\epsilon_2)^{1/2}e'^{(a)}_n$ and $\mathtt{f}'^{(a)}_n=(\epsilon_1+\epsilon_2)^{1/2}f'^{(a)}_n$.
\end{remark}

\section{Connections to $\mathcal{W}$-Algebras}\label{YandW}
As mentioned before, the quiver Yangians and certain $\mathcal{W}$-algebras are expected to have intimate relations that implement the BPS/CFT correspondence. Indeed, as we are now going to see, the rectangular $\mathcal{W}$-algebras for the associated generalized conifolds can be viewed as truncations of the quiver Yangians.

\subsection{From $\mathtt{Y}$ to $\mathcal{W}$}\label{truncations}
Here, we shall directly start with the commutation relations for the generators of rectangular $\mathcal{W}$-algebras for the generalized conifold $xy=z^Mw^N$. A mathematical definition of rectangular $W$-algebras $\mathcal{W}^k\left(\mathfrak{gl}(Ml|Nl),\left(l^{(M|N)}\right)\right)$ is given in Appendix \ref{recW} with the notations and conventions set up therein. For brevity, we shall abbreviate it as $\mathcal{W}_{M|N\times l}$.

The $\mathcal{W}$-algebras of interest in this paper can be generated by $U^{(s)}_{ij}$ with spin $s=1,2$ and $i,j\in\mathbb{Z}/(M+N)\mathbb{Z}$. Given a parity sequence $\varsigma=\{\varsigma_i\}$ as introduced in \S\ref{QY}, the generator $U^{(s)}_{ij}$ has the $\mathbb{Z}_2$-grading given by $(-1)^{p(i)+p(j)}$, where $(-1)^{p(i)}=\varsigma_i$ (see also \eqref{pdef})\footnote{As we will see shortly, $|a|$ and $p(i)$ are indeed consistent in the sense of $\varsigma$ when relating $\mathtt{Y}$ and $\mathcal{W}$. In other words, $|a|$ is bosonic when $p(a)=p(a+1)$ and fermionic otherwise.}. The OPEs of the currents $U^{(s)}_{ij}(z)$ were obtained in \cite{Rapcak:2019wzw,Eberhardt:2019xmf}. The following commutation relations for their modes $U^{(s)}_{ij}[m]$ can then be computed directly using \eqref{OPE2commreln}.
\begin{lemma}\label{UUlemma}
	We have
	\begin{align}
			&\left[U^{(1)}_{i_1j_1}[m],U^{(1)}_{i_2j_2}[n]\right]\nonumber\\
			=&\delta_{m,-n}ml\left(\delta_{j_1i_2}\delta_{i_1j_2}(-1)^{p(j_1)}\varkappa+\delta_{i_1j_1}\delta_{i_2j_2}\right)\nonumber\\
			&+(-1)^{p(i_1)p(j_1)+p(i_2)p(j_2)+p(j_1)p(i_2)}\delta_{i_1j_2}U^{(1)}_{i_2j_1}[m+n]-(-1)^{p(j_1)}\delta_{i_2j_1}U^{(1)}_{i_1j_2}[m+n],
	\end{align}
    \begin{align}
    		&\left[U^{(1)}_{i_1j_1}[m],U^{(2)}_{i_2j_2}[n]\right]\nonumber\\
    		=&\frac{1}{2}l(l-1)m(m-1)\varkappa\delta_{m,-n}\left((-1)^{p(j_1)}\varkappa\delta_{i_1j_2}\delta_{i_2j_1}+\delta_{i_1j_1}\delta_{i_2j_2}\right)\nonumber\\
    		&+m(l-1)\left((-1)^{p(i_1)p(j_1)+p(i_2)p(j_2)+p(j_1)p(i_2)}\varkappa\delta_{i_1j_2}U^{(1)}_{i_2j_1}[m+n]+\delta_{i_1j_1}U^{(1)}_{i_2j_2}[m+n]\right)\nonumber\\
    		&+(-1)^{p(i_1)p(j_1)+p(i_2)p(j_2)+p(j_1)p(i_2)}\delta_{i_1j_2}U^{(2)}_{i_2j_1}[m+n]-(-1)^{p(j_1)}\delta_{i_2j_1}U^{(2)}_{i_1j_2}[m+n],
    \end{align}
    \begin{align}
    		&\left[U^{(2)}_{ii}[m],U^{(2)}_{jj}[n]\right]\nonumber\\
    		=&\frac{1}{12}l(l-1)m(m-1)(m+1)\delta_{m,-n}\left(2\varkappa((1-2l)\alpha^2+1)(-1)^{p(i)}\delta_{ij}-(4l-3)\varkappa^2+1\right)\nonumber\\
    		&+\frac{1}{2}m(m+1)\left((l-1)^2\varkappa\left(U^{(1)}_{ii}[m+n]-U^{(1)}_{jj}[m+n]\right)\right)\nonumber\\
    		&-(m+1)\left(U^{(2)}_{jj}[m+n]+U^{(2)}_{ii}[m+n]+2\varkappa(-1)^{p(i)}\delta_{ij}U^{(2)}_{ii}[m+n]\right)\nonumber\\
    		&+(m+1)(l-1)\left(\sum_{k<0}U^{(1)}_{ii}[k]U^{(1)}_{jj}[m+n-k]+\sum_{k\geq0}U^{(1)}_{jj}[m+n-k]U^{(1)}_{ii}[k]\right)\nonumber\\
    		&+(m+1)(l-1)(-1)^{p(j)}\varkappa\left(\sum_{k<0}U^{(1)}_{ij}[k]U^{(1)}_{ji}[m+n-k]+(-1)^{p(i)+p(j)}\sum_{k\geq0}U^{(1)}_{ji}[m+n-k]U^{(1)}_{ij}[k]\right)\nonumber\\
    		&-(m+1)l(l-1)\varkappa(m+n+1)\left(1+(-1)^{p(i)}\varkappa\delta_{ij}\right)U^{(1)}_{ii}[m+n]\nonumber\\
    		&-(m+n+2)\left(U^{(2)}_{jj}[m+n]-(-1)^{p(i)}\varkappa\delta_{ij}U^{(2)}_{ii}[m+n]-2U^{(2)}_{ii}[m+n]\right)\nonumber\\
    		&+(-1)^{p(i)}\left(\sum_{k<-1}U^{(2)}_{ji}[k]U^{(1)}_{ij}[m+n-k]+\sum_{k\geq-1}(-1)^{p(i)+p(j)}U^{(1)}_{ij}[m+n-k]U^{(2)}_{ji}[k]\right)\nonumber\\
    		&-(-1)^{p(j)}\left(\sum_{k<-1}U^{(2)}_{ij}[k]U^{(1)}_{ji}[m+n-k]+\sum_{k\geq-1}(-1)^{p(i)+p(j)}U^{(1)}_{ji}[m+n-k]U^{(2)}_{ij}[k]\right)\nonumber\\
    		&+(l-1)\left(\sum_{k<0}(-k-1)U^{(1)}_{ii}[k]U^{(1)}_{jj}[m+n-k]+\sum_{k\geq0}(-k-1)U^{(1)}_{jj}[m+n-k]U^{(1)}_{ii}[k]\right)\nonumber\\
    		&+(l-1)\varkappa(-1)^{p(j)}\left(\sum_{k<0}(-k-1)U^{(1)}_{ij}[k]U^{(1)}_{ji}[m+n-k]\right.\nonumber\\
    		&\left.\qquad\qquad\qquad\qquad\quad+(-1)^{p(i)+p(j)}\sum_{k\geq0}(-k-1)U^{(1)}_{ji}[m+n-k]U^{(1)}_{ij}[k]\right)\nonumber\\
    		&+\frac{1}{2}(l-1)(m+n+1)(m+n+2)\left((l+1)\varkappa U^{(1)}_{ii}[m+n]-\varkappa U^{(1)}_{jj}[m+n]\right)\nonumber\\
    		&+\frac{1}{2}(l-1)(m+n+1)(m+n+2)(-1)^{p(j)}l\varkappa^2\delta_{ij}U^{(1)}_{ii}[m+n].
    \end{align}
\end{lemma}
Notice that we only give the $U^{(2)}_{i_1i_2}U^{(2)}_{j_1j_2}$ relation when $i_1=j_1$ and $i_2=j_2$ as this is sufficient for the use here. It is also straightforward to get the more general case from the OPE. In this paper, we shall always assume $\varkappa\neq0$.

The $\mathcal{W}$-algebra is often defined via the distinguished parity sequence, that is, only two fermionic $p(i)$ (the non-super case $M|0$ always has bosonic ones only). Here, we allow it to have different $\varsigma$. Analogous to the quiver Yangians related by Seiberg dualities, we would expect the $\mathcal{W}$-algebras with different $\varsigma$ are essentially the same. In fact, the proof of this is much simpler than the quiver Yangian case in \S\ref{isoQY} by virtue of the matrix presentation here.
\begin{proposition}
	Given $M$, $N$ and $l$, the rectangular $\mathcal{W}$-algebras $\mathcal{W}_{M|N\times l}$ are isomorphic for different $\varsigma$.\label{Wiso}
\end{proposition}
\begin{proof}
	The isomorphism can be constructed from a sequence of the following isomorphic maps. Suppose $\varsigma$ and $\varsigma'$ are related by $\sigma\in\mathfrak{S}_{M+N}$ that permutes the $i^\text{th}$ and $(i+1)^\text{th}$ elements. Then the transformation is given by $U^{(r)}_{ij}\mapsto U^{(r)}_{\sigma(i)\sigma(j)}$. It is straightforward to see that this preserves the relations for the generators.
\end{proof}
Therefore, when considering the map from the quiver Yangians to the (universal enveloping algebra of) $\mathcal{W}$-algebra below, we can simply take them to have the same $\varsigma$. The isomorphic ones are related by the transformations in Theorem \ref{QYisothm} and Proposition \ref{Wiso} respectively.

Let $\epsilon_{\pm}=\epsilon_1\pm\epsilon_2$. Since
\begin{equation}
	\varsigma_a(\epsilon_1-\epsilon_2)=\begin{cases}
		\varsigma_{a+1}(\epsilon_1-\epsilon_2),&\varsigma_a=\varsigma_{a+1}\\
		\varsigma_{a+1}(\epsilon_2-\epsilon_1),&\varsigma_a=-\varsigma_{a+1},
	\end{cases}
\end{equation}
we can take $\epsilon_-=\sigma_1^{a+1,a}-\sigma_1^{a,a+1}$ for any $a$ based on Figure \ref{epsilongencon} without loss of generality. This allows us to consider another presentation of the quiver Yangian which would be convenient for our discussions. Let us prepare the generators defined as
\begin{equation}
	\begin{cases}
		\mathcal{H}^{(a)}_0=\psi^{(a)}_0&\\
		\mathcal{E}^{(a)}_0=e^{(a)}_0\\
		\mathcal{F}^{(a)}_0=f^{(a)}_0,
	\end{cases}
    \qquad
    \begin{cases}
    	\mathcal{H}^{(a)}_1=\psi^{(a)}_1+\frac{1}{2}\nu(a)\epsilon_-\psi^{(a)}_0&\\
    	\mathcal{E}^{(a)}_1=e^{(a)}_1+\frac{1}{2}\nu(a)\epsilon_-e^{(a)}_0&\\
    	\mathcal{F}^{(a)}_1=f^{(a)}_1+\frac{1}{2}\nu(a)\epsilon_-f^{(a)}_0,
    \end{cases}
\end{equation}
where $\nu(a)$ can be any function satisfying $\nu(a\pm1)=\nu(a)\pm1$ for $0\leq a\leq M+N-1$. In particular, this means that $\nu(-1)\neq\nu(M+N-1)$ and $\nu(0)\neq\nu(M+N)$. For instance, the simplest example would be $\nu(a)=a$ ($-1\leq a\leq M+N$). We shall also pick a reference node labelled by $a=0$.
\begin{proposition}\label{HEFprop}
	The quiver Yangian is generated by $\mathcal{H}^{(a)}_r$, $\mathcal{E}^{(a)}_r$, $\mathcal{F}^{(a)}_r$ ($a\in Q_0$, $r=0,1$) subject to the relations
	\begin{align}
		&\left[\mathcal{H}^{(a)}_r,\mathcal{H}^{(b)}_s\right]=0,\\
		&\left[\mathcal{E}^{(a)}_r,\mathcal{F}^{(b)}_s\right]=\delta_{ab}\mathcal{H}^{(a)}_{r+s},\\
		&\left[\mathcal{H}^{(a)}_0,\mathcal{E}^{(a)}_r\right]=A_{ab}\mathcal{E}^{(a)}_r,\\
		&\left[\widetilde{\mathcal{H}}^{(a)}_1,\mathcal{E}^{(b)}_0\right]=\begin{cases}
			A_{ab}\mathcal{E}^{(b)}_1+\frac{1}{2}A_{ab}\nu(M+N)\epsilon_-\mathcal{E}^{(b)}_0,&(a,b)=(M+N-1,0)\\
			A_{ab}\mathcal{E}^{(b)}_1-\frac{1}{2}A_{ab}\nu(M+N)\epsilon_-\mathcal{E}^{(b)}_0,&(a,b)=(0,M+N-1)\\
			A_{ab}\mathcal{E}^{(b)}_1,&\text{otherwise},
		\end{cases}\\
	    &\left[\mathcal{H}^{(a)}_0,\mathcal{F}^{(a)}_r\right]=-A_{ab}\mathcal{F}^{(a)}_r,\\
	    &\left[\widetilde{\mathcal{H}}^{(a)}_1,\mathcal{F}^{(b)}_0\right]=\begin{cases}
	    	-A_{ab}\mathcal{E}^{(b)}_1-\frac{1}{2}A_{ab}\nu(M+N)\epsilon_-\mathcal{E}^{(b)}_0,&(a,b)=(M+N-1,0)\\
	    	-A_{ab}\mathcal{E}^{(b)}_1+\frac{1}{2}A_{ab}\nu(M+N)\epsilon_-\mathcal{E}^{(b)}_0,&(a,b)=(0,M+N-1)\\
	    	-A_{ab}\mathcal{E}^{(b)}_1,&\text{otherwise},
	    \end{cases}\\
        &\left[\mathcal{E}^{(a)}_0,\mathcal{E}^{(b)}_0\right]=\left[\mathcal{F}^{(a)}_0,\mathcal{F}^{(b)}_0\right]=0\qquad(\sigma_1^{ab}=0),\\
        &\nonumber\\
        &\left[\mathcal{E}^{(a)}_1,\mathcal{E}^{(b)}_0\right]-\left[\mathcal{E}^{(a)}_0,\mathcal{E}^{(b)}_1\right]\nonumber\\
        =&\begin{cases}
        	\frac{1}{2}A_{ab}\epsilon_+\left\{\mathcal{E}^{(a)}_0,\mathcal{E}^{(b)}_0\right\}-\frac{1}{2}\nu(M+N)\epsilon_-\left[\mathcal{E}^{(a)}_0,\mathcal{E}^{(b)}_0\right],&(a,b)=(0,M+N-1)\\
        	\frac{1}{2}A_{ab}\epsilon_+\left\{\mathcal{E}^{(a)}_0,\mathcal{E}^{(b)}_0\right\},&\text{otherwise},
        \end{cases}\\
        &\nonumber\\
        &\left[\mathcal{F}^{(a)}_1,\mathcal{F}^{(b)}_0\right]-\left[\mathcal{F}^{(a)}_0,\mathcal{F}^{(b)}_1\right]\nonumber\\
        =&\begin{cases}
        	-\frac{1}{2}A_{ab}\epsilon_+\left\{\mathcal{F}^{(a)}_0,\mathcal{F}^{(b)}_0\right\}-\frac{1}{2}\nu(M+N)\epsilon_-\left[\mathcal{F}^{(a)}_0,\mathcal{F}^{(b)}_0\right],&(a,b)=(0,M+N-1)\\
        	-\frac{1}{2}A_{ab}\epsilon_+\left\{\mathcal{F}^{(a)}_0,\mathcal{F}^{(b)}_0\right\},&\text{otherwise},
        \end{cases}\\
        &\nonumber\\
        &\text{Serre relations }\textup{(S)},
	\end{align}
    where $\widetilde{\mathcal{H}}^{(a)}_1:=\mathcal{H}^{(a)}_1-\frac{1}{2}\epsilon_+\left(\mathcal{H}^{(a)}_0\right)^2$. Recall that $\{x,y\}$ is used to denote $xy+(-1)^{|x||y|}yx$ here. The higher modes with $r\geq2$ can be obtained in a way similar to the presentation using $\psi,e,f$.
\end{proposition}
This can be verified by straightforward calculations. Hence, we omit the explicit proof here. When checking these relations, it is also worth noting that
\begin{equation}
	A_{ab}=\left(\alpha^{(a)},\alpha^{(b)}\right)=\begin{cases}
		-(\varsigma_b+\varsigma_{b+1}),&a=b\\
		\varsigma_{b+1},&a=b+1\\
		\varsigma_b,&b=a+1\\
		0,&\text{otherwise}.
	\end{cases}\label{Aabvarsigma}
\end{equation}

\begin{remark}
	As pointed out in \cite{Bao:2022fpk}, the quiver Yangian is only isomorphic to Ueda's affine super Yangian introduced in \cite{ueda2019affine} when $\epsilon_-=0$. By comparing the presentation above with the similar presentation for Ueda's affine super Yangian in \cite{ueda2022affine}, it is straightforward to see that this difference is encoded by $\nu(a)$ here and the coefficients in the presentation in \cite{ueda2022affine}.
\end{remark}

Now, we are ready to bridge the quiver Yangians and $\mathcal{W}$-algebras.
\begin{theorem}\label{YWthm}
	Given a generalized conifold with $M+N>2$, $MN\neq2$ and $M\neq N$, when $\nu(M+N)\epsilon_-=(2\varkappa-M-N)\epsilon_+$, there is a surjective algebra homomorphism from the quiver Yangian to the universal enveloping algebra of $\mathcal{W}_{M|N\times l}$. Fixing a parity sequence $\varsigma$, such map $\Phi:\mathtt{Y}\rightarrow U\left(\mathcal{W}_{M|N\times l}\right)$ can be uniquely determined by
	\begin{equation}
		\Phi\left(X^{(a)}_0\right)=\Phi\left(Y^{(a)}_0\right),\qquad\Phi\left(X^{(a)}_1\right)=\Phi\left(Y^{(a)}_1\right)-\frac{1}{2}\nu(a)\epsilon_-\Phi\left(Y^{(a)}_0\right)
	\end{equation}
    for $(X,Y)=(\psi,\mathcal{H}),(e,\mathcal{E}),(f,\mathcal{F})$, where
    \begin{align}
    	\Phi\left(\mathcal{H}^{(a)}_0\right)=&\begin{cases}
    		U^{(1)}_{M+N,M+N}[0]-U^{(1)}_{11}[0]+l\varkappa,&a=0\\
    		U^{(1)}_{aa}[0]-U^{(1)}_{a+1,a+1}[0],&a\neq0,
    	\end{cases}\\
        \Phi\left(\mathcal{E}^{(a)}_0\right)=&\begin{cases}
        	-(-1)^{p(1)}U^{(1)}_{M+N,1}[-1],&a=0\\
        	-(-1)^{p(a+1)}U^{(1)}_{a,a+1}[0],&a\neq0,
        \end{cases}\\
        \Phi\left(\mathcal{F}^{(a)}_0\right)=&\begin{cases}
        	U^{(1)}_{1,M+N}[1],&a=0\\
        	U^{(1)}_{a+1,a}[0],&a\neq0,
        \end{cases}\\
        \Phi\left(\mathcal{H}^{(0)}_1\right)=&\epsilon_+\bigg(U^{(2)}_{M+N,M+N}[0]-U^{(2)}_{11}[0]-U^{(1)}_{M+N,M+N}[0]U^{(1)}_{11}[0]\nonumber\\
        	&-\sum\limits_{c=1}^{M+N}\sum_{k\geq0}(-1)^{p(c)+p(M+N)}U^{(1)}_{c,M+N}[-k]U^{(1)}_{M+N,c}[k]\nonumber\\
        	&+\sum\limits_{c=1}^{M+N}\sum_{k\geq0}(-1)^{p(c)+p(1)}U^{(1)}_{c,1}[-k-1]U^{(1)}_{1,c}[k+1]\nonumber\\
        	&+\left(\frac{1}{2}\nu(1)-\frac{1}{2}-l\varkappa\right)\left(U^{(1)}_{M+N,M+N}[0]-U^{(1)}_{11}[0]+l\varkappa\right)+\varkappa U^{(1)}_{M+N,M+N}[0]\bigg),\\
        \Phi\left(\mathcal{H}^{(a\neq0)}_1\right)=&\epsilon_+\bigg(U^{(2)}_{aa}[0]-U^{(2)}_{a+1,a+1}[0]-U^{(1)}_{aa}[0]U^{(1)}_{a+1,a+1}[0]+\frac{1}{2}\nu(a)\left(U^{(1)}_{aa}[0]-U^{(1)}_{a+1,a+1}[0]\right)\nonumber\\
        &-\sum\limits_{c=1}^{a}\sum_{k\geq0}(-1)^{p(c)+p(a)}U^{(1)}_{ca}[-k]U^{(1)}_{ac}[k]\nonumber\\
        &+\sum\limits_{c=a+1}^{M+N}\sum_{k\geq0}(-1)^{p(c)+p(a)}U^{(1)}_{ca}[-k-1]U^{(1)}_{ac}[k+1]\nonumber\\
        &-\sum\limits_{c=1}^{a}\sum_{k\geq0}(-1)^{p(c)+p(a+1)}U^{(1)}_{c,a+1}[-k]U^{(1)}_{a+1,c}[k]\nonumber\\
        &+\sum\limits_{c=a+1}^{M+N}\sum_{k\geq0}(-1)^{p(c)+p(a+1)}U^{(1)}_{c,a+1}[-k-1]U^{(1)}_{a+1,c}[k+1]\bigg),\\
        \Phi\left(\mathcal{E}^{(0)}_1\right)=&\epsilon_+\bigg(-(-1)^{p(1)}U^{(2)}_{M+N,1}[-1]-(-1)^{p(1)}\left(\frac{1}{2}\nu(1)-\frac{1}{2}-l\varkappa\right)U^{(1)}_{M+N,1}[-1]\nonumber\\
        &+\sum\limits_{c=1}^{M+N}\sum_{k\geq0}(-1)^{p(c)+p(M+N)p(c)+p(1)p(c)+p(M+N)p(1)+p(M+N)}U^{(1)}_{c,1}[-k-1]U^{(1)}_{M+N,c}[k]\bigg),\\
        \Phi\left(\mathcal{E}^{(a\neq0)}_1\right)=&\epsilon_+\bigg(-(-1)^{p(a+1)}U^{(2)}_{a,a+1}[0]-\frac{1}{2}\nu(a)(-1)^{p(a+1)}U^{(1)}_{a,a+1}[0]\nonumber\\
        &-\sum\limits_{c=1}^{a}\sum_{k\geq0}(-1)^{p(c)+p(a)p(c)+p(a+1)p(c)+p(a)p(a+1)+p(a)}U^{(1)}_{c,a+1}[-k]U^{(1)}_{ac}[k]\nonumber\\
        &+\sum\limits_{c=a+1}^{M+N}\sum_{k\geq0}(-1)^{p(c)+p(a)p(c)+p(a+1)p(c)+p(a)p(a+1)+p(a)}U^{(1)}_{c,a+1}[-k-1]U^{(1)}_{ac}[k+1]\bigg)\\
        \Phi\left(\mathcal{F}^{(0)}_1\right)=&\epsilon_+\bigg(U^{(2)}_{1,M+N}[1]+\left(\frac{1}{2}\nu(1)-\frac{1}{2}+\varkappa\right)U^{(1)}_{1,M+N}[1]\nonumber\\
        &-\sum\limits_{c=1}^{M+N}\sum_{k\geq0}(-1)^{p(c)+p(M+N)p(c)+p(1)p(c)+p(M+N)p(1)}U^{(1)}_{c,M+N}[-k]U^{(1)}_{1,c}[k+1]\bigg),\\
        \Phi\left(\mathcal{F}^{(a\neq0)}_1\right)=&\epsilon_+\bigg(U^{(2)}_{a+1,a}[0]+\frac{1}{2}\nu(a)U^{(1)}_{a+1,a}[0]\nonumber\\
        &-\sum\limits_{c=1}^{a}\sum_{k\geq0}(-1)^{p(c)+p(a)p(c)+p(a+1)p(c)+p(a)p(a+1)}U^{(1)}_{ca}[-k]U^{(1)}_{a+1,c}[k]\nonumber\\
        &+\sum\limits_{c=a+1}^{M+N}\sum_{k\geq0}(-1)^{p(c)+p(a)p(c)+p(a+1)p(c)+p(a)p(a+1)}U^{(1)}_{ca}[-k-1]U^{(1)}_{a+1,c}[k+1]\bigg).
    \end{align}
\end{theorem}

\begin{proof}
	This is proven by directly applying Lemma \ref{UUlemma}. The construction of this map is analogous to the one for Ueda's affine Yangian in \cite{ueda2022affine} as well as the evaluation maps for certain Yangians in \cite{kodera2021guay,ueda2019affine}. Nevertheless, we should still be careful about the slight differences, and let us give a quick proof here. First, we shall check the relations involving only zero modes. This is essentially building the Chevalley generators from the matrices. Here, we would only explicitly write the proof for the $\mathcal{H}_0\mathcal{H}_0$, $\mathcal{H}_0\mathcal{E}_0$ and Serre relations. The other relations can be verified also following straightforward computations.
	
	The $\mathcal{H}_0\mathcal{H}_0$ relation is immediate since $U^{(1)}_{aa}[0]$ and $U^{(1)}_{bb}[0]$ commute with each other for any $a,b$. Now, suppose $a,b\neq0$, and the $\mathcal{H}_0\mathcal{E}_0$ relation is recovered by
	\begin{equation}
		\begin{split}
			&\left[U^{(1)}_{aa}[0]-U^{(1)}_{a+1,a+1}[0],-(-1)^{p(b+1)}U^{(1)}_{b,b+1}[0]\right]\\
			=&-(-1)^{p(b+1)+p(a)+p(b)p(b+1)+p(a)p(b)}\delta_{a,b+1}U^{(1)}_{ba}[0]+(-1)^{p(b+1)+p(a)}\delta_{ba}U^{(1)}_{a,b+1}[0]\\
			&+(-1)^{p(b+1)+p(a+1)+p(b)p(b+1)+p(a+1)p(b)}\delta_{a+1,b+1}U^{(1)}_{ba}[0]-(-1)^{p(b+1)+p(a+1)}\delta_{b,a+1}U^{(1)}_{a+1,b+1}[0]\\
			=&\begin{cases}
				-(-1)^{p(b+1)}\left(-(-1)^{p(b+1)}U^{(1)}_{b,b+1}[0]\right)-(-1)^{p(b)}\left(-(-1)^{p(b+1)}U^{(1)}_{b,b+1}[0]\right),&a=b\\
				(-1)^{p(b+1)}\left(-(-1)^{p(b+1)}U^{(1)}_{b,b+1}[0]\right),&a=b+1\\
				(-1)^{p(b)}\left(-(-1)^{p(b+1)}U^{(1)}_{b,b+1}[0]\right),&a+1=b\\
				0,&\text{otherwise}
			\end{cases}
		\end{split}
	\end{equation}
    as required following \eqref{Aabvarsigma}. When $a$ and/or $b$ equal(s) zero, the relation can be checked similarly. The Serre relations also hold immediately as each term we get from the commutation relations has factor $\delta_{a,a+1}$, $\delta_{a,a+2}$ or $\delta_{a+1,a-1}$ that vanishes.
    
    Now, let us check the relations involving higher modes. Here, we shall only verify the $\mathcal{H}_1\mathcal{H}_1$ and $\mathcal{H}_1\mathcal{F}_0$ relations explicitly, and the other relations can also be obtained following direct computations. For convenience, let us use the notation $\delta_{\texttt{cond}}$ which is equal to 1 when the condition $\texttt{cond}$ is true and 0 otherwise.
    
    Let us first consider the $\mathcal{H}_1\mathcal{H}_1$ relation for $a,b\neq0$. In the total commutation relation, we have the piece
    \begin{align}
    	&\left[U^{(2)}_{aa}[0]-U^{(2)}_{a+1,a+1}[0],U^{(2)}_{bb}[0]-U^{(2)}_{b+1,b+1}[0]\right]\nonumber\\
    	=&{\color{orange} (l-1)\left(\sum_{k>0}kU^{(1)}_{aa}[k]U^{(1)}_{bb}[-k]-\sum_{k>0}kU^{(1)}_{bb}[-k]U^{(1)}_{aa}[k]\right)}\nonumber\\
    	&{\color{Green} +(l-1)\left((-1)^{p(b)}\varkappa\sum_{k>0}kU^{(1)}_{ab}[k]U^{(1)}_{ba}[-k]-(-1)^{p(a)}\varkappa\sum_{k>0}kU^{(1)}_{ba}[-k]U^{(1)}_{ab}[k]\right)}\nonumber\\
    	&{\color{Magenta} +(-1)^{p(a)}\sum_{k>1}U^{(2)}_{ba}[-k]U^{(1)}_{ab}[k]+(-1)^{p(b)}\sum_{k\geq-1}U^{(1)}_{ab}[-k]U^{(2)}_{ba}[k]}\nonumber\\
    	&{\color{red} -(-1)^{p(b)}\sum_{k>1}U^{(2)}_{ab}[-k]U^{(1)}_{ba}[k]-(-1)^{p(a)}\sum_{k\geq-1}U^{(1)}_{ba}[-k]U^{(2)}_{ab}[k]}\nonumber\\
    	&-\begin{bmatrix}a\rightarrow a\\b\rightarrow b+1\end{bmatrix}-\begin{bmatrix}a\rightarrow a+1\\b\rightarrow b\end{bmatrix}-\begin{bmatrix}a\rightarrow a+1\\b\rightarrow b+1\end{bmatrix},
    \end{align}
    where all the single $U^{(s)}_{aa}[0]$ and $U^{(s)}_{bb}[0]$ get cancelled. We also have
    \begin{align}
    	&-\left[\sum_{c=1}^a\sum_{k\geq0}(-1)^{p(c)+p(a)}U^{(1)}_{ca}[-k]U^{(1)}_{ac}[k]+\sum_{c=a+1}^{M+N}\big(\dots\big),U^{(2)}_{bb}[0]-U^{(2)}_{b+1,b+1}[0]\right]\nonumber\\
    	=&-\sum_{k\geq0}\bigg({\color{orange} k(l-1)U^{(1)}_{aa}[-k]U^{(1)}_{bb}[k]}{\color{Green} -k(l-1)\varkappa(-1)^{p(a)}U^{(1)}_{ba}[-k]U^{(1)}_{ab}[k]}\nonumber\\
    	&{\color{Magenta} +\delta_{a\geq b}(-1)^{p(a)}U^{(2)}_{ba}[-k]U^{(1)}_{ab}[k]+\delta_{a<b}(-1)^{p(a)}U^{(2)}_{ba}[-k-1]U^{(1)}_{ab}[k+1]}\nonumber\\
    	&{\color{red} -\delta_{a\geq b}(-1)^{p(a)}U^{(1)}_{ba}[-k]U^{(2)}_{ab}[k]-\delta_{a<b}(-1)^{p(a)}U^{(1)}_{ba}[-k-1]U^{(2)}_{ab}[k+1]}\nonumber\\
    	&{\color{blue} +k(l-1)\varkappa\delta_{ab}\sum_{c=1}^a(-1)^{p(c)}U^{(1)}_{ca}[-k]U^{(1)}_{ac}[k]+\sum_{c=a+1}^{M+N}\big(\dots\big)}\nonumber\\
    	&{\color{blue} +\delta_{ab}\sum_{c=1}^a(-1)^{p(c)}U^{(1)}_{ca}[-k]U^{(2)}_{ac}[k]+\sum_{c=a+1}^{M+N}\big(\dots\big)}\nonumber\\
    	&{\color{blue} -\delta_{ab}\sum_{c=1}^a(-1)^{p(c)}U^{(2)}_{ca}[-k]U^{(1)}_{ac}[k]+\sum_{c=a+1}^{M+N}\big(\dots\big)}\nonumber\\
    	&{\color{brown} -k(l-1)U^{(1)}_{bb}[-k]U^{(1)}_{aa}[k]}-\begin{bmatrix}b\rightarrow b+1\end{bmatrix}\bigg),
    \end{align}
    where each $\sum\limits_{c=a+1}^{M+N}(\dots)$ indicates a term similar to its previous one (but with $\pm k$ changed to $\pm(k+1)$ in the \emph{modes only}). Likewise,
    \begin{align}
    	&-\left[U^{(2)}_{aa}[0]-U^{(2)}_{a+1,a+1}[0],\sum_{c=1}^b\sum_{k\geq0}(-1)^{p(c)+p(b)}U^{(1)}_{cb}[-k]U^{(1)}_{bc}[k]+\sum_{c=b+1}^{M+N}\big(\dots\big)\right]\nonumber\\
    	=&-\sum_{k\geq0}\bigg({\color{orange} k(l-1)U^{(1)}_{bb}[-k]U^{(1)}_{aa}[k]}{\color{Green} -k(l-1)\varkappa(-1)^{p(b)}U^{(1)}_{ab}[-k]U^{(1)}_{ba}[k]}\nonumber\\
    	&{\color{Magenta} +\delta_{b\geq a}(-1)^{p(b)}U^{(1)}_{ab}[-k]U^{(2)}_{ba}[k]+\delta_{b<a}(-1)^{p(b)}U^{(1)}_{ab}[-k-1]U^{(2)}_{ba}[k+1]}\nonumber\\
    	&{\color{red} -\delta_{b\geq a}(-1)^{p(b)}U^{(2)}_{ab}[-k]U^{(1)}_{ba}[k]-\delta_{b<a}(-1)^{p(b)}U^{(2)}_{ab}[-k-1]U^{(1)}_{ba}[k+1]}\nonumber\\
    	&{\color{blue} -\begin{bmatrix}\text{blue terms}\end{bmatrix}}{\color{brown} +k(l-1)U^{(1)}_{aa}[-k]U^{(1)}_{bb}[k]}-\begin{bmatrix}a\rightarrow a+1\end{bmatrix}\bigg).
    \end{align}
    The terms in the same colours cancel each other (except the two brown ones here) and similarly for the terms as well as the commutation relations with $a\rightarrow a+1$ and/or $b\rightarrow b+1$. Notice that when cancelling the $U^{(2)}U^{(1)}$ or $U^{(1)}U^{(2)}$ terms (in pink and red), we need to apply the commutation relation for $U^{(1)}$ and $U^{(2)}$ for the modes 0 and $\pm1$. The resulted extra terms with single $U^{(s)}$ are cancelled with those from a similar manipulation of the terms in a different colour. The brown terms above are compensated by
    \begin{align}
    	&\left[\sum_{c=1}^a\sum_{k\geq0}(-1)^{p(c)+p(a)}U^{(1)}_{ca}[-k]U^{(1)}_{ac}[k]+\sum_{c=a+1}^{M+N}\big(\dots\big),\begin{bmatrix}\text{terms with}\\a\rightarrow b\end{bmatrix}\right]\nonumber\\
    	=&{\color{brown} \sum_{k\geq0}\left(k(l-1)U^{(1)}_{aa}[-k]U^{(1)}_{bb}[k]-k(l-1)U^{(1)}_{bb}[-k]U^{(1)}_{aa}[k]\right)}.
    \end{align}
    One may check that $U^{(1)}_{aa}[0]$ (and hence $U^{(1)}_{aa}[0]U^{(1)}_{a+1,a+1}[0]$) commutes with other terms. This shows that $\Phi\left(\mathcal{H}^{(a)}_1\right)$ commutes with $\Phi\left(\mathcal{H}^{(b)}_1\right)$ when $a,b\neq0$. Suppose $a=0$ and $b\neq0$, the commutation relations we need to compute are the same as the case for $a,b\neq0$. When $a=b=0$, this relation is trivial.
    
    Let us next check the $\mathcal{H}_1\mathcal{F}_0$ relation when $a,b\neq0$. Similar to the computation for the zero modes, the commutation relation $\epsilon_+\left[U^{(2)}_{aa}[0]-U^{(2)}_{a+1,a+1}[0],U^{(1)}_{b+1,b}[0]\right]$ is responsible for the $-A_{ab}\epsilon_+U^{(2)}_{b+1,b}[0]$ term in $-A_{ab}\Phi\left(\mathcal{F}^{(b)}_1\right)$. Henceforth, we shall omit $-A_{ab}$ and $\epsilon_+$ when making similar statements for brevity (though the calculations would still carry $(-1)^{p(a)}$ factors that lead to the corresponding Cartan matrix element). Now, consider
    \begin{align}
    	&-\left[\sum_{c=1}^a\sum_{k\geq0}(-1)^{p(c)+p(a)}U^{(1)}_{ca}[-k]U^{(1)}_{ac}[k]+\sum_{c=a+1}^{M+N}\big(\dots\big),U^{(1)}_{b+1,b}[0]\right]\nonumber\\
    	=&{\color{orange} (-1)^{p(b)}\left(-\sum_{c=1}^a\sum_{k\geq0}(-1)^{p(c)+p(b)p(c)+p(b+1)p(c)+p(b)p(b+1)}\delta_{ab}U^{(1)}_{cb}[-k]U^{(1)}_{b+1,c}[k]\right)}\nonumber\\
    	&{\color{red} +\sum_{k\geq0}(-1)^{p(a)}\delta_{a\geq b+1}U^{(1)}_{b+1,a}[-k]U^{(1)}_{ab}[k]-\sum_{k\geq0}(-1)^{p(a)}\delta_{a\geq b}U^{(1)}_{b+1,a}[-k]U^{(1)}_{ab}[k]}\nonumber\\
    	&{\color{Green} -(-1)^{p(b+1)}\left(-\sum_{c=1}^a\sum_{k\geq0}(-1)^{p(c)+p(b)p(c)+p(b+1)p(c)+p(b)p(b+1)}\delta_{b+1,a}U^{(1)}_{cb}[-k]U^{(1)}_{b+1,c}[k]\right)}\nonumber\\
    	&{\color{orange} +(-1)^{p(b)}\left(-\sum_{c=a}^{M+N}\sum_{k\geq0}(-1)^{p(c)+p(b)p(c)+p(b+1)p(c)+p(b)p(b+1)}\delta_{ab}U^{(1)}_{cb}[-k-1]U^{(1)}_{b+1,c}[k+1]\right)}\nonumber\\
    	&{\color{red} +\sum_{k\geq0}(-1)^{p(a)}\delta_{a<b+1}U^{(1)}_{b+1,a}[-k-1]U^{(1)}_{ab}[k+1]-\sum_{k\geq0}(-1)^{p(a)}\delta_{a<b}U^{(1)}_{b+1,a}[-k-1]U^{(1)}_{ab}[k+1]}\nonumber\\
    	&{\color{Green} -(-1)^{p(b+1)}\left(-\sum_{c=a+1}^{M+N}\sum_{k\geq0}(-1)^{p(c)+p(b)p(c)+p(b+1)p(c)+p(b)p(b+1)}\delta_{b+1,a}U^{(1)}_{cb}[-k-1]U^{(1)}_{b+1,c}[k+1]\right)},
    \end{align}
    and
    \begin{align}
    	&\left[\sum_{c=1}^a\sum_{k\geq0}(-1)^{p(c)+p(a+1)}U^{(1)}_{c,a+1}[-k]U^{(1)}_{a+1,c}[k]+\sum_{c=a+1}^{M+N}\big(\dots\big),U^{(1)}_{b+1,b}[0]\right]\nonumber\\
    	=&{\color{Green} -(-1)^{p(b)}\left(-\sum_{c=1}^a\sum_{k\geq0}(-1)^{p(c)+p(b)p(c)+p(b+1)p(c)+p(b)p(b+1)}\delta_{a+1,b}U^{(1)}_{cb}[-k]U^{(1)}_{b+1,c}[k]\right)}\nonumber\\
    	&{\color{red} -\sum_{k\geq0}(-1)^{p(a+1)}\delta_{a\geq b+1}U^{(1)}_{b+1,a+1}[-k]U^{(1)}_{a+1,b}[k]+\sum_{k\geq0}(-1)^{p(a+1)}\delta_{a\geq b}U^{(1)}_{b+1,a+1}[-k]U^{(1)}_{a+1,b}[k]}\nonumber\\
    	&{\color{orange} +(-1)^{p(b+1)}\left(-\sum_{c=1}^a\sum_{k\geq0}(-1)^{p(c)+p(b)p(c)+p(b+1)p(c)+p(b)p(b+1)}\delta_{ba}U^{(1)}_{cb}[-k]U^{(1)}_{b+1,c}[k]\right)}\nonumber\\
    	&{\color{Green} -(-1)^{p(b)}\left(-\sum_{c=a}^{M+N}\sum_{k\geq0}(-1)^{p(c)+p(b)p(c)+p(b+1)p(c)+p(b)p(b+1)}\delta_{a+1,b}U^{(1)}_{cb}[-k-1]U^{(1)}_{b+1,c}[k+1]\right)}\nonumber\\
    	&{\color{red} -\sum_{k\geq0}(-1)^{p(a+1)}\delta_{a<b+1}U^{(1)}_{b+1,a+1}[-k-1]U^{(1)}_{a+1,b}[k+1]+\sum_{k\geq0}(-1)^{p(a+1)}\delta_{a<b}U^{(1)}_{b+1,a+1}[-k-1]U^{(1)}_{a+1,b}[k+1]}\nonumber\\
    	&{\color{orange} +(-1)^{p(b+1)}\left(-\sum_{c=a+1}^{M+N}\sum_{k\geq0}(-1)^{p(c)+p(b)p(c)+p(b+1)p(c)+p(b)p(b+1)}\delta_{ba}U^{(1)}_{cb}[-k-1]U^{(1)}_{b+1,c}[k+1]\right)}.
    \end{align}
    The four green lines are equal to
    \begin{equation}
    	\begin{cases}
    		{\color{orange} -\varsigma_{b+1}\begin{bmatrix}\text{orange}\\ \text{sums}\end{bmatrix}}{\color{red} +(-1)^{p(b+1)}\delta_{b+1,a}U^{(1)}_{b+1,b}[0]U^{(1)}_{aa}[0]},&a=b+1\\
    		{\color{orange} -\varsigma_b\begin{bmatrix}\text{orange}\\ \text{sums}\end{bmatrix}}{\color{red} -(-1)^{p(b)}\delta_{a+1,b}U^{(1)}_{a+1,a+1}[0]U^{(1)}_{b+1,b}[0]},&b=a+1.
    	\end{cases}
    \end{equation}
    All the orange terms above then give the corresponding sum in $\Phi\left(\mathcal{F}^{(b)}_1\right)$. For the remaining red terms, altogether they become
    \begin{equation}
    	\begin{cases}
    		{\color{Magenta} (-1)^{p(b+1)}\delta_{b+1,a}U^{(1)}_{b+1,b}[0]U^{(1)}_{aa}[0]},&a\geq b+1\\
    		{\color{Magenta} -(-1)^{p(b)}U^{(1)}_{b+1,b}[0]U^{(1)}_{aa}[0]+(-1)^{p(b+1)}U^{(1)}_{a+1,a+1}[0]U^{(1)}_{b+1,b}[0]},&a=b\\
    		{\color{Magenta} -(-1)^{p(b)}\delta_{a+1,b}U^{(1)}_{a+1,a+1}[0]U^{(1)}_{b+1,b}[0]},&a<b-1.
    	\end{cases}
    \end{equation}
    Since it would be more convenient to consider $\widehat{\mathcal{H}}^{(a)}_1$, the $U^{(1)}_{a+1,a+1}[0]U^{(1)}_aa[0]$ term gets cancelled and we are left with $\frac{1}{2}\left(\left(U^{(1)}_{aa}[0]\right)^2+\left(U^{(1)}_{a+1,a+1}[0]\right)^2\right)$. Then
    \begin{align}
    	&\left[\frac{1}{2}\left(\left(U^{(1)}_{aa}[0]\right)^2+\left(U^{(1)}_{a+1,a+1}[0]\right)^2\right),U^{(1)}_{b+1,1}[0]\right]\nonumber\\
    	=&{\color{Magenta} (-1)^{p(a)}\delta_{ab}U^{(1)}_{aa}[0]U^{(1)}_{b+1,b}[0]-(-1)^{p(a)}\delta_{a,b+1}U^{(1)}_{aa}[0]U^{(1)}_{b+1,b}[0]}\nonumber\\
    	&{\color{Magenta} +(-1)^{p(a)}\delta_{ab}U^{(1)}_{b+1,b}[0]U^{(1)}_{aa}[0]-(-1)^{p(a)}\delta_{a,b+1}U^{(1)}_{b+1,b}[0]U^{(1)}_{aa}[0]}\nonumber\\
    	&{\color{Magenta} +(-1)^{p(a+1)}\delta_{a+1,b}U^{(1)}_{a+1,a+1}[0]U^{(1)}_{b+1,b}[0]-(-1)^{p(a+1)}\delta_{ab}U^{(1)}_{a+1,a+1}[0]U^{(1)}_{b+1,b}[0]}\nonumber\\
    	&{\color{Magenta} +(-1)^{p(a+1)}\delta_{a+1,b}U^{(1)}_{b+1,b}[0]U^{(1)}_{a+1,a+1}[0]-(-1)^{p(a+1)}\delta_{ab}U^{(1)}_{b+1,b}[0]U^{(1)}_{a+1,a+1}[0]}.
    \end{align}
    All the pink terms give
    \begin{equation}
    	\begin{cases}
    		0,&a=b\\
    		-\varsigma_{b+1}\frac{1}{2}U^{(1)}_{b+1,b}[0],&a=b+1\\
    		\varsigma_b\frac{1}{2}U^{(1)}_{b+1,b}[0],&b=a+1.
    	\end{cases}
    \end{equation}
    Together with $\frac{1}{2}\nu(a)\left[U^{(1)}_{aa}[0]-U^{(1)}_{a+1,a+1}[0],U^{(1)}_{b+1,b}[0]\right]$ (recall that $\nu(a\pm1)=\nu(a)\pm1$), they recover the remaining term in $\Phi\left(\mathcal{F}^{(b)}_1\right)$. When $a$ and/or $b$ equal(s) zero, this can be verified in the same manner. Therefore, we would not write it explicitly here. Notice that the condition on $\epsilon_\pm$ would come from the cases when $a$ and/or $b$ equal(s) zero.
    
    We have shown that $\Phi$ is a homomorphism, and we still need to check its surjectivity. This can be done by showing that any element $U^{(s)}_{ab}[m]$ ($s=1,2$, $a,b=1,\dots,M+N$, $m\in\mathbb{Z}$) can be expressed in terms of $\Phi(\mathcal{H})$, $\Phi(\mathcal{E})$ and $\Phi(\mathcal{F})$. We shall omit the proof here as these generators will be constructed from $\Phi$ explicitly in the next subsection.
\end{proof}

\begin{remark}
	From Theorem \ref{YWthm}, we can see that the universal enveloping algebras of $\mathcal{W}_{M|N\times l}$ are essentially truncations of the quiver Yangians, that is,
	\begin{equation}
		\mathtt{Y}/\ker(\Phi)\cong U(\mathcal{W}_{M|N\times l}).
	\end{equation}
	Therefore, we may view the quiver Yangian as some sort of ``$U(\mathcal{W}_{M|N\times\infty})$'' algebra (cf.\cite{Rapcak:2019wzw})\footnote{Notice that this is not claiming that $\Phi$ becomes an isomorphism when taking the limit $l\rightarrow\infty$. It still requires to show the injectivity. However, this map does not seem to be well-defined when $l$ diverges, and the factors $l$ cannot be fully absorbed under redefinitions of the generators in these expressions. One might think of taking $\varkappa\rightarrow0$ as another possible way to bypass this divergence, but some properties of $\Phi$, such as surjectivity, rely on $\varkappa\neq0$.}. This allows us to apply our knowledge in BPS algebras to VOAs and vice versa.
\end{remark}

At first glance, one might wonder whether we could write such surjective homomorphism without choosing a reference $a=0$ so that the map would become more ``uniform''. However, as we will see shortly, this is actually very natural on the quiver Yangian side (not just due to the presentation in Proposition \ref{HEFprop}), especially when discussing the crystal melting models.

\paragraph{Coproduct and parabolic induction} As the coproduct for the quiver Yangians is obtained in \S\ref{coprod}, we can thence consider the parabolic induction for the $\mathcal{W}$-algebras. In other words, given representations $R_{1,2}$ of $U(\mathcal{W}_{M|N\times l_{1,2}})$, we have $R_1\otimes R_2$ as a representation of $U\left(\mathcal{W}_{M|N\times (l_1+l_2)}\right)$. In particular, the study in \cite{kodera2022coproduct} (see also \cite{genra2020screening} for non-super cases including some cases of BCD types) answers Conjecture 2 in \cite{Rapcak:2019wzw}.

Consider $\widehat{\mathfrak{gl}}(M|N)_\varkappa=\widehat{\mathfrak{sl}}(M|N)_\varkappa\oplus\widehat{\mathfrak{z}}(M|N)_{\varkappa+M-N}$, where $\widehat{\mathfrak{z}}$ is the Heisenberg algebra at level $\varkappa+M-N=k+l(M-N)$. We have an algebra automorphism given by \cite{kodera2022coproduct}
\begin{equation}
	\eta_\beta\left(E_{ij}[m]\right)=E_{ij}[m]+\delta_{m,0}\delta_{ij}\beta
\end{equation}
for some complex number $\beta$. This yields an algebra automorphism
\begin{equation}
	\eta_\beta^{\otimes l}=\bigotimes_{l\text{ times}}\eta_\beta\in\text{Aut}\left(U\left(\widehat{\mathfrak{gl}}(M+N)_\varkappa\right)^{\widehat{\otimes}l}\right).
\end{equation}
Using \eqref{U1U2}, we have
\begin{equation}
	\begin{split}
		&\eta^{\otimes l}_\beta\left(U^{(1)}_{ij}[m]\right)=U^{(1)}_{ij}[m]+\delta_{m,0}\delta_{ij}l\beta,\\
		&\eta^{\otimes l}_\beta\left(U^{(2)}_{ij}[m]\right)=U^{(2)}_{ij}[m]+(l-1)\varkappa U^{(1)}_{ij}[m]+\frac{1}{2}\delta_{m,0}\delta_{ij}l(l-1)(\beta^2-\varkappa\beta).
	\end{split}
\end{equation}

To relate the parabolic induction with the coproduct of quiver Yangians, let us take $l=l_1+l_2$ and $k+l(M-N)=k_1+l_1(M-N)=k_2+l_2(M-N)$ such that $\varkappa$ remains the same for $\mathcal{W}_{M|N\times l}$ and $\mathcal{W}_{M|N\times l_{1,2}}$. Then there exists an inclusion map $\Delta_{l_1,l_2}:\mathcal{W}_{M|N\times l}\rightarrow\mathcal{W}_{M|N\times l_1}\otimes\mathcal{W}_{M|N\times l_2}$ that splits \eqref{miura} into two pieces of sizes $l_1$ and $l_2$ (see also (5.5) in \cite{Rapcak:2019wzw}). As computed in \cite{Rapcak:2019wzw,kodera2022coproduct}, we have\footnote{We are not adding extra labels $l$ and $l_{1,2}$ to these $U^{(s)}_{ij}$ as it should be clear which elements belong to which parts.}
\begin{equation}
	\begin{split}
		&\Delta_{l_1,l_2}\left(U^{(1)}_{ij}[m]\right)=U^{(1)}_{ij}[m]\otimes1+1\otimes U^{(1)}_{ij}[m],\\
		&\Delta_{l_1,l_2}\left(U^{(2)}_{ij}[m]\right)=U^{(2)}_{ij}[m]\otimes1+1\otimes U^{(2)}_{ij}[m]+\sum_{c=1}^{M+N}\sum_{n\in\mathbb{Z}}U^{(1)}_{in}[k]\otimes U^{(2)}_{nj}[m-k]-(m+1)l_1\varkappa1\otimes U^{(1)}_{ij}[m].
	\end{split}
\end{equation}
Let us also define the map $\widetilde{\Delta}_{l_1,l_2}=\left(\text{id}^{\otimes l_1}\otimes\eta^{\otimes l_2}_{-l_1\varkappa}\right)\circ\Delta_{l_1,l_2}$. Then following the same proof as in \cite{kodera2022coproduct}, we have the commutative diagram
\begin{equation}
	\includegraphics[width=7.5cm]{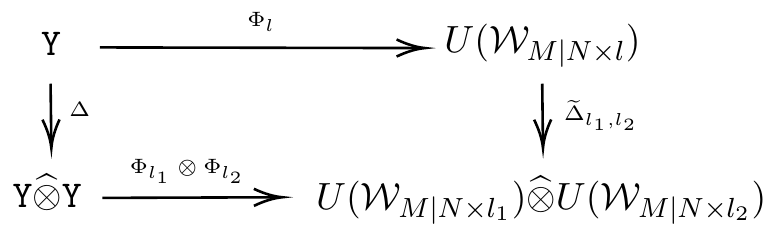},
\end{equation}
where we have labelled $\Phi$ with subscripts $l$ and $l_{1,2}$ for clarity.

\paragraph{More general truncations} As studied in \cite{Gaiotto:2017euk,Prochazka:2017qum,Prochazka:2018tlo} for the $\mathbb{C}^3$ case and \cite{Rapcak:2019wzw} for any generalized conifold, there exist larger families of truncations of the $\mathcal{W}$-algebras. These truncations, which are dictated by the functions $x^{l_3}y^{l_2}z^{l_4}w^{l_1}\in\mathbb{C}[x,y,z,w]/\left\langle xy=z^Mw^N\right\rangle$, can be built from $x$-, $y$-, $z$- and $w$-algebras associated to different divisors in the CY$_3$. In particular, the $x$-algebra just corresponds to the Miura operator of form \eqref{miura}. More generally, these truncations have generalized Miura/pseudo-differential operators of different types.

In terms of $(p,q)$-brane webs, certain stacks of D3s are stretched in different regions, indicating the multiplicities of smooth components of these divisors. See for example Figure 6 in \cite{Rapcak:2019wzw} for an illustration of the patterns of these elementary truncated algebras. The web diagram encodes the loci in the base where the $T^2$ part of the fibre degenerates to a circle in the (resolved) CY threefold. The complex coordinates that are used in the moment maps parametrizing the base can be grouped into variables $x,y,z,w$. This gives rise to the perspective of a 4d $\mathcal{N}=4$ gauge theory which is divided into a junction of four interfaces (or three for the $\mathbb{C}^3$ case).

The generators for the elementary building blocks of truncations can be found in \cite{Rapcak:2019wzw}. Here, we would just like to mention that with
\begin{equation}
	l=\frac{((N-M)\epsilon_1-\epsilon_2)l_3+\epsilon_2l_2-\epsilon_1l_4+\epsilon_1l_1}{(N-M)\epsilon_1-\epsilon_2},
\end{equation}
the same argument as in Theorem \ref{YWthm} indicates that we have this surjective homomorphism from the quiver Yangian to the general $x^{l_3}y^{l_2}z^{l_4}w^{l_1}$-algebra, which reflects the feature of the VOAs as truncations. Notice that now the generators $U^{(1)}_{ij}$ and $U^{(2)}_{ab}$ are given in \cite[(5.47)]{Rapcak:2019wzw}, satisfying the same OPEs as before.

\subsection{Crystal Melting}\label{crystal}
A consequence of Theorem \ref{YWthm} is that the universal enveloping algebra of $\mathcal{W}_{M|N\times l}$ admits (truncated) crystal representations. For quiver Yangians, the crystal melting representations encodes the information of BPS spectra with many salient features.

Let us first give a quick recap on crystal melting. The crystal model is a 3d uplift of the periodic quiver. The atoms in the crystal are in one-to-one correspondence with the gauge nodes in the quiver while the bifundamental/adjoint arrows are the ``chemical bonds'' connecting these atoms. Different atoms associated with different gauge nodes are then in different ``colours''.

Let us pick an initial atom $\mathfrak{o}$ in the periodic quiver as the first atom being molten from the empty room configuration. From the above map $\Phi$, a natural choice would be an atom having colour $a=0$. All the other atoms are then arranged level by level following the arrows/chemical bonds, at the positions of the corresponding nodes in the periodic quiver. There could be more than one paths from one atom to another, but they should be equivalent due to the F-term relations. This gives rise to the path algebra $\mathbb{C}Q/\langle\partial W\rangle$.

To recover the BPS counting, the crystal configurations should obey the melting rule. It states that an atom $\mathfrak{a}$ appears in the molten crystal $\mathfrak{C}$ if there exists an arrow $I$ such that $I\cdot\mathfrak{a}\in\mathfrak{C}$.

When considering the action of quiver Yangians on the crystal modules, it would be convenient to use the currents\footnote{Recall that we have set $\psi^{(a)}_{-1}=1/\epsilon_+$. Notice that this mode expansion here is only for symmetric quivers. For more general cases, there can be non-trivial $\psi^{(a)}_n$ with $n<0$.}
\begin{equation}
	\psi^{(a)}(z)=\sum_{n=-1}^{\infty}\frac{\psi^{(a)}_n}{z^{n+1}},\qquad e^{(a)}(z)=\sum_{n=0}^{\infty}\frac{e^{(a)}_n}{z^{n+1}},\qquad f^{(a)}(z)=\sum_{n=0}^{\infty}\frac{f^{(a)}_n}{z^{n+1}}.
\end{equation}
We can write down the OPEs for these currents from the commutation relations that define the quiver Yangians. They can be found for instance in \cite{Li:2020rij,Bao:2022fpk}. In particular, $e^{(a)}$ should act like a creation operator that add atoms of colour $a$ to the molten crystal while $f^{(a)}$ annihilates atoms in the crystal configuration. By analyzing how atoms can be added and removed following the melting rule, we can obtain the actions of these currents on any state $|\mathfrak{C}\rangle$. Let us write $\mathfrak{a}\in\mathfrak{C}_+$ (resp. $\mathfrak{a}\in\mathfrak{C}_-$) such that $|\mathfrak{C}\rangle$ would become $|\mathfrak{C}+\mathfrak{a}\rangle$ (resp. $|\mathfrak{C}-\mathfrak{a}\rangle$) after the atom $\mathfrak{a}$ with colour $a$ is added to (resp. removed from) $\mathfrak{C}$. From \cite{Li:2020rij}, we learn that the actions are
\begin{align}
	&\psi^{(a)}(z)|\mathfrak{C}\rangle=\Psi^{(a)}_{\mathfrak{C}}(z)|\mathfrak{C}\rangle,\\
	&e^{(a)}(z)|\mathfrak{C}\rangle=\sum_{\mathfrak{a}\in\mathfrak{C}_+}\frac{\pm\sqrt{-(-1)^{|a|}\text{Res}_{\widetilde{\epsilon}(\mathfrak{a})}\Psi^{(a)}_{\mathfrak{C}}(u)}}{z-\widetilde{\epsilon}(\mathfrak{a})}|\mathfrak{C}+\mathfrak{a}\rangle,\label{eaction}\\
	&f^{(a)}(z)|\mathfrak{C}\rangle=\sum_{\mathfrak{a}\in\mathfrak{C}_-}\frac{\pm\sqrt{\text{Res}_{\widetilde{\epsilon}(\mathfrak{a})}\Psi^{(a)}_{\mathfrak{C}}(u)}}{z-\widetilde{\epsilon}(\mathfrak{a})}|\mathfrak{C}-\mathfrak{a}\rangle,
\end{align}
where
\begin{align}
	&\Psi^{(a)}_{\mathfrak{C}}(z):=\left(\frac{z+C}{z}\right)^{\delta_{a,1}}\prod_{b\in Q_0}\prod_{\mathfrak{b}\in\mathfrak{C}}\phi^{b\Rightarrow a}(z-\widetilde{\epsilon}(\mathfrak{b})),\\
	&\phi^{b\Rightarrow a}(z)=\frac{\prod\limits_{I\in a\rightarrow b}(z+\widetilde{\epsilon}_I)}{\prod\limits_{I\in b\rightarrow a}(z-\widetilde{\epsilon}_I)},\label{bondfactor}\\
	&\widetilde{\epsilon}(\mathfrak{a})=\sum_{I\in\text{path}[\mathfrak{o}\rightarrow\mathfrak{a}]}\widetilde{\epsilon}_I.
\end{align}
Here, the numerical constant $C$ is the vacuum charge. For the cases considered in this paper whose toric CY$_3$ have no compact divisors, it can be identified as the central term $\sum\limits_{a\in Q_0}\psi^{(a)}_0$. Later, we shall relate this to the parameters on the $\mathcal{W}$-algebra side. The $\pm$ signs in the actions depend on the statistics of the quiver Yangian (see \cite[\S6]{Li:2020rij} for more details). Recall that the charge assignment $\widetilde{\epsilon}_I$ should satisfy the constraints discussed in \S\ref{QY} (and hence leads to the two parameters $\epsilon_{1,2}$ in the algebra). The action of each mode on the crystal can then be obtained via the corresponding contour integral around $\infty$.

As argued in \cite{Li:2020rij,Galakhov:2021xum}, the truncations of quiver Yangians lead to truncated crystal configurations where the melting would stop at one or more atoms. Therefore, with the map $\Phi$, we may consider the truncated crystals as modules of $U(\mathcal{W}_{M|N\times l})$.

To get the actions of $U^{(s)}_{ij}$ on the truncated crystals, we need to know how they can be expressed using the map $\Phi$. Here, let us find such expressions for all the elements at spin $s=1,2$, which essentially gives the proof of the surjectivity of $\Phi$.

First, let us consider $U^{(1)}_{ab}[m]$ for any $a,b=1,\dots,M+N$ and $m\in\mathbb{Z}$. The zero modes of $\mathcal{H},\mathcal{E},\mathcal{F}$ already gives $U^{(1)}_{a,a+1}[0]$, $U^{(1)}_{a+1,a}[0]$ (with two exceptions) and $U^{(1)}_{aa}[0]-U^{(1)}_{a+1,a+1}[0]$. It is immediate to obtain $U^{(1)}_{aa}[0]-U^{(1)}_{bb}[0]$ for any $a,b$. Using the commutation relations
\begin{equation}
	\begin{cases}
		-(-1)^{p(a+1)}\left[U^{(1)}_{a,a+1}[0],U^{(1)}_{a+1,a+2}[0]\right]=U^{(1)}_{a,a+2}[0],\\
		(-1)^{p(a)p(a+1)+p(a+1)p(a+2)+p(a)p(a+2)}\left[U^{(1)}_{a+1,a}[0],U^{(1)}_{a+2,a+1}[0]\right]=U^{(1)}_{a+2,a}[0]
	\end{cases}
\end{equation}
iteratively, we can get $U^{(1)}_{ab}[0]$ for any $a\neq b$ (including $U^{(1)}_{1,M+N}[0]$ and $U^{(1)}_{M+N,1}[0]$). This is consistent with the fact that $e^{(\alpha)}_0$ (and likewise for $f^{(\alpha)}_0$) is of form $\left[\dots\left[\left[e^{(a_1)}_0,e^{(a_2)}_0\right],e^{(a_3)}_0\right]\dots,e^{(a_n)}_0\right]$ when $\alpha=\alpha^{(a_1)}+\dots+\alpha^{(a_n)}$. Now, using $\Phi\left(\mathcal{E}^{(a)}_0\right)$ and $\Phi\left(\mathcal{F}^{(a)}_0\right)$, we can write
\begin{equation}
	\begin{cases}
		U^{(1)}_{M+N,a}[-1]=-(-1)^{p(1)}\left[U^{(1)}_{M+N,1}[-1],U^{(1)}_{1,a}[0]\right],\\
		U^{(1)}_{a,M+N}[-1]=(-1)^{p(1)p(M+N)+p(a)p(1)+p(a)p(M+N)}\left[U^{(1)}_{1,M+N}[1],U^{(1)}_{a,1}[0]\right]
	\end{cases}
\end{equation}
for $a\neq M+N$. Hence, via
\begin{equation}
	\begin{cases}
		U^{(1)}_{ba}[-1]=(-1)^{p(M+N)p(a)+p(M+N)p(b)+p(a)p(b)}\left[U^{(1)}_{M+N,a}[-1],U^{(1)}_{b,M+N}[0]\right],\\
		U^{(1)}_{ab}[1]=-(-1)^{p(M+N)}\left[U^{(1)}_{a,M+N}[1],U^{(1)}_{M+N,b}[0]\right],
	\end{cases}
\end{equation}
we can get $U^{(1)}_{ab}[\pm1]$ for any $a\neq b$. Keep this procedure, and we can obtain $U^{(1)}_{ab}[m]$ for any $a\neq b$ and $m\in\mathbb{Z}$.

For elements of spin 1, we are now left with $U^{(1)}_{aa}[m]$. Take
\begin{equation}
	\begin{split}
		&X_{ab}[m]:=\Bigg[U^{(1)}_{ab}[m],U^{(2)}_{bb}[0]-U^{(2)}_{b+1,b+1}[0]-U^{(1)}_{bb}[0]U^{(1)}_{b+1,b+1}[0]\\
		&\qquad\qquad\quad-\sum_{k\geq0}U^{(1)}_{bb}[-k]U^{(1)}_{bb}[k]+\sum_{k\geq0}U^{(1)}_{b+1,b+1}[-k-1]U^{(1)}_{b+1,b+1}[k+1]\Bigg]
	\end{split}\label{Xabm}
\end{equation}
for $a\neq b$. We are allowed to use this commutation relation because of $\Phi\left(\mathcal{H}^{(b)}_1\right)$. A straightforward computation yields
\begin{equation}
	\begin{split}
		X_{ab}[m]=&-(-1)^{p(b)}U^{(2)}_{ab}[m]+(-1)^{p(b)}U^{(1)}_{ab}[m]U^{(1)}_{b+1,b+1}[0]\\
		&+\sum_{k\geq0}(-1)^{p(b)}U^{(1)}_{ab}[m-k]U^{(1)}_{bb}[k]-\sum_{k\geq0}(-1)^{p(b)}U^{(1)}_{bb}[-k]U^{(1)}_{ab}[m+k].
	\end{split}
\end{equation}
Therefore,
\begin{equation}
	\begin{split}
		&\left[U^{(1)}_{ba}[0],(-1)^{p(b)}X_{ab}[m]\right]-\left[U^{(1)}_{ba}[1],(-1)^{p(b)}X_{ab}[m-1]\right]\\
		=&-(l-1)(-1)^{p(a)}\varkappa U^{(1)}_{bb}[m]-\delta_{m,0}l(-1)^{p(a)}\varkappa U^{(1)}_{b+1,b+1}[0]-\delta_{m\geq0}l(-1)^{p(a)}\varkappa U^{(1)}_{bb}[m]\\
		&+\delta_{m\leq0}l(-1)^{p(a)}\varkappa U^{(1)}_{bb}[m]+(-1)^{p(b)}U^{(1)}_{ab}[m]U^{(1)}_{ba}[0]-(-1)^{p(b)}U^{(1)}_{ba}[0]U^{(1)}_{ab}[m].
	\end{split}\label{UXUX}
\end{equation}
Notice that most terms appeared in the calculation get cancelled since $0+m=1+m-1$. When $m>0$, \eqref{UXUX} is equal to
\begin{equation}
	-(-1)^{p(a)}\varkappa U^{(1)}_{bb}[m]-l(-1)^{p(a)}\varkappa U^{(1)}_{bb}[m]+(-1)^{p(b)}U^{(1)}_{ab}[m]U^{(1)}_{ba}[0]-(-1)^{p(b)}U^{(1)}_{ba}[0]U^{(1)}_{ab}[m].
\end{equation}
As a result, we obtain $U^{(1)}_{bb}[m]$ for $m>0$. This is likewise for $m<0$. When $m=0$, \eqref{UXUX} is
\begin{equation}
	-(-1)^{p(a)}\varkappa U^{(1)}_{bb}[0]-l(-1)^{p(a)}\varkappa U^{(1)}_{b+1,b+1}[0]+(-1)^{p(b)}U^{(1)}_{ab}[0]U^{(1)}_{ba}[0]-(-1)^{p(b)}U^{(1)}_{ba}[0]U^{(1)}_{ab}[m].
\end{equation}
This in particular gives $-(-1)^{p(a)}\varkappa U^{(1)}_{bb}[0]-l(-1)^{p(a)}\varkappa U^{(1)}_{b+1,b+1}[0]$. Together with $U^{(1)}_{bb}[0]-U^{(1)}_{b+1,b+1}[0]$, we can get $U^{(1)}_{bb}[0]$.

Now, we shall consider the elements of spin 2. From $\Phi\left(\mathcal{H}^{(a)}_1\right)$, $\Phi\left(\mathcal{E}^{(a)}_1\right)$ and $\Phi\left(\mathcal{F}^{(a)}_1\right)$, we get $U^{(2)}_{a,a+1}[0]$, $U^{(2)}_{(a+1,a)}[0]$ (with two exceptions) and $U^{(2)}_{aa}[0]-U^{(2)}_{a+1,a+1}[0]$. Similar to the case of $U^{(1)}_{ab}[m]$, we can then obtain $U^{(2)}_{ab}[m]$, as well as $U^{(2)}_{aa}[m]-U^{(2)}_{bb}[m]$, for any $a\neq b$ and $m\in\mathbb{Z}$ using the $U^{(1)}U^{(2)}$ commutation relation.

To get the remaining elements $U^{(2)}_{aa}[m]$, let us compute
\begin{equation}
	\left[U^{(2)}_{aa}[1]-U^{(2)}_{bb}[1],U^{(2)}_{aa}[m]-U^{(2)}_{bb}[m]\right]-\left[U^{(2)}_{aa}[1]-U^{(2)}_{bb}[1],U^{(2)}_{aa}[m+1]-U^{(2)}_{bb}[m+1]\right]
\end{equation}
for $a\neq b$. This could be tedious due to all the $U^{(2)}U^{(2)}$ commutation relations, but we notice that most of the terms can be cancelled, and it becomes
\begin{align}
	&-2\varkappa(-1)^{p(a)}U^{(2)}_{aa}[m+1]-2\varkappa(-1)^{p(b)}U^{(2)}_{bb}[m+1]\nonumber\\
	&+(l-1)\left((1+(-1)^{p(a)})\left(\sum_{k>0}U^{(1)}_{aa}[-k]U^{(1)}_{aa}[m+1+k]+\sum_{k\geq0}U^{(1)}_{aa}[m+1-k]U^{(1)}_{aa}[k]\right)\right.\nonumber\\
	&+(1+(-1)^{p(b)})\left(\sum_{k>0}U^{(1)}_{bb}[-k]U^{(1)}_{bb}[m+1+k]-\sum_{k\geq0}U^{(1)}_{bb}[m+1-k]U^{(1)}_{bb}[k]\right)\nonumber\\
	&-\sum_{k>0}U^{(1)}_{aa}[-k]U^{(1)}_{bb}[m+1+k]-\sum_{k\geq0}U^{(1)}_{bb}[m+1-k]U^{(1)}_{aa}[k]\nonumber\\
	&-\sum_{k>0}U^{(1)}_{bb}[-k]U^{(1)}_{aa}[m+1+k]-\sum_{k\geq0}U^{(1)}_{aa}[m+1-k]U^{(1)}_{bb}[k]\nonumber\\
	&-(-1)^{p(b)}\varkappa\sum_{k>0}U^{(1)}_{ab}[-k]U^{(1)}_{ba}[m+1-k]-(-1)^{p(a)}\varkappa\sum_{k\geq0}U^{(1)}_{ba}[m+1-k]U^{(1)}_{ab}[k]\nonumber\\
	&\left.-(-1)^{p(a)}\varkappa\sum_{k>0}U^{(1)}_{ba}[-k]U^{(1)}_{ab}[m+1-k]-(-1)^{p(b)}\varkappa\sum_{k\geq0}U^{(1)}_{ab}[m+1-k]U^{(1)}_{ba}[k]\right)\nonumber\\
	&-l(l-1)(m+2)\varkappa^2\left((-1)^{p(a)}U^{(1)}_{aa}[m+1]+(-1)^{p(b)}U^{(1)}_{bb}[m+1]\right).
\end{align}
From this, we get $(-1)^{p(a)}U^{(2)}_{aa}[n]+(-1)^{p(b)}U^{(2)}_{bb}[n]$. Choose $a,b$ such that $p(a)=p(b)$ (which is always possible for the cases we focus on in this paper). Together with $U^{(2)}_{aa}[n]-U^{(2)}_{bb}[n]$, we can obtain $U^{(2)}_{aa}[n]$ for any $a=1,\dots,M+N$ and $n\in\mathbb{Z}$.

Now, we can in principle write the actions of any $U^{(s)}_{ab}[m]$ on the (truncated) crystals. Since the crystal configuration always starts from the empty state $|\varnothing\rangle$ on which only the $e^{(0)}_0$ and $\psi^{(0)}_0$ modes would have non-trivial action. It is natural to wonder whether the truncated crystal can be a highest weight representation of $U(\mathcal{W}_{M|N\times l})$ with $|\varnothing\rangle$ being the highest weight vector. In particular, all the modes $U^{(s)}_{ab}[m]$ with $s\in\mathbb{Z}_+$, $a,b=1,\dots,M+N$ and $m>0$ should annihilate the highest weight state. To make this state unique, we also need $U^{(1)}_{ab}[0]$ to act trivially on it for all $a>b$ in our convention here.
\begin{corollary}
	The truncated crystal is a highest weight representation of $U(\mathcal{W}_{M|N\times l})$ with the empty room configuration $|\varnothing\rangle$ as the highest weight state.
\end{corollary}
\begin{proof}
	As pointed out in \cite{Eberhardt:2019xmf}, following the OPEs, it suffices to show that $U^{(1)}_{a+1,a}[0]$, $U^{(1)}_{1,M+N}[1]$ and $U^{(r)}_{11}[1]$ ($r=1,2,3$) would annihilate $|\varnothing\rangle$. This is obvious for $U^{(1)}_{a+1,a}[0]$ since it is mapped from $\mathcal{F}^{(a)}_0=f^{(a)}_0$ ($a\neq0$) which acts trivially on $|\varnothing\rangle$. Likewise, $U^{(1)}_{1,M+N}[1]$ comes from $\mathcal{F}^{(0)}_0=f^{(0)}_0$.
	
	To show that the remaining elements have zero eigenvalues, we first notice that $U^{(1)}_{ab}[m]|\varnothing\rangle=0$ for any $a>b$ and $m\geq0$. This follows from the commutation relations we used to construct these elements when showing the surjectivity of $\Phi$ above. On the other hand, $\left(U^{(1)}_{bb}[0]-U^{(1)}_{b+1,b+1}[0]\right)|\varnothing\rangle=0$ for $b\neq0$. This is because $\psi^{(a)}_0|\varnothing\rangle=\delta_{a,0}C|\varnothing\rangle$ from the action of $\psi^{(a)}(u)$. Besides, we have $\mathcal{H}^{(a)}_1|\varnothing\rangle=\frac{1}{2}\nu(b)\epsilon_-\delta_{a,0}C|\varnothing\rangle$ since $\psi^{(a)}_1|\varnothing\rangle=0$. With these results, we find that $X_{ab}[m]|\varnothing\rangle=0$ for $a>b$ and $m=0,1$, where $X_{ab}[m]$ is defined in \eqref{Xabm}. Therefore,
	\begin{equation}
		\begin{split}
			&\left[U^{(1)}_{ba}[0],(-1)^{p(b)}X_{ab}[m]\right]-\left[U^{(1)}_{ba}[1],(-1)^{p(b)}X_{ab}[m-1]\right]|\varnothing\rangle\\
			=&-(-1)^{p(a)}X_{ab}[m]U^{(1)}_{ba}[0]|\varnothing\rangle+(-1)^{p(a)}X_{ab}[m-1]U^{(1)}_{ba}[1]|\varnothing\rangle
		\end{split}\label{UXUXaction}
	\end{equation}
    from \eqref{UXUX} for $b<a$.
    
    Let us take $a=3$ and $b=1$. Then $U^{(1)}_{13}|\varnothing\rangle\propto\left[U^{(1)}_{12}[0],U^{(1)}_{23}[0]\right]|\varnothing\rangle=0$ since the only non-trivial $e^{(c)}_0|\varnothing\rangle$ is $c=0$. This indicates that the first term in \eqref{UXUXaction} vanishes. Repeat the similar procedure iteratively, and we find that $U^{(1)}_{1,M+N}[0]|\varnothing\rangle=0$. Suppose $M+N\neq3$, then $U^{(1)}_{13}[1]|\varnothing\rangle\propto\left[U^{(1)}_{1,M+N}[1],U^{(1)}_{M+N,3}[0]\right]|\varnothing\rangle=0$. If $M+N=3$, then $U^{(1)}_{13}[1]|\varnothing\rangle=0$ is automatic since this comes from $\mathcal{F}^{(0)}_0|\varnothing\rangle=0$. As a result, the second term in \eqref{UXUXaction}, and hence the whole equation, would vanish.
    
    Now, take $m=1$. Recall that the action in \eqref{UXUXaction} can be written using \eqref{UXUX}. As the latter two terms in \eqref{UXUX} vanish when acting on $|\varnothing\rangle$ as we have shown, the only two terms that survive under $m=1$ would give a term proportional to $U^{(1)}_{11}[1]$. Since \eqref{UXUXaction} vanishes, we have $U^{(1)}_{11}[1]|\varnothing\rangle=0$. Moreover, $U^{(1)}_{1,a}|\varnothing\rangle\propto\left[U^{(1)}_{11}[1],U^{(1)}_{1,a}[0]\right]|\varnothing\rangle=0$ for $a>1$. Keep this procedure, and we find that $U^{(1)}_{1,a}[m]|\varnothing\rangle=0$ for any $m>0$. Together with the expression of $\Phi\left(\mathcal{F}^{(0)}_1\right)$, we have $U^{(2)}_{1,M+N}[1]|\varnothing\rangle=0$. Likewise, by considering $\mathcal{F}^{(a)}_2$ from $\left[\mathcal{H}^{(a+1)}_1,\mathcal{F}^{(a)}_1\right]$, we can get $U^{(3)}_{1,M+N}[1]|\varnothing\rangle=0$ following the same steps. Notice that this uses the general $U^{(2)}_{ab}U^{(2)}_{cd}$ commutation relation which is not explicitly listed here, but it is straightforward from the OPE in \cite{Rapcak:2019wzw}.
\end{proof}

When considering $\mathcal{H}^{(a\neq0)}_0|\varnothing\rangle=\psi^{(a)}_0|\varnothing\rangle=0$, we can see that $U^{(1)}_{11}|\varnothing\rangle=U^{(1)}_{22}|\varnothing\rangle=\dots=U^{(1)}_{M+N,M+N}|\varnothing\rangle$. On the other hand, $\left(U^{(1)}_{M+N,M+N}[0]-U^{(1)}_{11}[0]+l\varkappa\right)|\varnothing\rangle=\mathcal{H}^{(0)}_0|\varnothing\rangle=C|\varnothing\rangle$. This then relates the parameter on the $\mathcal{W}$-algebra side with the vacuum charge in the quiver Yangian:
\begin{corollary}
	We have $C=l\varkappa$.
\end{corollary}

\paragraph{Example} Let us illustrate the above discussions with an example. The simplest case would be $\mathbb{C}\times\mathbb{C}^2/\mathbb{Z}_3$ whose quiver and crystal model are depicted in Figure \ref{CC2Z3}.
\begin{figure}[h]
	\centering
	\includegraphics[width=12cm]{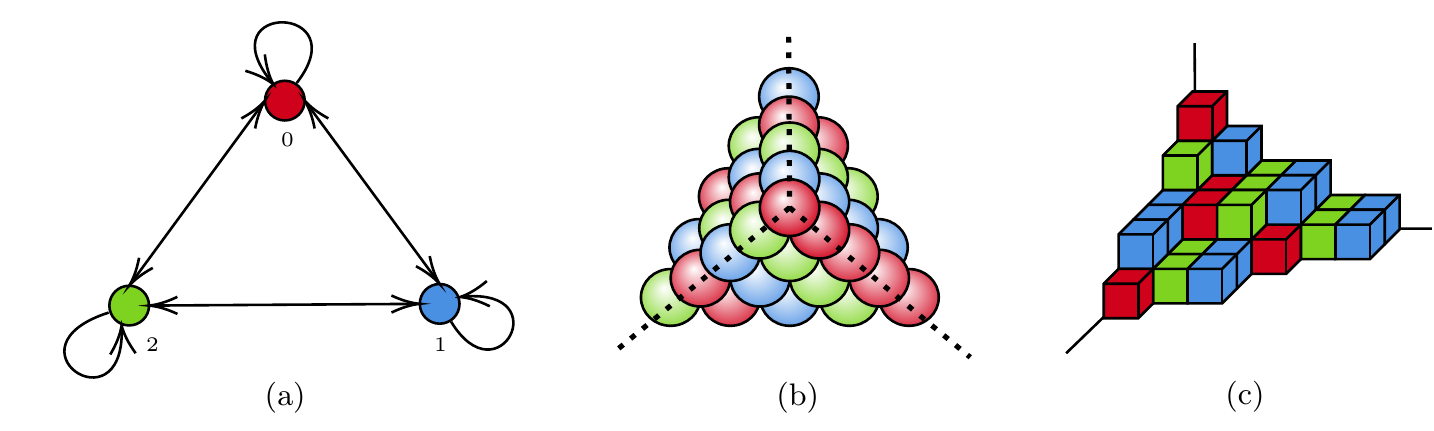}
	\caption{(a) The quiver for $\mathbb{C}\times\mathbb{C}^2/\mathbb{Z}_3$. (b) The corresponding crystal model with three colours. The dased lines are the ridges of the crystal. (c) The equivalent crystal model visualized using coloured plane partitions.}\label{CC2Z3}
\end{figure}
The possible configurations with the corresponding modes acting on $|\varnothing\rangle$ at low levels are listed in \cite{Li:2020rij}. For instance, take $|\mathfrak{C}\rangle$ to be
\begin{equation}
	\includegraphics[width=4.5cm]{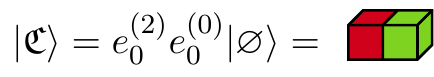}.
\end{equation}
One possible way to add an atom is to act $\mathcal{E}^{(1)}_0=e^{(1)}_0$. Using the action of $e^{(1)}(z)$ from \eqref{eaction} above, we have
\begin{equation}
	\Psi^{(1)}_{\mathfrak{C}}=\frac{\left(z+\widetilde{\epsilon}_{10}\right)\left(z-\widetilde{\epsilon}_{02}+\widetilde{\epsilon}_{12}\right)}{\left(z-\widetilde{\epsilon}_{01}\right)\left(z-\widetilde{\epsilon}_{02}-\widetilde{\epsilon}_{21}\right)},
\end{equation}
where we have kept the notation $\widetilde{\epsilon}_I$. Therefore,
\begin{equation}
	\mathcal{E}^{(1)}_0|\mathfrak{C}\rangle=-\left(\frac{\left(\widetilde{\epsilon}_{10}+\widetilde{\epsilon}_{02}+\widetilde{\epsilon}_{21}\right)\left(\widetilde{\epsilon}_{12}+\widetilde{\epsilon}_{21}\right)}{\widetilde{\epsilon}_{01}-\widetilde{\epsilon}_{02}-\widetilde{\epsilon}_{21}}\right)^{1/2}|\mathfrak{C}+\mathfrak{1}\rangle,
\end{equation}
where we have taken $+$ in $\pm$ in \eqref{eaction}. This gives
\begin{equation}
	U^{(1)}_{12}[0]|\mathfrak{C}\rangle=\frac{3\epsilon_2(\epsilon_1+\epsilon_2)}{\epsilon_1-2\epsilon_2}|\mathfrak{C}+\mathfrak{1}\rangle,
\end{equation}
where
\begin{equation}
	\includegraphics[width=3cm]{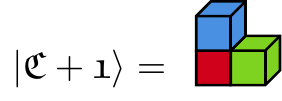}.
\end{equation}
One may apply the procedure above used for constructing other elements and get their actions. Here, we list a few examples for $U^{(1)}_{ij}[m]$:
\begin{equation}
	\includegraphics[width=14cm]{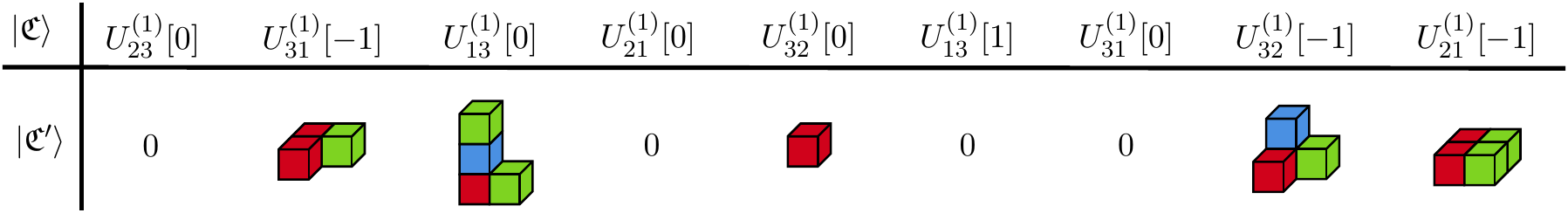},
\end{equation}
where we omit the coefficients and only show the crystal configurations. By considering higher modes of the quiver Yangian generators, we can also get the actions of $U^{(s)}$ with higher spins.

Let us take a brief look at how the truncations of the crystal could happen. Here, we shall only discuss the simplest example which truncates the algebra at $l=1$. In such case, we just have the universal enveloping algebra of the Kac-Moody algebra with only zero modes for the quiver Yangian, or equivalently, only spin 1 elements for the $\mathcal{W}$-algebra. It is straightforward to see that the truncated crystal has the shape
\begin{equation}
	\includegraphics[width=6cm]{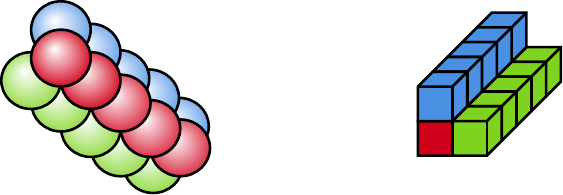}
\end{equation}
of semi-infinite length. We simply have $U^{(s\geq2)}_{ab}[m]|\mathfrak{C}\rangle=0$ for any $a,b\in\{1,2,3\}$, $m\in\mathbb{Z}$ and any configuration $\mathfrak{C}$ as $U^{(s)}_{ab}[m]$ vanishes for $s>l$.

In general, for any quiver Yangian we focus on in this paper, the truncation at the very first level $l=1$ can be desribed in this manner. For more general truncations of the crystal and larger $l$, this could be more involved, and we leave this to future work (see also \S\ref{outlook}).

\section{Outlook}\label{outlook}
We have discussed a few aspects of the quiver Yangians for most generalized conifolds and their connections to vertex algebras. Let us mention some possible directions for further research.

The construction of the coproduct we have benefits from the untwisted affine Lie superalgebra of A-type. A natural extension would be a more thorough study on quiver Yangians for the remaining generalized conifolds, as well as $\mathbb{C}^3/(\mathbb{Z}_2\times\mathbb{Z}_2)$. All of them have underlying Kac-Moody algebras. It is worth noting that a method of computing the coproduct perturbatively is given in \cite[Appendix C]{Galakhov:2022uyu} for generalized conifolds. More generally, for toric CYs with compact 4-cycles, the quiver Yangians do not seem to have such underlying Kac-Moody algebras. It would be desirable to find the coproduct of the algebra associated to any quiver.

The coproduct plays an important role when studying the Bethe/gauge correspondence from the BPS quiver algebra \cite{Galakhov:2022uyu,Bao:2022fpk}. In particular, there are obstructions for the Bethe/gauge correspondence to work in the chiral quiver cases as pointed out in \cite{Galakhov:2022uyu}, and this requires further investigations. Similar to the map from the quiver Yangian to the $\mathcal{W}$-algebra, it would be natural to wonder whether we can write some $\mathcal{RTT}$-like presentation of the quiver Yangian, which might in turn shed light on the Bethe/gauge correspondence.

One may also consider the trigonometric and elliptic versions of the quiver Yangians \cite{Galakhov:2021vbo}. They bear resemblance to the rational ones by means of the generators and their relations. However, the coproduct structures seem to behave very different.

We showed that the quiver Yangians discussed in this paper are isomorphic algebras under toric duality by virtue of the odd reflections of the Kac-Moody superalgebras. We expect similar isomorphic maps for quiver Yangians associated to the CYs without compact divisors that are not explored here. However, new methods are required when considering toric CYs with compact divisors. It is natural to conjecture that such quiver Yangians would still be isomorphic as their supersymmetric gauge theories are related by Seiberg duality. It could even be possible to consider the quivers outside the toric phases which can still be reached via Seiberg duality. We can always define the corresponding quiver Yangians from the quiver data though whether/how they would implement the BPS algebras would need further checks. Mathematically, Seiberg dual quivers essentially transform under mutations. Therefore, cluster algebras might be useful in proving the isomorphisms.

Regarding the truncations and VOAs, we have only discussed the most trivial case for truncations here with $l=1$. As analyzed in \cite{Li:2020rij}, when the crystal is truncated at some atom, we have the corresponding residue vanishing in the numerator of \eqref{eaction}. This leads to some extra conditions that $\widetilde{\epsilon}_I$ should satisfy on the quiver Yangian side. On the other hand, recall that we also have certain condition on $\epsilon_i$ for $\Phi$ to be homomorphic, and the truncation comes from the parameter $l$ on the $\mathcal{W}$-algebra side. We expect that the cut at $l$ would not provide all the possible truncations of the crystal. It is very likely that the coefficients in the actions of some $U$-modes become zero due to the truncation conditions on $\widetilde{\epsilon}_I$ from the quiver Yangians. Besides, there are more general truncations, namely the $x^{l_3}y^{l_2}z^{l_4}w^{l_1}$-algebras, as mentioned above. They might also give possible truncations on the crystal. Moreover, it would be interesting to see whether other crystal configurations, such as crystals in other chambers \cite{Yamazaki:2010fz,Aganagic:2010qr,Bao:2022oyn} and 2d crystals \cite{Nishinaka:2013mba,Galakhov:2022uyu,Bao:2022fpk}, could give similar relations.

It still remains an open question whether the quiver Yangians for more general geometry, especially those associated to toric CYs with compact divisors, could have some $\mathcal{W}$-algebras as their truncations. It might be possible to construct the VOAs from the quiver Yangians in this setting and compare them with other constructions. This could provide more insights into the BPS/CFT correspondence.

\section*{Acknowledgement}
The research is supported by a CSC scholarship.

\appendix

\section{Generators of Quiver Yangians for $\mathbb{C}\times\mathbb{C}^2/\mathbb{Z}_2$ and Conifold}\label{generators2}
Analogous to the cases discussed in the main context, the generators of the quiver Yangians with two gauge nodes can also be expressed using finitely many modes. The main difference is that there are more than one pairs of arrows connecting the two nodes in the quiver. See Figure \ref{twonodes}. Below we shall use the generators in the convention of $\uppsi$, $\mathtt{e}$ and $\mathtt{f}$.
\begin{figure}[h]
	\centering
	\includegraphics[width=10cm]{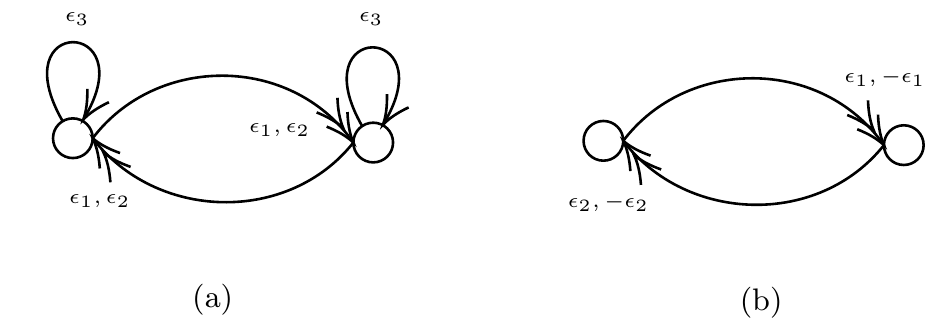}
	\caption{(a) The $\mathbb{C}\times\mathbb{C}^2/\mathbb{Z}_2$ quiver with charge assignment. (b) The conifold quiver with charge assignment.}\label{twonodes}
\end{figure}

\paragraph{Case 1: $\mathbb{C}\times\mathbb{C}^2/\mathbb{Z}_2$} We have
\begin{equation}
	\sigma_1^{ab}=\sigma_1^{ba}=\begin{cases}
		\epsilon_1+\epsilon_2=-\epsilon_3,&a\neq b\\
		\epsilon_3,&a=b,
	\end{cases}
    \qquad
    \sigma_2^{ab}=\sigma_2^{ba}=\begin{cases}
    	\epsilon_1\epsilon_2,&a\neq b\\
    	0,&a=b.
    \end{cases}
\end{equation}
The defining relations of the quiver Yangian are
\begin{align}
	&\left[\uppsi^{(a)}_n,\uppsi^{(b)}_m\right]=0,\\
	&\left[\mathtt{e}^{(a)}_n,\mathtt{f}^{(b)}_m\right]=\delta_{ab}\uppsi^{(a)}_{m+n},\\
	&\left[\uppsi^{(a)}_{n+1},\mathtt{e}^{(a)}_m\right]-\left[\uppsi^{(a)}_n,\mathtt{e}^{(a)}_{m+1}\right]=\epsilon_3\left\{\uppsi^{(a)}_n,\mathtt{e}^{(a)}_m\right\}\\
	&\left[\uppsi^{(a)}_{n+1},\mathtt{f}^{(a)}_m\right]-\left[\uppsi^{(a)}_n,\mathtt{f}^{(a)}_{m+1}\right]=-\epsilon_3\left\{\uppsi^{(a)}_n,\mathtt{f}^{(a)}_m\right\}\\
	&\left[\mathtt{e}^{(a)}_{n+1},\mathtt{e}^{(a)}_m\right]-\left[\mathtt{e}^{(a)}_n,\mathtt{e}^{(a)}_{m+1}\right]=\epsilon_3\left\{\mathtt{e}^{(a)}_n,\mathtt{e}^{(a)}_m\right\}\\
	&\left[\mathtt{f}^{(a)}_{n+1},\mathtt{f}^{(a)}_m\right]-\left[\mathtt{f}^{(a)}_n,\mathtt{f}^{(a)}_{m+1}\right]=-\epsilon_3\left\{\mathtt{f}^{(a)}_n,\mathtt{f}^{(a)}_m\right\}\\
	&\left[\uppsi^{(a)}_{n+2},\mathtt{e}^{(a+1)}_m\right]-2\left[\uppsi^{(a)}_{n+1},\mathtt{e}^{(a+1)}_{m+1}\right]+\left[\uppsi^{(a)}_n,\mathtt{e}^{(a+1)}_{m+2}\right]\nonumber\\
	=&-\epsilon_1\epsilon_2\left[\uppsi^{(a)}_n,\mathtt{e}^{(a+1)}_m\right]-\epsilon_3\left\{\uppsi^{(a)}_{n+1},\mathtt{e}^{(a+1)}_m\right\}+\epsilon_3\left\{\uppsi^{(a)}_n,\mathtt{e}^{(a+1)}_{m+1}\right\},\label{psiaea+1CC2Z2}\\
	&\left[\uppsi^{(a)}_{n+2},\mathtt{f}^{(a+1)}_m\right]-2\left[\uppsi^{(a)}_{n+1},\mathtt{f}^{(a+1)}_{m+1}\right]+\left[\uppsi^{(a)}_n,\mathtt{f}^{(a+1)}_{m+2}\right]\nonumber\\
	=&-\epsilon_1\epsilon_2\left[\uppsi^{(a)}_n,\mathtt{f}^{(a+1)}_m\right]+\epsilon_3\left\{\uppsi^{(a)}_{n+1},\mathtt{f}^{(a+1)}_m\right\}-\epsilon_3\left\{\uppsi^{(a)}_n,\mathtt{f}^{(a+1)}_{m+1}\right\},\\
	&\left[\mathtt{e}^{a}_{n+2},\mathtt{e}^{(a+1)}_m\right]-2\left[\mathtt{e}^{(a)}_{n+1},\mathtt{e}^{(a+1)}_{m+1}\right]+\left[\mathtt{e}^{(a)}_n,\mathtt{e}^{(a+1)}_{m+2}\right]\nonumber\\
	=&-\epsilon_1\epsilon_2\left[\mathtt{e}^{(a)}_n,\mathtt{e}^{(a+1)}_m\right]-\epsilon_3\left\{\mathtt{e}^{(a)}_{n+1},\mathtt{e}^{(a+1)}_m\right\}+\epsilon_3\left\{\mathtt{e}^{(a)}_n,\mathtt{e}^{(a+1)}_{m+1}\right\},\\
	&\left[\mathtt{f}^{a}_{n+2},\mathtt{f}^{(a+1)}_m\right]-2\left[\mathtt{f}^{(a)}_{n+1},\mathtt{f}^{(a+1)}_{m+1}\right]+\left[\mathtt{f}^{(a)}_n,\mathtt{f}^{(a+1)}_{m+2}\right]\nonumber\\
	=&-\epsilon_1\epsilon_2\left[\mathtt{f}^{(a)}_n,\mathtt{f}^{(a+1)}_m\right]+\epsilon_3\left\{\mathtt{f}^{(a)}_{n+1},\mathtt{f}^{(a+1)}_m\right\}-\epsilon_3\left\{\mathtt{f}^{(a)}_n,\mathtt{f}^{(a+1)}_{m+1}\right\},\\
	&\text{Sym}_{n_1,n_2}\left[\mathtt{e}^{(a)}_{n_1},\left[\mathtt{e}^{(a)}_{n_2},\mathtt{e}^{(a+1)}_m\right]\right]=\text{Sym}_{n_1,n_2}\left[\mathtt{f}^{(a)}_{n_1},\left[\mathtt{f}^{(a)}_{n_2},\mathtt{f}^{(a+1)}_m\right]\right]=0.
\end{align}
Recall that $\uppsi^{(a)}_{-1}=1$. Take $n=-2$ in \eqref{psiaea+1CC2Z2}, and we have
\begin{equation}
	\left[\uppsi^{(a)}_0,\mathtt{e}^{(a+1)}_m\right]=-2\epsilon_3\mathtt{e}^{(a+1)}_m,\quad\text{i.e.,}\quad\left[\uppsi^{(a)}_0,\mathtt{e}^{(a+1)}_{m+1}\right]=-2\epsilon_3\mathtt{e}^{(a+1)}_{m+1}.
\end{equation}
Then take $n=-1$ in \eqref{psiaea+1CC2Z2}, and we get
\begin{equation}
	\left[\uppsi^{(a)}_1,\mathtt{e}^{(a+1)}_m\right]-2\left[\uppsi^{(a)}_0,\mathtt{e}^{(a+1)}_{m+1}\right]=-\epsilon_3\left\{\uppsi^{(a)}_0,\mathtt{e}^{(a+1)}_m\right\}+2\epsilon_3\mathtt{e}^{(a+1)}_{m+1}.
\end{equation}
Therefore,
\begin{equation}
	\begin{split}
		\mathtt{e}^{(a+1)}_{m+1}=&-\frac{1}{6\epsilon_3}\left[\uppsi^{(a)}_1,\mathtt{e}^{(a+1)}_m\right]+\frac{1}{6}\uppsi^{(a)}_0\mathtt{e}^{(a+1)}_m+\frac{1}{6}\mathtt{e}^{(a+1)}_m\uppsi^{(a)}_0\\
		=&-\frac{1}{6\epsilon_3}\left[\uppsi^{(a)}_1,\mathtt{e}^{(a+1)}_m\right]-\frac{1}{6}\uppsi^{(a)}_0\frac{1}{2\epsilon_3}\left[\uppsi^{(a)}_0,\mathtt{e}^{(a+1)}_m\right]-\frac{1}{6}\frac{1}{2\epsilon_3}\left[\uppsi^{(a)}_0,\mathtt{e}^{(a+1)}_m\right]\uppsi^{(a)}_0\\
		=&-\frac{1}{6\epsilon_3}\left(\left[\uppsi^{(a)}_1,\mathtt{e}^{(a+1)}_m\right]+\frac{1}{2}\left[\left(\uppsi^{(a)}_0\right)^2,\mathtt{e}^{(a+1)}_m\right]\right),
	\end{split}
\end{equation}
and likewise for $\mathtt{f}$. Define $\widetilde{\uppsi}^{(a)}_1=\uppsi^{(a)}_1+\frac{1}{2}\left(\uppsi^{(a)}_0\right)^2$ (notice the difference with the namesake for the cases in the main context). As a result, all the modes can be obtained from $\uppsi^{(a)}_0$, $\uppsi^{(a)}_1$, $\mathtt{e}^{(a)}_0$ and $\mathtt{f}^{(a)}_0$ ($a\in Q_0$) inductively via
\begin{equation}
	\mathtt{e}^{(a)}_{m+1}=\frac{1}{6\left(\epsilon_1+\epsilon_2\right)}\left[\widetilde{\uppsi}^{(a+1)}_1,\mathtt{e}^{(a)}_m\right],\quad\mathtt{f}^{(a)}_{m+1}=-\frac{1}{6\left(\epsilon_1+\epsilon_2\right)}\left[\widetilde{\uppsi}^{(a+1)}_1,\mathtt{f}^{(a)}_m\right],\quad\uppsi^{(a)}_{m+1}=\left[\mathtt{e}^{(a)}_{m+1},\mathtt{f}^{(a)}_0\right].
\end{equation}

\paragraph{Case 2: conifold} We have
\begin{equation}
	\sigma_1^{ab}=\sigma_1^{ba}=0,\quad\sigma_2^{01}=-\epsilon_1^2,\quad\sigma_2^{10}=-\epsilon_2^2,\quad\sigma_2^{aa}=0.
\end{equation}
The defining relations of the quiver Yangian are
\begin{align}
	&\left[\uppsi^{(a)}_n,\uppsi^{(b)}_m\right]=0,\\
	&\left[\mathtt{e}^{(a)}_n,\mathtt{f}^{(b)}_m\right]=\delta_{ab}\uppsi^{(a)}_{m+n},\\
	&\left[\uppsi^{(a)}_{n+2},\mathtt{e}^{(a+1)}_m\right]-2\left[\uppsi^{(a)}_{n+1},\mathtt{e}^{(a+1)}_{m+1}\right]+\left[\uppsi^{(a)}_n,\mathtt{e}^{(a+1)}_{m+2}\right]=-\sigma_2^{ba}\uppsi^{(a)}_n\mathtt{e}^{(a+1)}_m+\sigma_2^{ab}\mathtt{e}^{(a+1)}_m\uppsi^{(a)}_n,\label{psiaea+1conifold}\\
	&\left[\uppsi^{(a)}_{n+2},\mathtt{f}^{(a+1)}_m\right]-2\left[\uppsi^{(a)}_{n+1},\mathtt{f}^{(a+1)}_{m+1}\right]+\left[\uppsi^{(a)}_n,\mathtt{f}^{(a+1)}_{m+2}\right]=-\sigma_2^{ab}\uppsi^{(a)}_n\mathtt{f}^{(a+1)}_m+\sigma_2^{ba}\mathtt{f}^{(a+1)}_m\uppsi^{(a)}_n,\\
	&\left[\mathtt{e}^{a}_{n+2},\mathtt{e}^{(a+1)}_m\right]-2\left[\mathtt{e}^{(a)}_{n+1},\mathtt{e}^{(a+1)}_{m+1}\right]+\left[\mathtt{e}^{(a)}_n,\mathtt{e}^{(a+1)}_{m+2}\right]=-\sigma_2^{ba}\mathtt{e}^{(a)}_n\mathtt{e}^{(a+1)}_m-\sigma_2^{ab}\mathtt{e}^{(a+1)}_m\mathtt{e}^{(a)}_n,\\
	&\left[\mathtt{f}^{a}_{n+2},\mathtt{f}^{(a+1)}_m\right]-2\left[\mathtt{f}^{(a)}_{n+1},\mathtt{f}^{(a+1)}_{m+1}\right]+\left[\mathtt{f}^{(a)}_n,\mathtt{f}^{(a+1)}_{m+2}\right]=-\sigma_2^{ab}\mathtt{f}^{(a)}_n\mathtt{f}^{(a+1)}_m-\sigma_2^{ba}\mathtt{f}^{(a+1)}_m\mathtt{f}^{(a)}_n,\\
	&\text{Sym}_{n_1,n_2}\text{Sym}_{m_1,m_2}\left[\mathtt{e}^{(a)}_{n_1},\left[\mathtt{e}^{(a+1)}_{m_1},\left[\mathtt{e}^{(a)}_{n_2},\mathtt{e}^{(a+1)}_{m_2}\right]\right]\right]=\text{Sym}_{n_1,n_2}\text{Sym}_{m_1,m_2}\left[\mathtt{f}^{(a)}_{n_1},\left[\mathtt{f}^{(a+1)}_{m_1},\left[\mathtt{f}^{(a)}_{n_2},\mathtt{f}^{(a+1)}_{m_2}\right]\right]\right]=0.
\end{align}
Recall that we simply use $[\text{-},\text{-}]$ to denote the supercommutator in our convention here. Take $n=-2$ in \eqref{psiaea+1conifold}, and we have
\begin{equation}
	\left[\uppsi^{(0)}_0,\mathtt{e}^{(1)}_m\right]=0.
\end{equation}
Then take $n=-1,0$ in \eqref{psiaea+1conifold}, and we get
\begin{align}
	&\left[\uppsi^{(0)}_1,\mathtt{e}^{(1)}_m\right]=\left(\epsilon_2^2-\epsilon_1^2\right)\mathtt{e}^{(1)}_m,\\
	&\left[\uppsi^{(0)}_2,\mathtt{e}^{(1)}_m\right]-2\left[\uppsi^{(0)}_1,\mathtt{e}^{(1)}_{m+1}\right]=\left(\epsilon_2^2-\epsilon_1^2\right)\uppsi^{(0)}_0\mathtt{e}^{(1)}_m.
\end{align}
Therefore,
\begin{equation}
	\begin{split}
		\left[\uppsi^{(0)}_2,\mathtt{e}^{(1)}_m\right]-2\left(\epsilon_2^2-\epsilon_1^2\right)\mathtt{e}^{(1)}_{m+1}=&\uppsi^{(0)}_0\left[\uppsi^{(0)}_1,\mathtt{e}^{(1)}_m\right]\\
		=&\left[\uppsi^{(0)}_0,\mathtt{e}^{(1)}_m\right]\uppsi^{(0)}_1+\uppsi^{(0)}_0\left[\uppsi^{(0)}_1,\mathtt{e}^{(1)}_m\right]\\
		=&\left[\uppsi^{(0)}_0\uppsi^{(0)}_1,\mathtt{e}^{(1)}_m\right],
	\end{split}
\end{equation}
and likewise for $a=1$, as well as $\mathtt{f}$. Define $\widetilde{\uppsi}^{(a)}_2=\uppsi^{(a)}_2-\uppsi^{(a)}_0\uppsi^{(a)}_1$. As a result\footnote{As $\uppsi^{(a)}_0$ commutes with $\mathtt{e}^{(a+1)}_m$ and $\mathtt{f}^{(a+1)}_m$, it could be possible to add terms such as $\frac{1}{3}\left(\uppsi^{(a)}_0\right)^3$ to the definition of $\widetilde{\uppsi}^{(a)}_2$, which might be more convenient when finding a minimalistic presentation.}, all the modes can be obtained from $\uppsi^{(a)}_0$, $\uppsi^{(a)}_1$, $\uppsi^{(a)}_2$, $\mathtt{e}^{(a)}_0$ and $\mathtt{f}^{(a)}_0$ ($a\in Q_0$) inductively via
\begin{equation}
	\begin{split}
		&\mathtt{e}^{(1)}_{m+1}=\frac{1}{2\left(\epsilon_2^2-\epsilon_1^2\right)}\left[\widetilde{\uppsi}^{(0)}_2,\mathtt{e}^{(1)}_m\right],\quad\mathtt{f}^{(1)}_{m+1}=-\frac{1}{2\left(\epsilon_2^2-\epsilon_1^2\right)}\left[\widetilde{\uppsi}^{(0)}_2,\mathtt{f}^{(1)}_m\right],\quad\uppsi^{(1)}_{m+1}=\left[\mathtt{e}^{(1)}_{m+1},\mathtt{f}^{(1)}_{0}\right],\\
		&\mathtt{e}^{(0)}_{m+1}=\frac{1}{2\left(\epsilon_1^2-\epsilon_2^2\right)}\left[\widetilde{\uppsi}^{(1)}_2,\mathtt{e}^{(0)}_m\right],\quad\mathtt{f}^{(0)}_{m+1}=-\frac{1}{2\left(\epsilon_1^2-\epsilon_2^2\right)}\left[\widetilde{\uppsi}^{(1)}_2,\mathtt{f}^{(0)}_m\right],\quad\uppsi^{(0)}_{m+1}=\left[\mathtt{e}^{(0)}_{m+1},\mathtt{f}^{(0)}_{0}\right].
	\end{split}
\end{equation}
Unlike all the cases discussed above, we further need the extra $\uppsi^{(a)}_2$ for the conifold case.

\paragraph{Comments on $\mathbb{C}^3/(\mathbb{Z}_2\times\mathbb{Z}_2)$} Let us also briefly mention the quiver Yangian for $\mathbb{C}^3/(\mathbb{Z}_2\times\mathbb{Z}_2)$ here, which is the only toric CY$_3$ without compact divisors that is not a generalized conifold. Its toric diagram and quiver are depicted in Figure \ref{C3Z2Z2}.
\begin{figure}[h]
	\centering
	\includegraphics[width=8cm]{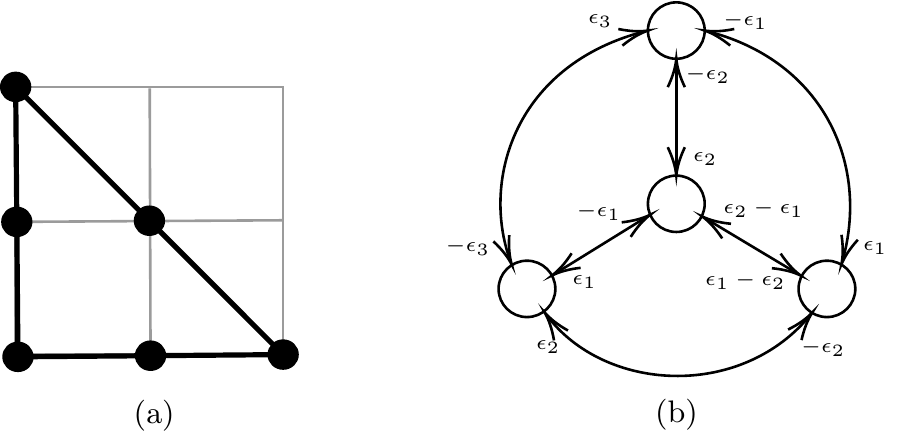}
	\caption{(a) The toric diagram of $\mathbb{C}^3/(\mathbb{Z}_2\times\mathbb{Z}_2)$. (b) The corresponding quiver with charge assignment.}\label{C3Z2Z2}
\end{figure}
The generators $\psi^{(a)}_n$, $e^{(a)}_n$ and $f^{(a)}_n$ of the quiver Yangian (whose underlying algebra is the affine $D(2,1;\alpha)$) satisfy the same defining relations as in Definition \ref{Ydef}. One can equivalently write the quiver Yangian using $\uppsi^{(a)}_n$, $\mathtt{e}^{(a)}_n$ and $\mathtt{f}^{(a)}_n$. Notice that the only vanishing $\sigma_1^{(ab)}$ in this case would have $a=b$. It is straightforward to see that the quiver Yangian has a minimalistic presentation as given in Theorem \ref{minthm}.

\section{Odd Reflections and Chevalley Generators}\label{oddrefchevgen}
Let us verify that the primed Chevalley generators in \eqref{chevgenprimed} satisfy the corresponding relations of the Kac-Moody superalgebras, which are essentially the ones involving only zero modes in Theorem \ref{minthm}. For convenience, \eqref{chevgenprimed} is reproduced here:
\begin{equation}
	\begin{split}
		&\psi'^{(a)}_0=\sum_{b=1}^{M+N}R_{ab}\psi^{(b)}_0=\begin{cases}
			-\psi^{(a)}_0,&\quad a=\digamma,\\
			\psi^{(a)}_0+\psi^{(\digamma)}_0,&\quad a=\digamma\pm1,\\
			\psi^{(a)}_0,&\quad\text{otherwise};
		\end{cases}\\
		&e'^{(a)}_0=\begin{cases}
			f^{(a)}_0,&\quad a=\digamma,\\
			\left[e^{(\digamma)}_0,e^{(a)}_0\right],&\quad a=\digamma\pm1,\\
			e^{(a)}_0,&\quad\text{otherwise};
		\end{cases}
		\qquad
		f'^{(a)}_0=\begin{cases}
			-e^{(a)}_0,&\quad a=\digamma,\\
			-\frac{1}{A_{a\digamma}}\left[f^{(a)}_0,f^{(\digamma)}_0\right],&\quad a=\digamma\pm1,\\
			f^{(a)}_0,&\quad\text{otherwise}.
		\end{cases}
	\end{split}
\end{equation}
We also recall that
\begin{equation}
	A'_{ab}=\begin{cases}
		-A_{ab},&\quad (a,b)=(\digamma\pm1,\digamma),(\digamma,\digamma\pm1),\\
		A_{aa}+2A_{a\digamma},&\quad a=b=\digamma\pm1,\\
		A_{ab},&\quad\text{otherwise}.
	\end{cases}
\end{equation}

The $\psi'\psi'$ relation, the $e'^{(a)}e'^{(b)}$, $f'^{(a)}f'^{(b)}$ relations for $\sigma_1^{ab}=0$ and the Serre relations whose right hand sides are all zero follow from direct computations. Now, let us check the $e'f'$ relation:
\begin{itemize}
	\item $a=b=\digamma$: We have
	\begin{equation}
		\left[e'^{(a)}_0,f'^{(b)}_0\right]=-\left[f^{(\digamma)}_0,e^{(\digamma)}_0\right]=-\psi^{(\digamma)}_0=\psi'^{(a)}_0.
	\end{equation}
    
    \item $a=b=\digamma\pm1$: We have
    \begin{equation}
    	\begin{split}
    		\left[e'^{(a)}_0,f'^{(b)}_0\right]=&-\frac{1}{A_{a\digamma}}\left[\left[e^{(\digamma)}_0,e^{(a)}_0\right],\left[f^{(a)}_0,f^{(\digamma)}_0\right]\right]\\
    		=&-\frac{1}{A_{a\digamma}}\left(\left[e^{(\digamma)}_0,\left[e^{(a)}_0,\left[f^{(a)}_0,f^{(\digamma)}_0\right]\right]\right]-(-1)^{|a|}\left[e^{(a)}_0,\left[e^{(\digamma)}_0,\left[f^{(a)}_0,f^{(\digamma)}_0\right]\right]\right]\right)\\
    		=&-\frac{1}{A_{a\digamma}}\left(\left[e^{(\digamma)}_0,-A_{a\digamma}f^{(\digamma)}_0\right]+\left[e^{(a)}_0,-A_{a\digamma}f^{(a)}_0\right]\right)\\
    		=&\psi'^{(\digamma)}_0+\psi'^{(a)}_0,
    	\end{split}
    \end{equation}
    where we have used the Jacobi identity in the second and third equalities.
    
    \item $a=\digamma,b=\digamma\pm1$: We have $\left[e'^{(a)}_0,f'^{(b)}_0\right]=-\frac{1}{A_{b\digamma}}\left[f^{(\digamma)}_0,\left[f^{(b)}_0,f^{(\digamma)}_0\right]\right]=0$ as $\text{ad}_{f^{(\digamma)}_0}^2=0$.
    
    \item $a=\digamma\mp1,b=\digamma\pm1$: This is similar to the case when $a=b=\digamma\pm1$. By using the Jacobi identity, we have $\left[e'^{(a)}_0,f'^{(b)}_0\right]=-\frac{1}{A_{b\digamma}}\left(0+\left[e^{(a)}_0,-A_{a\digamma}f^{(b)}_0\right]\right)=0$.
    
    \item $a\neq\digamma,\digamma\pm1, b=\digamma\pm1$: We have $\left[e'^{(a)}_0,f'^{(b)}_0\right]=-\frac{1}{A_{b\digamma}}\left[e^{(a)}_0,\left[f^{(b)}_0,f^{(\digamma)}_0\right]\right]=0$ using the Jacobi identity.
    
    \item otherwise: The remaining cases are immediate following the expressions of the primed generators and the similar arguments as above.
\end{itemize}

Next, we shall check the $\psi'e'$ relation:
\begin{itemize}
	\item $a=b=\digamma\pm1$: We have
	\begin{equation}
		\begin{split}
			\left[\psi'^{(a)}_0,e'^{(b)}_0\right]=&\left[\psi^{(a)}+\psi^{(\digamma)}_0,\left[e^{(\digamma)}_0,e^{(a)}_0\right]\right]\\
			=&-\left[e^{(\digamma)}_0,\left[e^{(a)}_0,\psi^{(a)}_0\right]\right]+(-1)^{|a|}\left[e^{(a)}_0,\left[e^{\digamma}_0,\psi^{(a)}_0\right]\right]-\left[e^{(\digamma)}_0,\left[e^{(a)}_0,\psi^{(\digamma)}_0\right]\right]-0\\
			=&A_{aa}\left[e^{(\digamma)}_0,e^{(a)}_0\right]-(-1)^{|a|}A_{a\digamma}\left[e^{(a)}_0,e^{\digamma}_0\right]+A_{a\digamma}\left[e^{(\digamma)}_0,e^{(a)}_0\right]\\
			=&(2A_{a\digamma}+A_{aa})\left[e^{(\digamma)}_0,e^{(a)}_0\right]\\
			=&A'_{ab}e'^{(b)}_0,
		\end{split}
	\end{equation}
    where we have used the Jacobi identity in the second eqaulity.
    
    \item $a=\digamma,b=\digamma\pm1$: We have
    \begin{equation}
    	\left[\psi'^{(a)}_0,e'^{(b)}_0\right]=-\left[\psi^{\digamma},\left[e^{(\digamma)}_0,e^{(b)}_0\right]\right]=-\left[e^{(\digamma)}_0,A_{b\digamma}e^{(b)}_0\right]=-A_{b\digamma}e'^{(b)}_0=A'_{ab}e'^{(b)}_0.
    \end{equation}
    
    \item $a=\digamma\mp1,b=\digamma\pm1$: This is similar to the case when $a=b=\digamma\pm1$. By using the Jacobi identity, we have $\left[\psi'^{(a)}_0,e'^{(b)}_0\right]=-(-1)^{|a|}A_{a\digamma}\left[e^{(b)}_0,e^{(\digamma)}_0\right]+A_{b\digamma}\left[e^{(k)}_0,e^{(b)}_0\right]=0$.
    
    \item $a\neq\digamma,\digamma\pm1,b=\digamma\pm1$: We have $\left[\psi'^{(a)}_0,e'^{(b)}_0\right]=\left[\psi^{(a)}_0,\left[e^{(\digamma)}_0,e^{(b)}_0\right]\right]=0$ using the Jacobi identity.
    
    \item $a=b=\digamma$: We have $\left[\psi'^{(a)}_0,e'^{(b)}_0\right]=-\left[\psi^{(\digamma)}_0,f^{(\digamma)}_0\right]=0$.
    
    \item $a=\digamma\pm1,b=\digamma$: We have $\left[\psi'^{(a)}_0,e'^{(b)}_0\right]=\left[\psi^{(a)}_0+\psi^{(\digamma)}_0,f^{(\digamma)}_0\right]=-A_{a\digamma}f^{(\digamma)}_0=A'_{ab}e^{(b)}$.
    
    \item otherwise: The remaining cases are immediate following the expressions of the primed generators and the similar arguments as above.
\end{itemize}

The $\psi'f'$ relation can be verified in the same manner. Nevertheless, let us explicitly write the proof for three of the cases here:
\begin{itemize}
	\item $a=\digamma\pm1,b=\digamma$: We have
	\begin{equation}
		\left[\psi'^{(a)}_0,f'^{(b)}_0\right]=\left[\psi^{(a)}+\psi^{(\digamma)}_0,-e^{(\digamma)}_0\right]=-A_{a\digamma}e^{(\digamma)}_0=A_{ab}f'^{(b)}_0=-A'_{ab}f'^{(b)}_0.
	\end{equation}
	
	\item $a=\digamma,b=\digamma\pm1$: We have
	\begin{equation}
		\left[\psi'^{(a)}_0,f'^{(b)}_0\right]=\frac{1}{A_{b\digamma}}\left[-\psi^{(\digamma)}_0,-\left[f^{(b)}_0,f^{(\digamma)}_0\right]\right]=\frac{1}{A_{b\digamma}}\left[-A_{b\digamma}f^{(b)}_0,f^{(\digamma)}_0\right]=A_{b\digamma}f'^{(b)}_0=-A_{ab}f'^{(b)}_0.
	\end{equation}

    \item $a=b=\digamma\pm1$: We have
    \begin{equation}
    	\begin{split}
    		\left[\psi'^{(a)}_0,f'^{(b)}_0\right]=&-\frac{1}{A_{a\digamma}}\left[\psi^{(a)}_0+\psi^{(\digamma)}_0,\left[f^{(a)}_0,f^{(\digamma)}_0\right]\right]\\
    		=&\frac{1}{A_{a\digamma}}\left(\left[\left[f^{(a)}_0,f^{(\digamma)}_0\right],\psi^{(a)}_0\right]+\left[\left[f^{(a)}_0,f^{(\digamma)}_0\right],\psi^{(\digamma)}_0\right]\right)\\
    		=&\frac{1}{A_{a\digamma}}\left(\left[f^{(a)}_0,\left[f^{\digamma}_0,\psi^{(a)}_0\right]\right]-(-1)^{|a|}\left[f^{(\digamma)}_0,\left[f^{a}_0,\psi^{(a)}_0\right]\right]+0-(-1)^{|a|}\left[f^{(\digamma)}_0,\left[f^{a}_0,\psi^{(\digamma)}_0\right]\right]\right)\\
    		=&\frac{1}{A_{a\digamma}}\left(A_{a\digamma}\left[f^{(a)}_0,f^{(\digamma)}_0\right]-(-1)^{|a|}A_{aa}\left[f^{(\digamma)}_0,f^{(a)}_0\right]-(-1)^{|a|}A_{a\digamma}\left[f^{(\digamma)}_0,f^{(a)}_0\right]\right)\\
    		=&-(2A_{a\digamma}+A_{aa})\times\left(-\frac{1}{A_{a\digamma}}\left[f^{(a)}_0,f^{(\digamma)}_0\right]\right)\\
    		=&-A'_{ab}f'^{(b)}_0.
    	\end{split}
    \end{equation}
\end{itemize}

\section{Rectangular $\mathcal{W}$-Algebras}\label{recW}
In literature, the rectangular $\mathcal{W}$-algebra (of type A) is often defined based on the distinguished case for the Kac-Moody superalgebra where the number of fermionic nodes is minimized. Nevertheless, we can certainly consider $\mathcal{W}$-algebras with any underlying root systems/Dynkin diagrams as we are going to discuss now. Proposition \ref{Wiso} then ensures that they are isomorphic for a given generalized conifold.

In this appendix only, we will use $\mathfrak{g}$ to denote the algebra $\mathfrak{gl}(Ml|Nl)$ for some positive integer $l$. We shall choose the convention such that the basis matrix $E_{ij}$ has entry $(-1)^{p(j)}$ at position $(i,j)$ with all other elements being zero\footnote{Another convention often adopted in literature such as \cite{ueda2022affine} would naturally be 1 at entry $(i,j)$.}. Notice that we have used
\begin{equation}
	p(i)=\begin{cases}
		0,&i\text{ is bosonic}\\
		1,&i\text{ is fermionic}
	\end{cases}\label{pdef}
\end{equation}
so as to distinguish it from $|a|$ in the quiver Yangians. Given a parity sequence $\varsigma$ composed of $(-1)^{p(i)}$, the $\mathbb{Z}_2$-grading of $E_{ij}$ is $p(i)+p(j)$. In particular, $E_{i_1j_1}E_{i_2j_2}=(-1)^{p(j_1)}\delta_{j_1i_2}E_{i_1j_2}$ and $\text{str}(E_{i_1j_1}E_{i_2j_2})=(-1)^{p(j_1)}\delta_{j_1i_2}\delta_{i_1j_2}$. Then
\begin{equation}
	\mathfrak{g}=\bigoplus_{\substack{1\leq i,j\leq M+N\\1\leq r,s\leq l}}\mathbb{C}E_{(r-1)(M+N)+i,(s-1)(M+N)+j},
\end{equation}
and $E_{(r-1)(M+N)+i,(s-1)(M+N)+j}=E_{ij}\otimes E_{rs}$ as $\mathfrak{g}$ is isomorphic to $\mathfrak{gl}(M|N)\otimes\mathfrak{gl}(l)$ as a vector space\footnote{Notice that the $\mathfrak{gl}(l)$ part (with subscripts $r,s$) is always bosonic.}. We shall take the bosonic nilpotent matrix $x^-=\sum\limits_{s=1}^{l-1}\sum\limits_{i=1}^{M+N}E_{s(M+N)+i,(s-1)(M+N)+i}$ which can be written as $\left(\sum\limits_{i=1}^{M+N}E_{ii}\right)\otimes\left(\sum\limits_{s=1}^{l-1}E_{s,s-1}\right)$. This nilpotent matrix is of Jordan type with the rectangle Young tableau $\left(l^{(M|N)}\right)$ (and hence the name rectangular $\mathcal{W}$-algebra). Given a complex number $k$, there is an inner product of $\mathfrak{g}$ given by
\begin{equation}
	(u|v)=\begin{cases}
		k\text{str}(uv),&u\in\mathfrak{sl}(Ml|Nl)\text{ or }v\in\mathfrak{sl}(Ml|Nl)\\
		k\text{str}(uv)+(-1)^{p(i)+p(j)}(1-c),&u=E_{ii}\otimes E_{rr}\text{ and }v=E_{jj}\otimes E_{ss}
	\end{cases}
\end{equation}
for some $c\in\mathbb{C}$.

Now, $\mathfrak{g}$ has a good grading in the sense of \cite{kac2003quantum} for the nilpotent element with
\begin{equation}
	\mathfrak{g}_r:=\bigoplus_{\substack{1\leq i,j\leq M+N\\0\leq s\leq l-1\\0\leq s+r\leq l-1}}\mathbb{C}E_{s(M+N)+i,(s+r)(M+N)+j}.
\end{equation}
We then have an $\mathfrak{sl}_2$ triple $(h,x^+,x^-)$ such that $\mathfrak{g}_r=\{y\in\mathfrak{g}|[h,y]=ry\}$. Define the subalgebras $\mathfrak{b}=\mathfrak{g}_{\leq0}=\bigoplus\limits_{r\leq0}\mathfrak{g}_r$ and $\mathfrak{g}_0=\mathfrak{g}_{r=0}$. We have an inner product on $\mathfrak{b}$ which reads
\begin{equation}
	\kappa(u,v)=(u|v)+\frac{1}{2}(\kappa_{\mathfrak{g}}(u,v)-\kappa_{\mathfrak{g}_0}(\text{pr}(u),\text{pr}(v)))
\end{equation}
for any $u,v\in\mathfrak{b}$, where $\kappa_{\mathfrak{g}}$ (resp. $\kappa_{\mathfrak{g}_0}$) is the Killing form on $\mathfrak{g}$ (resp. $\mathfrak{g}_0$) and $\text{pr}:\mathfrak{b}\rightarrow\mathfrak{g}_0$ is the projection map. Recall that in general, the Killing form is $\kappa_{\mathfrak{gl}(M|N)}(x,y)=2(M-N)\text{str}(xy)-2\text{str}(x)\text{str}(y)$. Then
\begin{equation}
	\begin{split}
		&\kappa(E_{r_1(M+N)+i_1,s_1(M+N)+j_1},E_{r_2(M+N)+i_2,s_2(M+N)+j_2})\\
		=&\delta_{r_1,s_2}\delta_{r_2,s_1}\delta_{i_1,j_2}\delta_{i_2,j_1}(-1)^{p(j)}\varkappa+\delta_{r_1,s_1}\delta_{r_2,s_2}\delta_{i_1,j_1}\delta_{i_2,j_2}(\delta_{r_1,r_2}-c),
	\end{split}
\end{equation}
where $\varkappa:=k+(l-1)(M-N)$. Consider the affinization $\hat{\mathfrak{b}}=\mathfrak{b}\left[t^{\pm1}\right]\oplus\mathbb{C}\bm{1}$ (with $\bm{1}$ central). The commutation relation reads $[at^m,bt^n]=[a,b]t^{m+n}+\delta_{m,-n}m\kappa(a,b)\bm{1}$. The associated (universal affine) vertex algebra $V^k(\mathfrak{b})$ is defined to be $U(\hat{\mathfrak{b}})/U(\hat{\mathfrak{b}})(\mathfrak{b}[t]\oplus\mathbb{C}(\bm{1}-1))\cong U(\hat{\mathfrak{b}})\otimes_{U(\mathfrak{b}[t]\oplus\mathbb{C}\bm{1})}\mathbb{C}$. Here, $\mathbb{C}$ denotes the one-dimensional representation of $\mathfrak{b}[t]\oplus\mathbb{C}{\bm{1}}$ where $\mathfrak{b}[t]$ acts trivially as 0 and $\bm{1}$ acts as 1. This is isomorphic to $U\left(\mathfrak{b}\left[t^{-1}\right]t^{-1}\right)$ as a vector space by PBW theorem. Here, we shall take the mode expansion of a current $a(z)$ in the vertex algebra depending on its spin $s$:
\begin{equation}
	a(z)=\sum_{n\in\mathbb{Z}}\frac{a[n]}{z^{n+s}},
\end{equation}
where we have denoted $at^n$ as $a[n]$ to avoid potential clutter of subscripts later on. In this paper, we use the normal ordered product with the convention\footnote{One can also define the $n$-product given by
\begin{equation}
	\begin{split}
		(a_{(n)}b)(z)=a(z)_{(n)}b(z)=\begin{cases}
			\text{Res}_w(w-z)^n[a(w),b(z)],&n\geq0\\
			\frac{1}{(n+1)!}:\partial^{n+1}a(z)b(z):,&n<0,
		\end{cases}\nonumber
	\end{split}
\end{equation}
as well as the $\lambda$-bracket $[a_\lambda b]=\sum\limits_{n\in\mathbb{Z}_+}\frac{\lambda^n}{n!}a_{(n)}b$ which enjoys certain properties such as the noncommutative Wick formula. See for example \cite{kac1998vertex} for more details. The pair of fields is local if $(a_{(n)}b)(z)$ vanishes for sufficiently large (positive) $n$. It is clear that the $(-1)$-product coincides with the normal ordered product.}
\begin{equation}
	\begin{split}
		&:a(z)b(z):~=a(z)_{\leq}b(z)+(-1)^{p(a)+p(b)}b(z)a(z)_>,\\
		&\text{ where }a(z)_{\leq}=\sum_{n\leq-s}\frac{a[n]}{z^{n+s}}\text{ and }a(z)_>=\sum_{n>-s}\frac{a[n]}{z^{n+s}}.
	\end{split}
\end{equation}
In terms of modes, we have
\begin{equation}
	:a[n]b[m]:~=\begin{cases}
		a[n]b[m],&n\leq-s\\
		(-1)^{p(a)+p(b)}b[m]a[n],&n>-s.
	\end{cases}
\end{equation}
In the main context and below, we shall use $(ab)$ instead of :$ab$: to denote the normal ordering for convenience when it would not cause confusions. Of course, different conventions of the normal ordered product would not change our results in \S\ref{YandW}. For instance, if we ``split'' the normal ordering at the zero modes, one may check that the homomorphism $\Phi$ from $\mathtt{Y}$ to $\mathcal{W}$ would remain the same.

Let us also consider the Lie superalgebra $\mathfrak{a}=\left(\bigoplus\limits_{u\in\mathfrak{b}}\mathbb{C}\mathcal{A}^{(u)}\right)\oplus\left(\bigoplus\limits_{u\in\mathfrak{g}_{<0}}\mathbb{C}\mathcal{A}_{(u)}\right)$ with $p\left(\mathcal{A}^{(u)}\right)=p(u)$ and $p\left(\mathcal{A}_{(u)}\right)=p(u)+1$. The commutation relations are
\begin{equation}
	\begin{split}
		&\left[\mathcal{A}^{(u)},\mathcal{A}^{(v)}\right]=\mathcal{A}^{([u,v])},\quad\left[\mathcal{A}_{(u)},\mathcal{A}_{(v)}\right]=0,\\
		&\left[\mathcal{A}^{(E_{i_1j_1})},\mathcal{A}_{(E_{i_2j_2})}\right]=\delta_{j_1,i_1}\mathcal{A}_{(E_{i_1,j_2})}-\delta_{i_1,j_2}(-1)^{(p(i_1)+p(j_1))(p(i_2)+p(j_2)+1)}\mathcal{A}_{(E_{i_2,j_1})}.
	\end{split}
\end{equation}
Suppose $u=\sum\limits_{i}a_iu_i$ ($a_i\in\mathbb{C}$), then $\mathcal{A}^{(u)}$ is $\sum\limits_{i}a_i\mathcal{A}^{(u_i)}$ (and similarly for $\mathcal{A}_{(u)}$). We can write the inner product determined by
\begin{equation}
	\kappa_{\mathfrak{a}}\left(\mathcal{A}^{(u)},\mathcal{A}^{(v)}\right)=\kappa_\mathfrak{b}(u,v),\quad\kappa_{\mathfrak{a}}\left(\mathcal{A}^{(u)},\mathcal{A}_{(v)}\right)=\kappa_{\mathfrak{a}}\left(\mathcal{A}_{(u)},\mathcal{A}_{(v)}\right)=0,
\end{equation}
and likewise consider the affinization of $\mathfrak{a}$. Then the associated vertex algebra $V^k(\mathfrak{a})$ contains $V^k(\mathfrak{b})$ as a subalgebra. Both of the vertex algebras can be regarded as non-associative algebras with respect to the normal ordered product.

The $\mathcal{W}$-algebra is then the collection of elements in $V^k(\mathfrak{b})$ that are annihilated by the BRST charge. More specifically, the BRST cohomology has fermionic derivation $Q:V^k(\mathfrak{a})\rightarrow V^k(\mathfrak{a})$ commuting with the translation operator $\partial$ of the vertex algebra. The other commutation relations $Q$ should satisfy can be found for example in \cite[\S3]{arakawa2017explicit} and in \cite[\S3]{ueda2022affine} (with the convention therein).
\begin{definition}
	Given the above data, the rectangular $\mathcal{W}$-algebra is $\mathcal{W}^k\left(\mathfrak{gl}(M|N),\left(l^{(M|N)}\right)\right):=\{v\in V^k(\mathfrak{b})\subset V^k(\mathfrak{a})|Qv=0\}$.
\end{definition}
Notice that we have omitted the parity sequence $\varsigma$ in the notation as different $\varsigma$ give isomorphic $\mathcal{W}$-algberas by Proposition \ref{Wiso}. For our discussions, it would be of great help to obtain the generators of the $\mathcal{W}$-algebra. This can be constructed by considering the non-associative free algebra $T\left(\mathfrak{gl}(l)_{\leq0}\left[t^{-1}\right]t^{-1}\right)\otimes\mathbb{C}[\tau]$ with the even element $\tau$ commuting with $\bm{1}$ and $[\tau,y[m]]=-my[m-1]$ for $y\in\mathfrak{gl}(l)_{\leq0}$. We then have an algebra homomorphism $\mathfrak{T}:T\left(\mathfrak{gl}(l)_{\leq0}\left[t^{-1}\right]t^{-1}\right)\otimes\mathbb{C}[\tau]\rightarrow\mathfrak{gl}(M|N)\otimes V^k(\mathfrak{b})\otimes\mathbb{C}[\tau]$ such that
\begin{equation}
	\mathfrak{T}(x)=\sum_{i,j=1}^{M+N}(-1)^{p(i)p(j)}E_{ij}\otimes\mathfrak{T}_{ij}(x),\quad\mathfrak{T}(\tau)=\tau,
\end{equation}
where
\begin{equation}
	\mathfrak{T}_{ij}(x)=x\otimes E_{ji}\in\mathfrak{gl}(l)_{\leq0}\left[t^{-1}\right]t^{-1}\otimes\mathfrak{gl}(M|N)=\mathfrak{b}\left[t^{-1}\right]t^{-1}.
\end{equation}
Since $\mathfrak{T}(xy)=\mathfrak{T}(x)\mathfrak{T}(y)$, we find that
\begin{equation}
	\mathfrak{T}_{ij}(xy)=\sum_{r=1}^{M+N}\mathfrak{T}_{ir}(x)\mathfrak{T}_{rj}(y).
\end{equation}

Let us now consider the $l\times l$ matrix
\begin{equation}
	B=\begin{pmatrix}
		\varkappa\tau+E_{11}[-1]~&~-1~&~0~&~\dots~&~0\\
		E_{21}[-1]~&~\varkappa\tau+E_{22}[-1]~&~-1~&~\dots~&~0\\
		\vdots~&~\vdots~&~\ddots~&~\vdots~&~\vdots\\
		E_{l-1,1}[-1]~&~E_{l-1,2}[-1]~&~\dots~&~\varkappa\tau+E_{l-1,l-1}[-1]~&~-1\\
		E_{l1}[-1]~&~E_{l2}[-1]~&~\dots~&~E_{l,l-1}[-1]~&~\varkappa\tau+E_{ll}[-1]
	\end{pmatrix}
\end{equation}
and compute its column determinant
\begin{equation}
	\text{cdet}(B)=\sum_{\sigma\in\mathfrak{S}_l}\text{sgn}\sigma~ b_{\sigma(1)1}(b_{\sigma(2)2}(b_{\sigma(3)3}\dots(b_{\sigma(l-1),l-1}b_{\sigma(l),l})\dots)).
\end{equation}
As the entries $b_{rs}$ of $B$ are in $T\left(\mathfrak{gl}(l)_{\leq0}\left[t^{-1}\right]t^{-1}\right)\otimes\mathbb{C}[\tau]$, we can write
\begin{equation}
	\mathfrak{T}_{ij}(\text{cdet}(B))=\sum_{r=0}^l\widetilde{U}^{(r)}_{ji}(\varkappa\tau)^{l-r}.
\end{equation}
We then have the remarkable results from \cite{ueda2022affine,arakawa2017explicit}:
\begin{theorem}
	The rectangular $\mathcal{W}$-algebra $\mathcal{W}^k\left(\mathfrak{gl}(M|N),\left(l^{(M|N)}\right)\right)$ is freely generated by $\widetilde{U}^{(r)}_{ij}$ for $1\leq r\leq l$ and $1\leq i,j\leq M+N$. Moreover, when $M\neq N$, $M+N\geq2$ (and $\varkappa\neq0$), it is generated by $\widetilde{U}^{(1)}_{ij}$ and $\widetilde{U}^{(2)}_{ij}$.
\end{theorem}

Following \cite{arakawa2017explicit,arakawa2017introduction}, the projection $\mathfrak{b}\rightarrow\mathfrak{l}=(\mathfrak{gl}_l)_0\otimes\mathfrak{gl}_{M|N}$ induces an injective algebra homomorphism $\mu:\mathcal{W}^k\left(\mathfrak{gl}(M|N),\left(l^{(M|N)}\right)\right)\rightarrow V^k(\mathfrak{l})$ known as the (quantum) Miura transformation. Under the Miura transformation, we have
\begin{equation}
	\sum_{r=0}^l\mu\left(\widetilde{U}^{(r)}_{ji}\right)(\varkappa\tau)^{l-r}=\mathfrak{T}_{ij}\left((\varkappa\tau+E_{11}[-1])(\varkappa\tau+E_{22}[-1])\dots(\varkappa\tau+E_{ll}[-1])\right).\label{miura}
\end{equation}
Let us write\footnote{Notice that we could have also started with $-E_{ij}$ as our basis matrix from the very beginning.}
\begin{equation}
	\mathcal{J}^s_{ij}=-E_{(s-1)(M+N)+i,(s-1)(M+N)+j}[-1],\quad\partial\mathcal{J}^s_{ij}=-E_{(s-1)(M+N)+i,(s-1)(M+N)+j}[-2].
\end{equation}
This gives the same convention as in \cite{Rapcak:2019wzw}. The generators of the $\mathcal{W}$-algebra can be written as
\begin{equation}
	U^{(1)}_{ij}=\sum_{1\leq s\leq l}\mathcal{J}^s_{ij},\qquad U^{(2)}_{ij}=\varkappa\sum_{1\leq s\leq l}(s-1)\partial\mathcal{J}^s_{ij}+\sum_{\substack{1\leq s_1<s_2\leq l\\1\leq n\leq M+N}}\left(\mathcal{J}^{s_1}_{in}\mathcal{J}^{s_2}_{nj}\right).\label{U1U2}
\end{equation}
By definition of the vertex algebra, the OPE of $\mathcal{J}^s_{ij}$ reads
\begin{equation}
	\begin{split}
		\mathcal{J}^{s_1}_{i_1j_1}(z)\mathcal{J}^{s_2}_{i_2j_2}(w)\sim&\frac{\kappa(E_{(s_1-1)(M+N)+i_1,(s_1-1)(M+N)+j_1},E_{(s_2-1)(M+N)+i_2,(s_1-2)(M+N)+j_2})}{(z-w)^2}\\
		&\frac{[E_{(s_1-1)(M+N)+i_1,(s_1-1)(M+N)+j_1},E_{(s_2-1)(M+N)+i_2,(s_1-2)(M+N)+j_2}][-1](w)}{z-w}\\
		=&\frac{\delta_{s_1s_2}\delta_{j_1i_2}\delta_{i_1j_2}(-1)^{(p(j_1))}\varkappa+\delta_{i_1j_1}\delta_{i_2j_2}(\delta_{s_1s_2}-c)}{(z-w)^2}\\
		&+\frac{(-1)^{p(i_1)p(j_1)+p(i_2)p(j_2)+p(j_1)p(i_2)}\delta_{s_1s_2}\delta_{i_1j_2}\mathcal{J}^{s_1}_{i_2j_1}-(-1)^{p(j_1)}\delta_{s_1s_2}\delta_{i_2j_1}\mathcal{J}^{s_1}_{i_1j_2}}{z-w}.
	\end{split}
\end{equation}
The OPEs for $U^{(r)}_{ij}$ can then be obtained from this, as well as the commutation relations for their modes via
\begin{equation}
	\left[U^{(r)}_{i_1j_1}[m],U^{(s)}_{i_2j_2}[n]\right]=\frac{1}{(2\pi i)^2}\oint_0\text{d}w\oint_w\text{d}z~z^{m+r-1}w^{n+s-1}U^{(r)}_{i_1j_1}(z)U^{(s)}_{i_2j_2}(w).\label{OPE2commreln}
\end{equation}
In this paper, we shall focus on the case when the parameter $c=0$. The commutation relations used in this paper are listed in Lemma \ref{UUlemma}.

To relate the non-associative $\mathcal{W}$-algebra with the quiver Yangian, we shall consider the universal enveloping algebra $U(\mathcal{W})$. In general, for any vertex algebra $V$, its universal enveloping algebra $U(V)$ is an associative algebra topologically generated by $ut^m$ (or $ut^{m+s-1}$ depending on the convention) for $u\in V$ and $m\in\mathbb{Z}$ which correspond to the modes $u[m]$ in the vertex algebra. Therefore, we shall slightly abuse the notation and write $u[m]$ as well for the elements in $U(V)$. For more details on vertex algebras and their universal enveloping algebras, see for example \cite{frenkel2004vertex}.

\linespread{0.9}\selectfont
\addcontentsline{toc}{section}{References}
\bibliographystyle{utphys}
\bibliography{references}

\end{document}